\documentclass[11pt]{article}
\usepackage{stmaryrd}
\usepackage{amsfonts}
\usepackage{bbm}
\usepackage{float}
\usepackage{amscd}
\usepackage{mathrsfs}
\usepackage{latexsym,amssymb,amsmath,amscd,amscd,amsthm,amsxtra,xypic}
\usepackage[dvips]{graphicx}
\usepackage[utf8]{inputenc}
\usepackage[T1]{fontenc}

\usepackage{lmodern}
\usepackage{amssymb}
\usepackage[all]{xy}
\usepackage{nicefrac,mathtools,enumitem}
\usepackage{microtype}
\usepackage{lipsum}
\usepackage{mathtools}
\usepackage{xfrac}
\usepackage[colorinlistoftodos]{todonotes}
\usepackage{soul}
\usepackage{xcolor}
\usepackage{tikz}
\usepackage{xcolor}
\usepackage{geometry}
\usepackage[hidelinks]{hyperref}

\def\dn{\mathrm{dn}}
\def\res{\mathrm{res}}
\def\sech{\mathrm{sech}}

\usepackage{times}
\usepackage{amsthm}
\usepackage{amsmath}
\usepackage{amsfonts}
\usepackage{amssymb}
\usepackage{mathdots}
\usepackage{color}
\usepackage{authblk}
\usepackage{mathrsfs}
\usepackage{stmaryrd}
\SetSymbolFont{stmry}{bold}{U}{stmry}{m}{n}

\usepackage{bm}

\usepackage{tikz-cd}
\usepackage{tikz}
\usepackage{graphicx}
\usepackage[percent]{overpic}
\usepackage{subfigure}
\usepackage[toc,page]{appendix}

\definecolor{dgreen}{RGB}{166,33,127}

\newcommand{\be }{\begin{equation}}
	\newcommand{\ee }{\end{equation}}

\newcommand{\R}{\mathbb R}
\newcommand{\C}{\mathbb C}

\newcommand{\br}[1]{   [ \cdot,    \cdot  ]   }

\newtheorem{assumption}{Assumption}

\renewcommand {\Im}{\operatorname{Im}}

\usepackage{thmtools}

\DeclareMathOperator{\re}{Re}

\declaretheoremstyle[spaceabove=0.23cm,spacebelow=0.23cm,notefont=\normalfont\bfseries, notebraces={(}{)}]{theorem}
\declaretheoremstyle[spaceabove=0.23cm,spacebelow=0.23cm,bodyfont=\normalfont,notefont=\normalfont\bfseries, notebraces={(}{)}]{noital}
\declaretheoremstyle[spaceabove=0.23cm,spacebelow=0.23cm,bodyfont=\normalfont\color{darkgreen},notefont=\normalfont\bfseries, notebraces={(}{)}]{green}
\declaretheoremstyle[spaceabove=0.23cm,spacebelow=0.23cm,bodyfont=\normalfont,notefont=\normalfont\bfseries,qed=$\qedsymbol$,notebraces={(}{)}]{proofstyle}

\declaretheorem[name=Theorem,numberwithin=section,style=theorem]{thm}
\declaretheorem[name=Conjecture,sibling=thm,style=theorem]{conj}

\declaretheorem[name=Lemma,sibling=thm,style=theorem]{lem}

\declaretheorem[name=Remark,sibling=thm,style=theorem]{rmk}

\definecolor{lightblue}{rgb}{0.68, 0.85, 0.9}
\definecolor{lightred}{rgb}{1.0, 0.8, 0.8}
\definecolor{lightgreen}{rgb}{0.8, 1.0, 0.8}
\definecolor{darkgreen}{rgb}{0.0, 0.5, 0.0}
\definecolor{color12lines}{rgb}{0.5, 0.0, 0.0}
\definecolor{color13lines}{rgb}{0.0, 0.5, 0.0}
\definecolor{color23lines}{rgb}{0.0, 0.0, 0.6}
\definecolor{darkgreen}{rgb}{0.0, 0.39, 0.0}
\definecolor{lightgray}{rgb}{0.75, 0.75, 0.75}

\tikzset{
	branchpoint/.style={orange, thick, mark=x, mark options={orange, line width=1.25pt}}
}

\tikzset{
	singularpoint/.style={blue, mark=*, mark options={blue, mark size=1.25pt}}
}

\tikzset{
	stokeslabel/.style={gray!95,font=\tiny}
}

\tikzset{
	cutlabel/.style={orange,font=\tiny}
}

\tikzset{
	sheetlabel/.style={font=\small}
}

\tikzset{
	branchcut/.style={orange,dashed,semithick}
}

\tikzset{
	disccolor/.style={gray!12}
}

\tikzset{
	wall/.style={black,thick}
}

\tikzset{
	path/.style={thick,rounded corners}
}

\tikzset{
	witharrow/.style={
		decoration={markings, mark=at position #1 with {\arrow[sloped]{Latex}}},
		postaction={decorate}
	}
}

\tikzset{
	antistokesmark/.style={semithick, black, radius=1.5pt, fill=yellow}
}

\tikzset{
	withbackgroundrectangle/.style={show background rectangle, background rectangle/.style={fill=gray!7}}
}

\numberwithin{equation}{subsection}

\title{\bf {Genus two KdV soliton gases and their long-time asymptotics}}
\author[1]{Deng-Shan Wang}
\author[1]{Dinghao Zhu}
\author[1,2]{Xiaodong Zhu\thanks{Email: \texttt{xdzbnu@mail.bnu.edu.cn}}}

\affil[1]{Laboratory of Mathematics and Complex Systems (Ministry of Education), School of Mathematical Sciences, Beijing Normal University, Beijing 100875, China}
\affil[2]{{SISSA, via Bonomea 265, 34136 Trieste, Italy, INFN Sezione di Trieste}}
\date{}

\begin{document}

	\maketitle
\vspace{-2em}
	\begin{abstract}
		This paper employs the Riemann-Hilbert problem and nonlinear steepest
		descent method of Deift-Zhou to provide a comprehensive analysis of the asymptotic behavior of the genus two Korteweg-de Vries soliton gases. It is demonstrated that the genus two soliton gas is related to the two-phase Riemann-Theta function as \(x \to +\infty\), and approaches to zero as \(x \to -\infty\). Additionally, the long-time asymptotic behavior of this genus two soliton gas can be categorized into five distinct regions in the \(x\)-\(t\) plane, which from left to right are {quiescent region}, modulated one-phase wave, unmodulated one-phase wave, modulated two-phase wave, and unmodulated two-phase wave. Moreover, an innovative method is introduced to solve the model problem associated with the high-genus Riemann surface, leading to the determination of the leading terms, which is also related with the multi-phase Riemann-Theta function. A general discussion on the case of arbitrary genus $N$ soliton gas is also presented.\\
		
		\par
		
		
		{\bf AMS subjectclassi cations:} 35Q15,35Q51,35Q53
	\end{abstract}
	
		\tableofcontents
	\section{Introduction}
	It is well known that the Korteweg-de Vries (KdV) equation
	\begin{equation}\label{KdV}
		u_t-6uu_x+u_{xxx}=0
	\end{equation}
	can be presented as the compatibility condition of the Lax pair \cite{Lax1968}
	\begin{equation}\label{Lax pair}
		\begin{aligned}
			&(-\partial_{xx}+u)\varphi(x,t)=E\varphi(x,t),\\
			&\varphi_t(x,t)=(4\partial_{x}-6u\partial_{x}-3u_x)\varphi(x,t),
		\end{aligned}
	\end{equation}
	where $E$ is the spectral parameter. Based on the Lax pair formulation (\ref{Lax pair}), extensive researches have been conducted on the KdV equation (\ref{KdV}) by using the inverse scattering transform \cite{Ablowitz-Clarkson,Bilman2020,Miller Clarke} and Riemann-Hilbert formulation \cite{Deift1994,Teschl2013}. One of the most notable results is the existence of a special class of localized wave solutions known as solitons. The simplest example of a single soliton solution is given by
	\begin{equation}\label{single soliton}
		u(x, t) = -2 \eta^2 \operatorname{sech}^2 ( \eta(x - 4 \eta^2 t - x_0)),
	\end{equation}
	where the spectral parameter is \( E = \eta^2 \), and \( x_0 \) is the phase parameter that determines the initial position of the soliton. In this context, the position is defined by the location of the maximum of the soliton profile. On the other hand, the KdV equation admits the periodic traveling wave solution of the form \cite{Grava Witham,Ruizhi}
	\begin{equation}\label{periodic-solution}
		{u(x, t)=k_3+\left(k_1-k_3\right) \mathrm{dn}^2\left(\frac{\sqrt{k_1-k_3}}{\sqrt{2} }\left(x-2(k_1+k_2+k_3) t+\frac{\phi_0}{k}\right)-K(m) ; m\right),}
	\end{equation}
	where $k_1>k_2>k_3$, {${\dn}(s;m)$ is the Jacobi elliptic function} and $K(m)$ is a complete elliptic integral of the first kind, i.e., $K(m):=\int_0^{\frac{\pi}{2}} \frac{d \vartheta}{\sqrt{1-m^2 \sin ^2 \vartheta}}$ with $m=\frac{k_1-k_2}{k_1-k_3}$ and $k=\pi \frac{\sqrt{k_1-k_3}}{\sqrt{2} K(m)}$. Especially, as $k_2\to k_3$, the periodic solution (\ref{periodic-solution}) degenerates into the soliton solution by the identity $\dn(\bullet;1)=\sech(\bullet)$.
	\par	
	In 1971, Zakharov \cite{Zakharov 1971} first introduced the concept of ``soliton gas'' and derived an integro-differential kinetic equation for the soliton gas by evaluating the efficient modification of the soliton velocity within a rarefied gas. Specifically, he treated solitons as ``particles'', and a soliton gas can be understood as a collection of randomly distributed solitons, resembling the behavior of a gas \cite{Pelinovsky 2016}. Forty-five years later, in 2016, Zakharov and his collaborators \cite{DZZ16} revisited the soliton gas for the KdV equation by using the dressing method and proposed an alternate construction
	of the Bargmann potentials. In particular, they formulated a Riemann-Hilbert problem (RH problem) associated with the soliton gas, given as follows:
	\begin{equation}\label{DZZ16 RHP}
		\begin{aligned}
			&\Xi^{+}(i \kappa)=M(x,\kappa) \Xi^{-}(i \kappa), \quad \Xi^{+}(-i \kappa)=M^T(x,\kappa) \Xi^{-}(-i \kappa),\\
			&M(x, \kappa)=\frac{1}{1+R_1 R_2}\left[\begin{array}{ll}
				1-R_1 R_2 & 2 i R_1 e^{-2 \kappa x} \\
				2 i R_2 e^{2 \kappa x} & 1-R_1 R_2
			\end{array}\right],
		\end{aligned}
	\end{equation}
	where $\Xi:\mathbb{C}\to\mathbb{C}^2$ is a vector-valued function, and $\kappa \in [k_1, k_2]$ with $0<k_1<k_2$. Although research on soliton gas began many years ago, the understanding of the properties of an interacting ensemble of large solitons and their dynamic behavior, even in the absence of randomness, remains incomplete from a mathematically precise perspective. In 2003, El \cite{El2003,El KdV} proposed a unified extension of Zakharov's kinetic equation for the KdV dense soliton gas by considering the thermodynamic limit of KdV–Whitham equations. Subsequently, the kinetic equation for soliton gas was examined for its diverse and complex mathematical characteristics \cite{Ablowitz KdV,El2011,Ferapontov2022}. In addition, Bertola et al. derived the kinetic equation for the KdV equation by using the method of genus degeneration in \cite{Bertola Nonlinearlarity}. Recently, Girotti and her collaborators \cite{Girotti CMP} investigated the genus one KdV soliton gas and established an asymptotic description of soliton gas dynamics for large time by using the Deift-Zhou nonlinear steepest descent method \cite{Deift 1993}. They \cite{Grava CPAM} also investigated the behaviors of a trial soliton travelling through a mKdV soliton gas and built the kinetic theory for soliton gas. For a concise relationship between the mKdV equation and KdV equation, please refer to \cite{Lenells}. {In 2020, Nabelek gave an insightful investigation of the algebro-geometric finite gap solutions to the KdV equation \cite{Nabelek 2020 PhysD} utilizing the primitive solution framework \cite{Nabelek 2020 TMP,Zakharov 2016 LMP} for the general case of ``$N$ bands'' and $R_1R_2\neq 0$ in (\ref{DZZ16 RHP}).} In fact, the results presented in \cite{Girotti CMP,Grava CPAM} represent a particular case of (\ref{DZZ16 RHP}) for $R_2=0$, involving only two disjoint stability zones. Furthermore, the study of soliton gas for the NLS equation was examined in \cite{Bertola PRL, Grava FNLS, Biondini Breather gas, Fudong JPA}, and the relationship between periodic potentials was explored in \cite{Nabelek IMRN}. {In particular, in \cite{GravaDuke}, the authors investigated random configurations of soliton gases for the focusing NLS equation and established both a law of large numbers and a central limit theorem for random sets of solitons.}
	\par
	This paper investigates the high-genus soliton gas  for the KdV equation (\ref{KdV}), focusing specifically on the genus two soliton gas potential and its long-time asymptotics. More precisely, we consider the special case of (\ref{DZZ16 RHP}) with $R_1 = 0$ and $R_2 =r_2(\lambda)$, which involves four disjoint stability zones. Suppose $0 < \eta_1 < \eta_2 < \eta_3 < \eta_4$ and let $\Sigma_1 := (\eta_1, \eta_2)$, $\Sigma_2 := (-\eta_2, -\eta_1)$, $\Sigma_3 := (\eta_3, \eta_4)$, and $\Sigma_4 := (-\eta_4, -\eta_3)$, {see Figure \ref{jumpforY}}. Additionally, denote $\Sigma_{i,\cdots,k} = \Sigma_i \cup \cdots \cup \Sigma_k$. Let {$\theta(x,t;\lambda):=x\lambda +4 t\lambda^3$}, and  then construct the following RH problem for the vector-valued function $X(\lambda)$ as
	\begin{equation}\label{RHP X jumps}
		\begin{aligned}
			X_+(\lambda)=X_-(\lambda)
			\begin{cases}
				\begin{aligned}
					&\begin{pmatrix}
						1 & -2ir_2(\lambda)e^{-2 i\theta(x,t;\lambda)}\\
						0 & 1
					\end{pmatrix},&&\lambda\in i\Sigma_{1,3},\\
					&\begin{pmatrix}
						1 & 0\\
						2ir_2(\lambda)e^{2 i \theta(x,t;\lambda)} & 1
					\end{pmatrix},&&\lambda\in i\Sigma_{2,4},			
				\end{aligned}
			\end{cases}
		\end{aligned}
	\end{equation}
	\begin{equation}\label{RHP X asymp}
		X(\lambda)\to \begin{pmatrix}
			1 & 1
		\end{pmatrix},\quad \lambda\to\infty,
	\end{equation}
	\begin{equation}\label{RHP X sym}
		X(-\lambda)=X(\lambda)\begin{pmatrix}
			0 & 1\\
			1 & 0
		\end{pmatrix}.
	\end{equation}
	Then the genus two soliton gas potential of the KdV equation (\ref{KdV}) is given by the reconstruction formula
	\begin{equation}\label{two-genus-soliton-gas-potential}
		u(x) = 2 \frac{\mathrm{d}}{\mathrm{d} x} \left( \lim_{\lambda \to \infty} \frac{\lambda}{i} \left(X_1(\lambda) - 1\right) \right),
	\end{equation}
	where \(X_1(\lambda)\) is the first component of \(X(\lambda)\).
	\par
	In what follows, we propose the main results of this work.
	
	\subsection{Statement of the main results}
	
	Firstly, a genus two KdV soliton gas potential (\ref{two-genus-soliton-gas-potential}) is constructed by formulating the Riemann-Hilbert problem (\ref{RHP X jumps})-(\ref{RHP X sym}) from the pure \(N\)-soliton Riemann-Hilbert problem in Section \ref{Soliton gas RH problem} for \(N \to +\infty\), where the initial positions of the \(N\) solitons are located on the positive real axis. Then, in Section \ref{potential behavior}, we establish the large $x$ behaviors of this soliton gas potential in Theorem \ref{2genus potential}.
	
	\begin{thm}\label{Spatial behavior}
		The potential function \( u(x) \), which satisfies the reconstruction formula (\ref{two-genus-soliton-gas-potential}) and the Riemann-Hilbert problem (\ref{RHP X jumps})-(\ref{RHP X sym}), exhibits the following asymptotic behaviors:
		\begin{equation}\label{2genus potential}
			u(x)=\begin{cases}
				\begin{aligned}
					&-\left(2\alpha+{\sum_{j=1}^4\eta_j^2}+2\partial_x^2\log\left(\Theta\left(\frac{\Omega}{2\pi i};\hat\tau\right)\right)\right)+\mathcal{O}\left(\frac{1}{x}\right),&& x\to+\infty,\\
					&\mathcal{O}(e^{-c|x|}),&& x\to-\infty.
				\end{aligned}
			\end{cases}
		\end{equation}
		Here, \(\Theta(\bullet;\hat{\tau})\) denotes the two-phase Riemann-Theta function defined by (\ref{rsf}), $\Omega$ is a two-dimensional column vector given by (\ref{Omee}) and the imaginary part of the period matrix \(\hat{\tau}\), as defined in (\ref{period matrix of hat tau}), is positive definite. Furthermore, the parameter \(\alpha\) is defined in Remark \ref{alpha}, and \(c\) is a fixed positive constant.
	\end{thm}
    {Indeed, the initial configuration can be regarded as a Riemann problem for the KdV condensate. In~\cite{Congy JNS}, Congy et al. studied the case involving a transition between genus~$0$ and genus~$1$. In our case, the problem can be interpreted as a generalized rarefaction scenario.} Figure \ref{2genusGas} presents a direct numerical simulation of the KdV equation (\ref{KdV}) with initial potential (\ref{two-genus-soliton-gas-potential}) behaving the asymptotics in equation (\ref{2genus potential}) with parameters \(\eta_1 = 0.8\), \(\eta_2 = 1.2\), \(\eta_3 = 1.6\), \(\eta_4 = 2\), and \( r_2(\lambda) = 1\). The Figure \ref{2genusGas} clearly shows that the plane is divided into five distinct regions, which from left to right are {quiescent} region, modulated one-phase wave region, unmodulated one-phase wave region, modulated two-phase wave region and unmodulated two-phase wave region.
	
	\begin{figure}
		\centering
		\includegraphics[width=15cm]{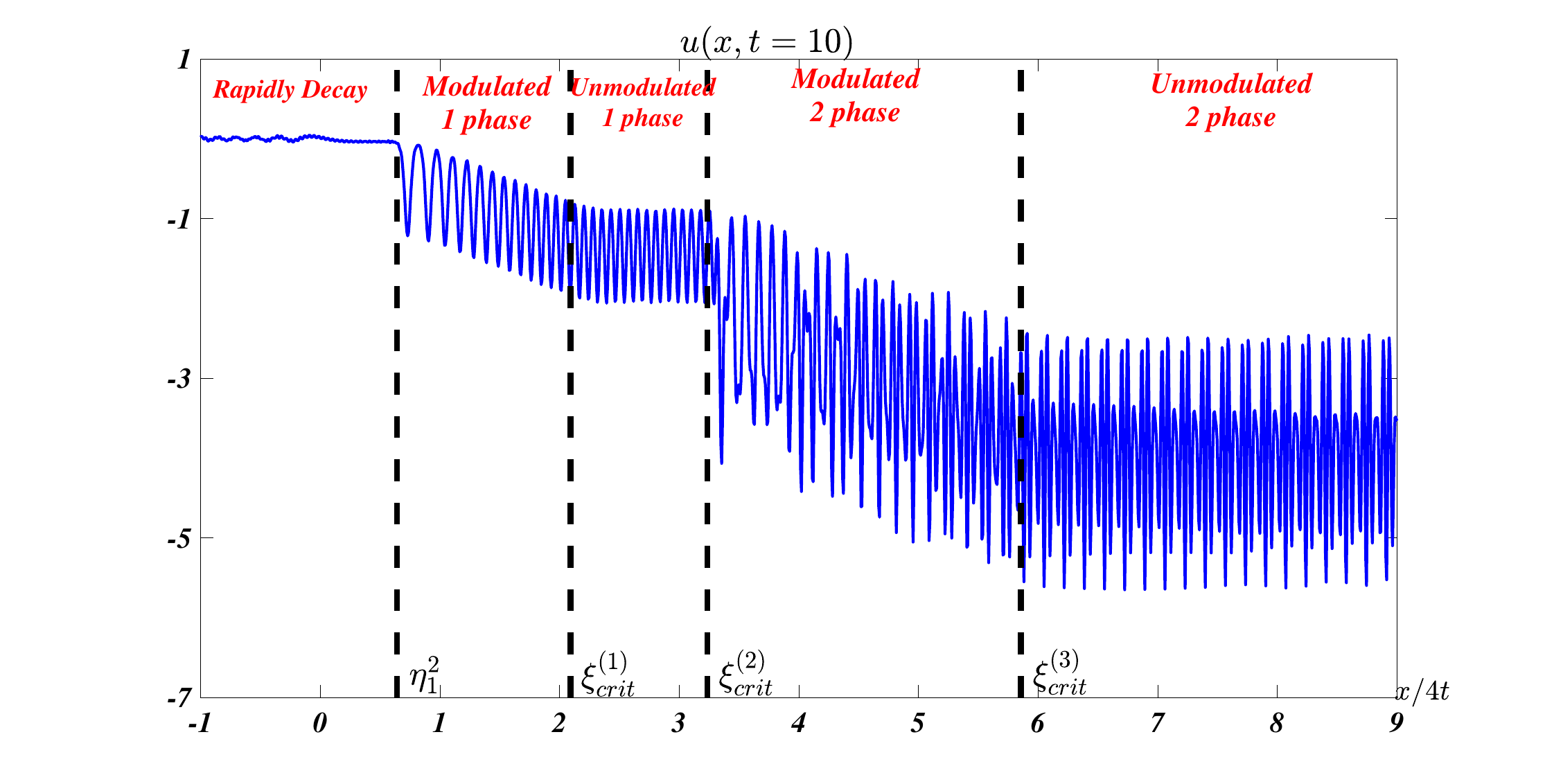}
		\caption{{\protect\small The evolution of the genus two soliton gas potential of the KdV equation at \( t = 10 \) for parameters \(\eta_1 = 0.8\), \(\eta_2 = 1.2\), \(\eta_3 = 1.6\), \(\eta_4 = 2\), and \( r_2(\lambda) = 1 \). The horizontal axis represents \(\frac{x}{4t}\), and the critical points $\eta_1^2$ and \(\xi_{\text{crit}}^{(j)}\) for \( j = 1, 2, 3 \) partition the plane into five distinct regions. These critical values, \(\xi_{\text{crit}}^{(j)}\), can be calculated by using equations (\ref{xi crit}) and (\ref{xi critical}) approximately. For the given parameters, the approximate values are \(\xi_{\text{crit}}^{(1)} \approx 2.0905\), \(\xi_{\text{crit}}^{(2)} \approx 3.2338\), and \(\xi_{\text{crit}}^{(3)} \approx 5.8561\).
		}}
		\label{2genusGas}
	\end{figure}
	
	More precisely, the long-time asymptotics of \( u(x,t) \) for the genus two KdV soliton gas potential depends on the parameter \(\xi := \frac{x}{4t}\). There are four critical values, i.e., $\eta_1^2$ and \(\xi_{\text{crit}}^{(j)}\) for \( j = 1, 2, 3 \), defined in equations (\ref{xi critical}) and (\ref{xi crit}), which serve as the boundaries between different regions, as illustrated in Figure \ref{x-t plane} and Theorem \ref{long time asymp} below. The proof of Theorem \ref{long time asymp} will be provided in detail in Section \ref{long time section}.
	\begin{figure}[h]
		\centering
		\begin{tikzpicture}[>=latex]
			\draw[->,black,very thick] (-5,0) to (5,0) node[black,below=1mm]  {\small $x$};
			\draw[->,dashed,black,very thick] (-3,0) to (-3,5) node[black,right=1mm]  {\small $t$};
			\draw[-,black,very thick] (-3,0) to (5.0,2.8) node[black,above=0.5mm]  {\small $\xi=\xi_{crit}^{(3)}$};
			\draw[-,black,very thick] (-3,0) to (-2,4.8) node[black,above=0.5mm]  {\small $\xi=\eta_{1}^2$};
			\draw[-,black,very thick] (-3,0) to (0.5,4.8) node[black,above=0.5mm]  {\small $\xi=\xi_{crit}^{(1)}$};
			\draw[-,black,very thick] (-3,0) to (3.1,4.5) node[black,above=0.1mm]  {\small $\xi=\xi_{crit}^{(2)}$} ;
			\node at (3.5,1.5) {\small Unmodulated};
			\node at (3.5,1.0) {\small 2 phase };
			\node at (3.2,3.6) {\small Modulated };
			\node at (3.2,3.1) {\small 2 phase };
			\node at (1,4) {\small Unmodulated };
			\node at (1,3.5) {\small 1 phase };
			\node at (-1.2,4) {\small Modulated };
			\node at (-1.2,3.5) {\small 1 phase };
			\node at (-3.1,3.5) {\small {Quiescent}};
			\node at (-3.1,3) {\small {Region}};
			\node at (-3,-0.2) {\small 0 };
		\end{tikzpicture}
		\caption{{\protect\small
				Five asymptotic regions of the genus two KdV soliton gas potential in the $x$-$t$ half plane.}}
		\label{x-t plane}
	\end{figure}
	
	\begin{thm}\label{long time asymp}
		As \( t \to +\infty \), the global long-time asymptotic behaviors of \( u(x,t) \) for the KdV equation with initial potential (\ref{two-genus-soliton-gas-potential}) behaving the asymptotics in equation (\ref{2genus potential}) can be described as follows:
		\begin{enumerate}
			\item For fixed \(\xi < \eta_1^2\), there exists a positive constant \( c \) such that
			\[
			u(x, t) = \mathcal{O}\left(e^{-c t}\right).
			\]
			\item For \(\eta_1^2 < \xi < \xi_{\text{crit}}^{(1)}\), the long-time asymptotics of $u(x,t)$ can be described by a Jacobi elliptic function ``$\mathrm{dn}$'' with modulated parameter {\(\alpha_1\in(\eta_1,\eta_2)\)} and modulated modulus \(m_{\alpha_1} = \frac{\eta_1}{\alpha_1}\) as
			\[
			u(x, t) = \alpha_1^2 - \eta_1^2 - 2\alpha_1^2 \mathrm{dn}^2\left(\alpha_1 \left(x - 2(\alpha_1^2 + \eta_1^2)t + \phi_{\alpha_1}\right) + K(m_{\alpha_1}); m_{\alpha_1} \right) + \mathcal{O}\left(\frac{1}{t}\right),
			\]
			where the parameter \(\alpha_1\) is determined by equation (\ref{alpha1 formular}), and
			\[
			\phi_{\alpha_1} = \int_{\alpha_1}^{\eta_1} \frac{\log r(\zeta)}{R_{\alpha_1,+}(\zeta)} \frac{d\zeta}{\pi i},
			\]
			with \( R_{\alpha_1}(\lambda) := \sqrt{(\lambda^2 - \eta_1^2)(\lambda^2 - \alpha_1^2)} \), where \( R_{\alpha_1,+}(\lambda) \) denotes the left boundary of $ R_{\alpha_1}(\lambda)$ along the branch cuts $(\eta_1,\eta_2)$ and $(-\eta_2,-\eta_1)$.
  \( K(m_{\alpha_1}) \) is the complete elliptic integral of the first kind, defined as \( K(m_{\alpha_1}) = \int_0^{\frac{\pi}{2}} \frac{d\vartheta}{\sqrt{1 - m_{\alpha_1}^2 \sin^2 \vartheta}} \).
			
			\item For \(\xi_{\text{crit}}^{(1)} < \xi < \xi_{\text{crit}}^{(2)}\), the long-time asymptotics of $u(x,t)$ can be described by a Jacobi elliptic function ``$\mathrm{dn}$'' with constant coefficients below
			$$
			u(x,t)=\eta_{2}^2-\eta_{1}^2-2\eta_{2}^2 \dn^2\left(\eta_{2}(x-2(\eta_{2}^2+\eta_{1}^2)t+\phi_{\eta_{2}})+K(m_{\eta_{2}}); m_{\eta_{2}}\right)
			+\mathcal{O}\left(\frac{1}{t}\right),
			$$
			where $m_{\eta_{2}}=\frac{\eta_{1}}{\eta_{2}}$, and  
$$
\phi_{\eta_{2}}=\int_{\eta_{2}}^{\eta_{1}}\frac{\log r(\zeta)}{R_{\eta_{2},+}(\zeta)}\frac{d\zeta}{\pi i},
$$  
{where \( R_{\eta_2}(\lambda) \) is obtained by replacing \( \alpha_1 \) with \( \eta_2 \) in \( R_{\alpha_1}(\lambda) \), while all other conventions remain consistent.}  
	
			\item For \(\xi_{\text{crit}}^{(2)} < \xi < \xi_{\text{crit}}^{(3)}\), the long-time asymptotics of $u(x,t)$ can be described by the modulated two-phase wave
			\begin{equation}\label{modulated-two-phase-wave}
				u(x, t) = -\left(2b_{\alpha_2,1} + \sum_{j=1}^3 \eta_j^2 + \alpha_2^2 + 2 \partial_x^2 \log\left(\Theta\left(\frac{\Omega_{\alpha_2}}{2\pi i}; \hat{\tau}_{\alpha_2}\right)\right)\right) + \mathcal{O}\left(\frac{1}{t}\right),
			\end{equation}
			where the constant \( b_{\alpha_2,1} \) is determined by equation (\ref{Q alpha2 condition}),
			parameter \(\alpha_2\) is determined by equation (\ref{alpha2 formular}) and the period matrix $\hat{\tau}_{\alpha_2}$ is given in equation (\ref{period matrix of hat tau alpha2}).
			
			\item For fixed \(\xi_{\text{crit}}^{(3)} < \xi\), the long-time asymptotics of $u(x,t)$ can be described by the unmodulated two-phase wave
			$$			u(x,t)=-\left(2b_{\eta_{4},1}+{\sum_{j=1}^4\eta_j^2}+2\partial_x^2\log\left(\Theta\left(\frac{\Omega_{\eta_{4}}}{2\pi i};\hat \tau_{\eta_{4}}\right)\right)\right)+\mathcal{O}\left(\frac{1}{t}\right),
			$$
			where $\hat \tau_{\eta_{4}}=\hat{\tau}$ in (\ref{period matrix of hat tau}) and $b_{\eta_{4}}, \Omega_{\eta_{4}}=\begin{pmatrix}
				{t\Omega_{\eta_{4},1}+{\Delta_{\eta_{4},1}}}&{t\Omega_{\eta_{4},0}+{\Delta_{\eta_{4},0}}}
			\end{pmatrix}^T$ are defined by (\ref{Q alpha2 condition}) and (\ref{Omega alpha2}), respectively.
			
		\end{enumerate}
	\end{thm}
	\par
	In fact, the method for genus two KdV soliton gas can be generalized to investigate the soliton gases of arbitrary genus. In Section \ref{Ngenus sector}, the construction of KdV soliton gases of general genus \( \mathcal{N} \) is discussed, along with a preliminary analysis of their evolutionary properties.
	
	\subsection{Some remarks on Theorem \ref{Spatial behavior} and Theorem \ref{long time asymp}}

	It will be seen that the studies of asymptotic behaviors of the genus two KdV soliton gas potential (\ref{two-genus-soliton-gas-potential}) are not the trivial generalization of that in the genus one KdV soliton gas potential in \cite{Girotti CMP}. The leading-order term in asymptotic expression (\ref{modulated-two-phase-wave}) also arises from the context of the small-dispersion limit of the KdV equation, as discussed in \cite{Claeys Grava,Deift 1997}. Although the initial RH problem (\ref{RHP X jumps})-(\ref{RHP X sym}) involves four jump bands, corresponding to a genus three scenario, after applying the holomorphic map \( z= -\lambda^2 \) and by using the Riemann-Hurwitz formula, the corresponding Riemann surface is indeed of genus two. Moreover, the model problem associated with this new Riemann surface resembles the one in \cite{Deift 1997}, and this approach can be extended to higher-genus cases.
	
	\section{Construction of genus two KdV soliton gas potential}\label{Soliton gas RH problem}
	
	It is known that a pure $N$-soliton solution of the KdV equation (\ref{KdV}) associates with a vector-valued RH problem. To be specific, let $M(\lambda)$ be a $1\times2$ vector satisfying:
	
	(i) $M(\lambda)$ is meromorphic in the whole complex plane, with simple poles at $\left\{\lambda_j\right\}_{j=1}^N$ in $i \mathbb{R}_{+}$ and the corresponding conjugate points $\left\{\bar{\lambda}_j\right\}_{j=1}^N$ in $i \mathbb{R}_{-}$;
	
	(ii) The following residue conditions for $M(\lambda)$ hold
	$$
	\underset{\lambda=\lambda_j}{\operatorname{res}} M(\lambda)=\lim _{\lambda \rightarrow \lambda_j} M(\lambda)
	\begin{pmatrix}
		0 & \frac{c_j e^{-2 i \theta(x,t;\lambda)}}{N} \\
		0 & 0
	\end{pmatrix}
	, \quad \underset{\lambda=\bar\lambda_j}{\operatorname{res}} M(\lambda)=\lim _{\lambda \rightarrow \bar{\lambda}_j} M(\lambda)\begin{pmatrix}
		0 &  0\\
		\frac{-c_j e^{2 i \theta(x,t;\lambda)}}{N} & 0
	\end{pmatrix},
	$$
	where $c_j \in i \mathbb{R}_{+},$ $j=1,\cdots,N$ and {$\theta(x,t;\lambda)=x\lambda+4t\lambda^3$}.
	
	(iii) $M(\lambda)$ satisfies the asymptotics $M(\lambda)= \begin{pmatrix} 			1 & 1 		\end{pmatrix}+\mathcal{O}\left(\frac{1}{\lambda}\right)$ for $\lambda \rightarrow \infty$;
	
	(iv) $M$ admits the symmetry
	$$
	M(-\lambda)=M(\lambda)\left(\begin{array}{ll}
		0 & 1 \\
		1 & 0
	\end{array}\right).
	$$
	\par
	The stationary $N$-soliton solution of the KdV equation (\ref{KdV}) is constructed by
	{$$
	u(x,t)=2 \frac{\mathrm{d}}{\mathrm{d} x}\left(\lim _{\lambda \rightarrow \infty} \frac{\lambda}{i}\left(M_1(\lambda)-1\right)\right),
	$$}
	where $M_1(\lambda)$ is the first element of row vector $M(\lambda)$. Specially, for $N=1$ and taking $\lambda_1=i\eta, c_1=ic,c>0$, the stationary one-soliton solution of the KdV equation (\ref{KdV}) is derived as
	{$$
	u(x,t)=-2 \eta^2 \operatorname{sech}^2\left( \eta\left(x-4\eta^2t-x_0\right)\right),
	$$}
	which is a time-independent version of the single soliton solution (\ref{single soliton}) with initial position
	$$
	x_0=\frac{1}{2\eta}\log\frac{2\eta}{c}\in\R.
	$$
	\par
	Reminding the notations $\Sigma_j~(j=1,2,3,4)$ above, for the sake of simplicity, when considering $N\to+\infty$ restricted on the four bands $(i\eta_1,i\eta_2), (i\eta_3,i\eta_4), (-i\eta_2, -i\eta_1)$ and $(-i\eta_4, -i\eta_3)$ respectively, we take the following assumptions.
	\begin{assumption} Assume the next three items hold:\
		\begin{enumerate}
			\item Divide $\{\lambda_j\}_{j=1}^N$ into two parts, that are $\{\lambda_l\}_{l=1}^{N_1}$ and $\{\lambda_k\}_{k=1}^{N_2}$, with $N_1+N_2=N$. Suppose the $N_1$ poles are uniformly distributed on $(i \eta_1,i\eta_2)$, while the $N_2$ poles are uniformly distributed on $(i \eta_3,i\eta_4)$. More explicitly, let $|\lambda_{l+1}-\lambda_{l}|=\frac{\eta_2-\eta_1}{N_1}$ and $|\lambda_{k+1}-\lambda_{k}|=\frac{\eta_4-\eta_3}{N_2}$.
			
			\item Similarly, the coefficients $c_j$ are also divided into two groups, i.e., $\{c_l\}_{l=1}^{N_1},\{c_k\}_{k=1}^{N_2}$, and are purely imaginary. Moreover, assume that
			$$
			\begin{aligned}				
				c_l=\frac{i(\eta_2-\eta_1)r_2(\lambda_l)}{\pi},~l=1,2,\cdots,N_1,\\
				c_k=\frac{i(\eta_4-\eta_3)r_2(\lambda_k)}{\pi},~k=1,2,\cdots,N_2,\\
			\end{aligned}
			$$
			where $r_2(\lambda)$ is an analytic function for $\lambda$ near $i\Sigma_{1,2,3,4}$, with symmetry $r_2(\bar \lambda)=r_2(\lambda)$. Moreover, $r_2(\lambda)$ is assumed to be a real-valued positive nonvanishing function for $\lambda$ in the closure of $\Sigma_{1,2,3,4}$.
			
			\item Indeed, we only consider the case that $N_1\to+\infty$ and $N_2\to+\infty$, simultaneously.
		\end{enumerate}
	\end{assumption}
	Note that the corresponding conjugate points $\left\{\bar{\lambda}_j\right\}_{j=1}^N$ are considered in the same way on $(-i\eta_2, -i\eta_1)$ and $(-i\eta_4, -i\eta_3)$. Now, remove the poles in $M(\lambda)$ by taking the transformations
	$$
	Z(\lambda)=\begin{cases}
		\begin{aligned}
			&M(\lambda)
			\begin{pmatrix}
				1 & -\frac{1}{N}\sum_{j=1}^{N}\frac{c_j e^{-2 i {\theta(x,t;\lambda)}}}{\lambda-\lambda_j} \\
				0 & 1
			\end{pmatrix},&&\lambda\ \text{within}\ \gamma_+,\\
			&M(\lambda)
			\begin{pmatrix}
				1 & 0  \\
				\frac{1}{N}\sum_{j=1}^{N}\frac{c_j e^{{2 i \theta(x,t;\lambda)}}}{\lambda+\lambda_j} & 1
			\end{pmatrix},&&\lambda\ \text{within}\ \gamma_-,
		\end{aligned}
	\end{cases}
	$$
	where $\gamma_+$ is a counter clockwise contour,  which surrounds the interval $(i\eta_1,i\eta_2)$ and $(i\eta_3,i\eta_4)$  in the upper half plane, and $\gamma_-$is a clockwise contour, which surrounds the interval $(-i\eta_2,-i\eta_1)$ and $(-i\eta_4,-i\eta_3)$ in the lower half plane. Consequently, the jump conditions for function $Z(\lambda)$ are converted into
	$$
	Z_+(\lambda)=Z_-(\lambda)\begin{cases}
		\begin{aligned}
			&\begin{pmatrix}
				1 & -\frac{1}{N}\sum_{j=1}^{N}\frac{c_j e^{-2 i{ \theta(x,t;\lambda)}}}{\lambda-\lambda_j} \\
				0 & 1
			\end{pmatrix}, &&\lambda\in\gamma_+,\\
			&\begin{pmatrix}
				1 & 0  \\
				-\frac{1}{N}\sum_{j=1}^{N}\frac{c_j e^{2 i { \theta(x,t;\lambda)}}}{\lambda+\lambda_j} & 1
			\end{pmatrix},	 &&\lambda\in\gamma_-.
		\end{aligned}
	\end{cases}
	$$
	\par	
	As $N_1,N_2\to+\infty$, one has $N=N_1+N_2\to+\infty$. For any open set $U_{1}$ containing $i\Sigma_{1}\cup i\Sigma_{3}$, the series converges uniformly for all $\lambda\in\C\setminus U_1$, that is
	$$
	\begin{aligned}
		\mathop{{\rm lim}}\limits_{N\to+\infty}\frac{1}{N}\sum_{j=1}^{N}\frac{c_j }{\lambda-\lambda_j}&=	\mathop{{\rm lim}}\limits_{N_1\to+\infty}\frac{1}{N_1}\sum_{l=1}^{N_1}\frac{1}{\lambda-\lambda_l}\frac{(\eta_2-\eta_1)
			ir_2(\lambda_l)}{\pi}+\mathop{{\rm lim}}\limits_{N_2\to+\infty}\frac{1}{N_2}\sum_{k=1}^{N_2}\frac{1}{\lambda-\lambda_k}\frac{(\eta_4-\eta_3)ir_2(\lambda_k)}{\pi}\\
		&=\int_{i\eta_1}^{i\eta_2}\frac{2ir_2(\zeta)}{\lambda-\zeta}\frac{d\zeta}{2\pi i}+\int_{i\eta_3}^{i\eta_4}\frac{2ir_2(\zeta)}{\lambda-\zeta}\frac{d\zeta}{2\pi i}.
	\end{aligned}
	$$
	Similarly, for any open set $U_{2}$ containing $i\Sigma_{2}\cup i\Sigma_{4}$, the series converges uniformly for all $\lambda\in\C\setminus U_2$, that is
	$$
	\mathop{{\rm lim}}\limits_{N\to+\infty}\frac{1}{N}\sum_{j=1}^{N}\frac{c_j }{\lambda+\lambda_j}=\int_{-i\eta_2}^{-i\eta_1}\frac{2ir_2(\zeta)}{\zeta-\lambda}\frac{d\zeta}{2\pi i}+\int_{-i\eta_4}^{-i\eta_3}\frac{2ir_2(\zeta)}{\zeta-\lambda}\frac{d\zeta}{2\pi i}.
	$$
	\par
	As result, a limiting RH problem for $Z(\lambda)$ is obtained below
	$$
	Z_+(\lambda)=Z_-(\lambda)\begin{cases}
		\begin{aligned}
			&\begin{pmatrix}
				1 & e^{-2i{ \theta(x,t;\lambda)}}(\int_{i\eta_1}^{i\eta_2}\frac{2ir_2(\zeta)}{\zeta-\lambda}\frac{d\zeta}{2\pi i}+\int_{i\eta_3}^{i\eta_4}\frac{2ir_2(\zeta)}{\zeta-\lambda}\frac{d\zeta}{2\pi i})\\
				0 & 1
			\end{pmatrix},&&\lambda\in\gamma_+,\\
			&\begin{pmatrix}
				1 & 0\\
				e^{2i{ \theta(x,t;\lambda)}}(\int_{-i\eta_2}^{-i\eta_1}\frac{2ir_2(\zeta)}{\zeta-\lambda}\frac{d\zeta}{2\pi i}+\int_{-i\eta_4}^{-i\eta_3}\frac{2ir_2(\zeta)}{\zeta-\lambda}\frac{d\zeta}{2\pi i}) & 1
			\end{pmatrix},&&\lambda\in\gamma_-,\\
		\end{aligned}
	\end{cases}
	$$
	$$Z(\lambda)= \begin{pmatrix} 			1 & 1 		\end{pmatrix}+\mathcal{O}\left(\frac{1}{\lambda}\right),\quad  \text{for}\quad  \lambda \to \infty.$$
	\par
	Now, {comparing} the jump conditions of $Z(\lambda)$ on the contour $\gamma_{\pm}$ into jumps on $\Sigma_{1,2,3,4}$ by defining
	$$
	X(\lambda)=
	\begin{cases}
		\begin{aligned}
			&Z(\lambda)\begin{pmatrix}
				1 & -e^{-2i{ \theta(x,t;\lambda)}}(\int_{i\eta_1}^{i\eta_2}\frac{2ir_2(\zeta)}{\zeta-\lambda}\frac{d\zeta}{2\pi i}+\int_{i\eta_3}^{i\eta_4}\frac{2ir_2(\zeta)}{\zeta-\lambda}\frac{d\zeta}{2\pi i})\\
				0 & 1
			\end{pmatrix},&& \lambda\ \text{within}\ \gamma_+,\\
			&Z(\lambda)\begin{pmatrix}
				1 & 0\\
				e^{2i{ \theta(x,t;\lambda)}}(\int_{-i\eta_2}^{-i\eta_1}\frac{2ir_2(\zeta)}{\zeta-\lambda}\frac{d\zeta}{2\pi i}+\int_{-i\eta_4}^{-i\eta_3}\frac{2ir_2(\zeta)}{\zeta-\lambda}\frac{d\zeta}{2\pi i}) & 1
			\end{pmatrix},&& \lambda\ \text{within}\ \gamma_-,\\
			&Z(\lambda),  && \lambda\ \text{outside}\ \gamma_{\pm}.
		\end{aligned}
	\end{cases}
	$$
	\par
	By using the Plemelj formula, the RH problem for function $X(\lambda)$ in (\ref{RHP X jumps})-(\ref{RHP X sym}) is derived immediately. Finally, transform the RH problem on contours $i\Sigma_{1,2,3,4}$ into that on contours $\Sigma_{1,2,3,4}$ by defining $Y(\lambda)=X(i\lambda),\ r(\lambda)=2r_2(i\lambda)$, then we arrive at the RH problem for the soliton gas potential of the KdV equation as follows.
	\par
	The function $Y(\lambda)$ is analytic for $\lambda \in \mathbb{C} \setminus \Sigma_{1,2,3,4}$ with $\Sigma_{1,2,3,4}:=\Sigma_1\cup \Sigma_2\cup \Sigma_3\cup\Sigma_4$, {see Figure \ref{jumpforY}}, and has the properties:	
	$$
	Y_{+}( \lambda)=Y_{-}( \lambda) \begin{cases}{\left(\begin{array}{cc}
				1 &  -i r( \lambda) e^{{2 \lambda x-8\lambda^3t}} \\
				0 & 1
			\end{array}\right)}, & \lambda \in \Sigma_{1,3}, \\
		{\left(\begin{array}{cc}
				1 & 0 \\
				i r( \lambda) e^{{-2 \lambda x+8\lambda^3t}} & 1
			\end{array}\right)}, & \lambda \in \Sigma_{2,4},\end{cases}
	$$	
	$$
	Y(\lambda)=(1\quad 1)+\mathcal{O}\left(\frac{1}{\lambda}\right),
	$$	
	$$
	Y(-\lambda)=Y(\lambda)\left(\begin{matrix}
		0&1\\
		1&0
	\end{matrix}\right).
	$$	
\begin{figure}[h]
		\centering
		\begin{tikzpicture}[>=latex]
			\draw[lightgray,very thick,dashed] (-8.5,0) to (-7.5,0);
			\draw[lightgray,very thick,dashed] (-7.5,0) to (-6.5,0) ;
			\draw[-,very thick,lightgray,dashed] (-7.5,0) to (-7,0);
			\filldraw[black] (-6.5,0) node[black,below=1mm]{$-\eta_{4}$} circle (1.5pt);
			\draw[-,very thick] (-6.5,0) to (-4.5,0);
			\draw[->,very thick] (-6.5,0) to (-5.5,0)node[black,above=1mm]{\small $\begin{pmatrix}
				1 & 0 \\
				i r( \lambda) e^{-2 \lambda x} & 1
				\end{pmatrix}$}
                node[black,below=1mm]{\small $\Sigma_4$};
			\draw[-,,dashed,lightgray,very thick] (-4.5,0) to (-3.5,0);
			
			\filldraw[black] (-4.5,0) node[black,below=1mm]{$-\eta_{3}$} circle (1.5pt);
			\filldraw[black] (-3.5,0) node[black,below=1mm]{$-\eta_2$} circle (1.5pt);

			\draw[-,very thick,black] (-3.5,0) to (-0.5,0);
			\draw[->,very thick,black] (-3.5,0) to (-2.0,0) node[black,above=1mm]{\small $\begin{pmatrix}
	            1 & 0 \\
				i r( \lambda) e^{-2 \lambda x} & 1
				\end{pmatrix}$}
                node[black,below=1mm]{\small $\Sigma_2$};
			\filldraw[black] (-0.5,0) node[black,below=1mm]{$-\eta_1$} circle (1.5pt);

			\draw [lightgray,dashed,very thick] (-0.5,0) to (0.5,0);

);
			\draw[lightgray,very thick,dashed] (7.5,0) to (8.5,0);
			\draw[lightgray,very thick,dashed] (6.5,0) to (7.5,0) ;
			\draw[-,very thick,lightgray,dashed] (6.5,0) to (7.2,0);

			\filldraw[black] (6.5,0) node[black,below=1mm]{$\eta_{4}$} circle (1.5pt);
			\draw[-,very thick,lightgray,dashed] (3.5,0) to (4.5,0);
	
			\filldraw[black] (4.5,0) node[black,below=1mm]{$\eta_{3}$} circle (1.5pt);
			\filldraw[black] (3.5,0) node[black,below=1mm]{$\eta_{2}$} circle (1.5pt);
			\draw[-,very thick,black] (0.5,0) to (3.5,0);
			\draw[->,very thick,black] (0.5,0) to (2.0,0) node[black,above=1mm]{\small $\begin{pmatrix}
				1 &  -i r( \lambda) e^{2 \lambda x} \\
				0 & 1
				\end{pmatrix}$}node[black,below=1mm]{\small $\Sigma_1$};
			\filldraw[black] (0.5,0) node[black,below=1mm]{$\eta_1$} circle (1.5pt);
			\draw[-,very thick,black] (4.5,0) to (6.5,0);
			\draw[->,very thick,black] (4.5,0) to (5.5,0)node[black,above=1mm]{\small $\begin{pmatrix}
				1 &  -i r( \lambda) e^{2 \lambda x} \\
				0 & 1
				\end{pmatrix}$}node[black,below=1mm]{\small $\Sigma_3$};

		\end{tikzpicture}
		\caption{{\protect\small
				{The jump contour for \( Y(\lambda) \) and the associated jump matrices.}}}
		\label{jumpforY}
	\end{figure}
    
	\par
	So the KdV soliton gas potential can be reformulated by
	{$$
	u(x,t)=2 \frac{\mathrm{d}}{\mathrm{d} x}\left(\lim _{\lambda \rightarrow \infty} {\lambda}\left(Y_1(\lambda )-1\right)\right),
	$$}
	where \(Y_1(\lambda)\) is the first component of vector-valued function \(Y(\lambda)\).

	\begin{lem}
		The solution to the RH problem concerning row vector $Y(\lambda)$ stated above exists and is unique.
	\end{lem}
	\begin{proof}
		
		Rewrite row vector $Y(\lambda)$ as {$(y^{(1)}(\lambda),y^{(2)}(\lambda))$}. Combining the jump conditions on $\Sigma_{2,4}$, it is deduced that
		{$$
		y^{(1)}_+(\lambda)=y^{(1)}_-(\lambda)-ir(\lambda)y^{(2)}_+(\lambda),\quad
		y^{(2)}_+(\lambda)=y^{(2)}_-(\lambda).
		$$}
		It is evident that {$y^{(2)}(\lambda)$} is holomorphic across $\Sigma_{2,4}$, while {$y^{(1)}(\lambda)$} satisfies an inhomogeneous scalar RH problem. For convenience, denote $f(\lambda):=-i\sqrt{r(\lambda)}y^{(2)}(\lambda)$. As a result, the solution for {$y^{(1)}(\lambda)$} can be represented as
		$$
		y^{(1)}(\lambda)=1+\frac{1}{2\pi i}\int_{\Sigma_{2,4}}\frac{\sqrt{r(s)}f(s)}{s-\lambda}ds.
		$$
		Moreover, the symmetry of $Y(\lambda)$ implies that {$y^{(1)}(-\lambda)=y^{(2)}(\lambda)$}, which shows that
		$$
		{y^{(2)}(\lambda)=1+\frac{1}{2\pi i}\int_{\Sigma_{2,4}}\frac{\sqrt{r(s)}f(s)}{s+\lambda}ds.}
		$$
		Multiplying both sides of the above equation by $-i\sqrt{r(x,t;\lambda)}$, an integral equation for $f(\lambda)$ is obtained as
		$$	f(\lambda)+\frac{\sqrt{r(\lambda)}}{2\pi}\int_{\Sigma_{2,4}}\frac{\sqrt{r(s)}f(s)}{s+\lambda}ds
		=-i\sqrt{r(\lambda)},
		$$
		which is equivalent to $(I+T)f=b$, where $T=\frac{\sqrt{r(x,t;\lambda)}}{2\pi}\int_{\Sigma_{2,4}}\frac{\sqrt{r(x,t;s)}}{s+\lambda}ds$ and $b=-i\sqrt{r(\lambda)}$. Due to the fact that the finite interval integral can be treated as a Riemann integral, it follows that the operator $T$ is compact. Moreover, the index of $I+T$ is zero, i.e., $\mathrm{Ind}(I+T)=\text{dim } N_{I+T}-\text{Codim } R_{I+T}=0$, {where \( N_{I+T} \) and \( R_{I+T} \) denote the kernel and range of the operator \( I+T \), respectively.} It implies that $I+T$ is an injective if and only if it is a surjective. It suffices to show that $T$ is a positive operator, i.e., $(Tx,x)\geq 0$. This has been proven in the appendix of Ref. \cite{Girotti CMP}.
	\end{proof}
	
	\section{The large $x$ behaviors of the genus two KdV soliton gas potential}\label{potential behavior}
	
	This section proves the Theorem \ref{Spatial behavior}, which is to examine the large $x$ behaviors of the genus two KdV soliton gas potential constructed in Section \ref{Soliton gas RH problem}.
	\par
	Firstly, let $t=0$ and consider the case of $x\to +\infty$. To deform the RH problem associated with the genus two KdV soliton gas potential, suppose that $g(\lambda)$ satisfies the following scalar RH problem:
	\par	
	The function $g(\lambda)$ is analytic for $\lambda\in \mathbb{C}\setminus[-\eta_4,\eta_4]$, and
	\[
	\begin{aligned}
		& g_+(\lambda)+g_-(\lambda)=2\lambda, && \lambda\in\Sigma_{1,2,3,4}, \\
		& g_+(\lambda)-g_-(\lambda)=\Omega_0, && \lambda\in [-\eta_1,\eta_1], \\
		& g_+(\lambda)-g_-(\lambda)=\Omega_1, && \lambda\in [\eta_2,\eta_3], \\
		& g_+(\lambda)-g_-(\lambda)=\Omega_2, && \lambda\in [-\eta_3,-\eta_2], \\
		& g(\lambda)={\mathcal{O}\left(\frac{1}{\lambda}\right),} && \lambda\to\infty,
	\end{aligned}
	\]
	where $\Omega_{0,1,2}$ are independent of $x$.
	Moreover, the derivative of the function $g(\lambda)$ also satisfies a scalar RH problem of the form
	\[
	\begin{aligned}
		& g'_+(\lambda)+g'_-(\lambda)=2, && \lambda\in\Sigma_{1,2,3,4}, \\
		& g'_+(\lambda)-g'_-(\lambda)=0, && \lambda\in [-\eta_4,\eta_4]\setminus\Sigma_{1,2,3,4}, \\
		& g'(\lambda)=\mathcal{O}\left(\frac{1}{\lambda^2}\right), && \lambda\to\infty.
	\end{aligned}
	\]
	By the uniqueness of solution to the RH problem, it can be checked that $g'(\lambda)$ is an even function.
	\par	
	Introduce
	\[
	R(\lambda)=\sqrt{(\lambda^2-\eta_1^2)(\lambda^2-\eta_2^2)(\lambda^2-\eta_3^2)(\lambda^2-\eta_4^2)},
	\]
	and assume $R_+(\lambda)$ as the upper sheet of $R(\lambda)$, with $R(\lambda)\to+\infty$ as $\lambda\to+\infty$, for the sake of the subsequent discussion. Define
	\[
	\begin{aligned} g'({\lambda})=1-\frac{\lambda^4+\alpha\lambda^2+\beta}{R(\lambda)},
	\end{aligned}
	\]
	and
	\begin{equation}\label{gfunction} g(\lambda)=\lambda-\int_{\eta_4}^\lambda\frac{\zeta^4+\alpha\zeta^2+\beta}{R(\zeta)}d\zeta.
	\end{equation}
	Moreover, introduce a two-sheeted Riemann surface of genus three as follows:
	$$ \mathcal{S}=\{(\lambda,\eta)|\eta^2=(\lambda^2-\eta_1^2)(\lambda^2-\eta_2^2)(\lambda^2-\eta_3^2)(\lambda^2-\eta_4^2)\},
	$$
	which includes two infinite points $\infty_{\pm}$ to ensure the compactness of the Riemann surface. Subsequently, define the basis of cycles for Riemann surface $\mathcal{S}$ shown in Figure \ref{Contour-2-Cavitation}.
	\begin{figure}[h!]
		\centering
		\includegraphics[width=11cm]{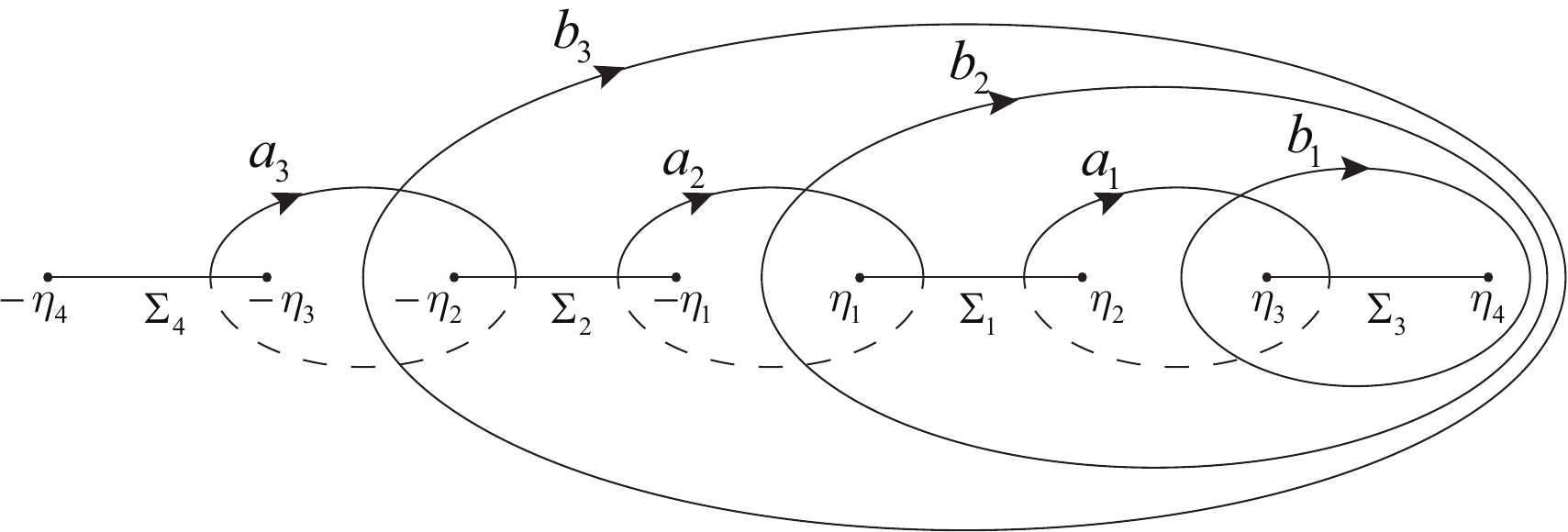}
		\caption{{\protect\small The Riemann surface $\mathcal{S}$ of genus three and its basis of circles.}}
		\label{Contour-2-Cavitation}
	\end{figure}	
	
	The jump conditions for $g(\lambda)$ implies that for $j=1,2,3$
	\begin{equation}\label{jump_condition}
		\begin{aligned}
			& \oint_{a_j}\frac{\zeta^4+\alpha\zeta^2+\beta}{R_+(\zeta)}d\zeta=0, \ \oint_{b_1}\frac{\zeta^4+\alpha\zeta^2+\beta}{R_+(\zeta)}d\zeta=\Omega_1, \\
			& \oint_{b_2}\frac{\zeta^4+\alpha\zeta^2+\beta} {R_+(\zeta)}d\zeta=\Omega_0, \
			\oint_{b_3}\frac{\zeta^4+\alpha\zeta^2+\beta} {R_+(\zeta)}d\zeta=\Omega_2. \\
		\end{aligned}
	\end{equation}
	Furthermore, it should be noted that $\frac{\zeta^4+\alpha\zeta^2+\beta}{R(\zeta)}d\zeta$, denoted as $\eta$, is a second kind Abelian differential on $\mathcal{S}$, with poles only at $\infty_{\pm}$. Introduce a basis of holomorphic differential as
	$$
	{\tilde \omega}_j=\frac{\zeta^{j-1}}{R(\zeta)}d\zeta,\quad j=1,2,3,
	$$
	and denote $A:=(\oint_{a_j}\tilde \omega_i)_{3\times3}$ which is a non-degenerated matrix.
	According to the {Riemann Bilinear relations} \cite{Bertola Riemmansurface}, one can obtain that
	$$
	\sum_{j=1}^3\oint_{a_j}\tilde \omega^i\oint_{b_j}\eta=-{2i\pi}\sum_{p=\infty_{\pm}}\mathrm{res}_{p}\frac{\tilde \omega_i}{\lambda},i=1,2,3,
	$$
	that is,
	\begin{equation}\label{sysofOmega}
		\begin{aligned}	 &\Omega_1\int_{a_1}\tilde\omega_1+\Omega_0\int_{a_2}\tilde\omega_1+\Omega_2\int_{a_3}\tilde\omega_1=0,\\		&\Omega_1\int_{a_1}\tilde\omega_2+\Omega_0\int_{a_2}\tilde\omega_2+\Omega_2\int_{a_3}\tilde\omega_2=0,\\	&\Omega_1\int_{a_1}\tilde\omega_3+\Omega_0\int_{a_2}\tilde\omega_3+\Omega_2\int_{a_3}\tilde\omega_3=4\pi i.\\
		\end{aligned}		
	\end{equation}
	Consequently,  the quantities $\Omega_0$, $\Omega_1$ and $\Omega_2$ can be expressed by
	$$
	\Omega_1=4\pi i(A^{-1})_{13},\ \Omega_0=4\pi i(A^{-1})_{23},\ \Omega_2=4\pi i(A^{-1})_{33}.
	$$
	
	\begin{rmk}\label{alpha}
		Since $\zeta/{R(\zeta)}$ is odd, and $1/{R(\zeta)}$ and $\zeta^2/{R(\zeta)}$ are even, it follows that $A_{13}=A_{11}$ and $A_{22}=0$, which implies that $(A^{-1})_{13}=(A^{-1})_{33}$, i.e., $\Omega_1=\Omega_2$. Alternatively, by using the equalities for $j=1,2$ below
		$$
		\oint_{a_j}\frac{\zeta^4+\alpha\zeta^2+\beta}{R(\zeta)}d\zeta=0,
		$$
		the parameters $\alpha$ and $\beta$ can be determined immediately. 	
	\end{rmk}
	\par	
	Now, {we are} ready to deform the RH problem. To do so, take the transformation
	$$
	T(\lambda)=Y(\lambda)e^{xg(\lambda)\sigma_3}f(\lambda)^{\sigma_3},
	$$
	where $f(\lambda)$ is a function to be determined and $T(\lambda)$ satisfies the RH problem:
	$$
	\begin{aligned}
		& T_{+}(\lambda)=T_{-}(\lambda) V(\lambda) \\
		& V(\lambda)=
		\begin{cases}
			\begin{aligned}
				&\left(\begin{array}{cc}
					e^{x(g_{+}(\lambda)-g_{-}(\lambda))} \frac{f_{+}(\lambda)}{f_{-}(\lambda)} & \frac{-i r(\lambda)}{f_{+}(\lambda) f_{-}(\lambda)} \\
					0 & e^{-x(g_{+}(\lambda)-g_{-}(\lambda))} \frac{f_{-}(\lambda)}{f_{+}(\lambda)}
				\end{array}\right), & \lambda \in \Sigma_{1,3},\\
				& \left(\begin{array}{cc}
					e^{x(g_{+}(\lambda)-g_{-}(\lambda))} \frac{f_{+}(\lambda)}{f_{-}(\lambda)} & 0 \\
					i r(\lambda) f_{+}(\lambda) f_{-}(\lambda) & e^{-x(g_{+}(\lambda)-g_{-}(\lambda))} \frac{f_{-}(\lambda)}{f_{+}(\lambda)}
				\end{array}\right), & \lambda \in \Sigma_{2,4},   \\
				& \left(\begin{array}{cc}
					e^{x \Omega_0} \frac{f_{+}(\lambda)}{f_{-}(\lambda)} & 0 \\
					0 & e^{-x \Omega_0} \frac{f_{-}(\lambda)}{f_{+}(\lambda)}
				\end{array}\right), & \lambda \in\left[-\eta_1, \eta_1\right], \\
				& \left(\begin{array}{cc}
					e^{x \Omega_1} \frac{f_{+}(\lambda)}{f_{-}(\lambda)} & 0 \\
					0 & e^{-x \Omega_1} \frac{f_{+}(\lambda)}{f_{-}(\lambda)}
				\end{array}\right), & \lambda \in[\eta_2,\eta_3], \\
				& \left(\begin{array}{cc}
					e^{x \Omega_1} \frac{f_{+}(\lambda)}{f_{-}(\lambda)} & 0 \\
					0 & e^{-x \Omega_1} \frac{f_{+}(\lambda)}{f_{-}(\lambda)}
				\end{array}\right), & \lambda \in[-\eta_3,-\eta_2],
			\end{aligned}
		\end{cases} \\
		& T(\lambda)=\left(\begin{array}{cc}
			1 & 1
		\end{array}\right) +\mathcal{O}\left(\frac{1}{\lambda}\right),\   \lambda \rightarrow \infty.
	\end{aligned}
	$$
	\par
	Moreover, the function $f(\lambda)$ can be established by the scalar RH problem:
	\begin{align*}
		&f_+(\lambda)f_-(\lambda)={r(\lambda)},&&\lambda\in\Sigma_{1,3},\\
		&f_+(\lambda)f_-(\lambda)=\frac{1}{r(\lambda)},&&\lambda\in\Sigma_{2,4},\\
		&\frac{f_+(\lambda)}{f_-(\lambda)}=e^{\Delta_0},&&\lambda\in [-\eta_1,\eta_1],\\
		&\frac{f_+(\lambda)}{f_-(\lambda)}=e^{\Delta_1},&&\lambda\in [\eta_2,\eta_3],\\
		&\frac{f_+(\lambda)}{f_-(\lambda)}=e^{\Delta_2},&&\lambda\in [-\eta_3,-\eta_2],\\
		&f(\lambda)=1+\mathcal{O}\left(\frac{1}{\lambda}\right),&&\lambda\to\infty,
	\end{align*}
	where $\Delta_0$, $\Delta_1$ and $\Delta_2$ are determined below. The Plemelj formula gives the solution of $f(\lambda)$ as
	\begin{equation}\label{f formular}
		\begin{aligned}
			f(\lambda)=&\exp\left(\frac{R(\lambda)}{2\pi i}\left[
			\int_{\Sigma_{1,3}}\frac{\log{r(\zeta)}}{R_+(\zeta)(\zeta-\lambda)}d\zeta
			+\int_{\Sigma_{2,4}}\frac{\log\frac{1}{r(\zeta)}}{R_+(\zeta)(\zeta-\lambda)}d\zeta
			+\int_{-\eta_1}^{\eta_1}\frac{\Delta_0}{R(\zeta)(\zeta-\lambda)}d\zeta\right.\right.\\
			&\left.\left.+\int_{\eta_2}^{\eta_3}\frac{\Delta_1}{R(\zeta)(\zeta-\lambda)}d\zeta	+\int_{-\eta_3}^{-\eta_2}\frac{\Delta_2}{R(\zeta)(\zeta-\lambda)}d\zeta\right]\right).
		\end{aligned}
	\end{equation}	
	\par	
	Based on the boundary values of $f(\lambda)$, one can determine $\Delta_0$, $\Delta_1$ and $\Delta_2$ through the following system of linear algebraic equations:
	\begin{gather} \int_{\Sigma_{1,3}}\frac{\log{r(\zeta)}}{R_+(\zeta)}d\zeta+\int_{\Sigma_{2,4}}\frac{\log\frac{1}
			{r(\zeta)}}{R_+(\zeta)}d\zeta+\int_{-\eta_1}^{\eta_1}\frac{\Delta_0}{R(\zeta)}d\zeta+\int_{\eta_2}^{\eta_3}
		\frac{\Delta_1}{R(\zeta)}d\zeta+\int_{-\eta_3}^{-\eta_2}\frac{\Delta_2}{R(\zeta)}d\zeta=0,\\
		\int_{\Sigma_{1,3}}\frac{\log{r(\zeta)}}{R_+(\zeta)}\zeta d\zeta+\int_{\Sigma_{2,4}}\frac{\log\frac{1}{r(\zeta)}}{R_+(\zeta)}\zeta d\zeta+\int_{-\eta_1}^{\eta_1}\frac{\Delta_0}{R(\zeta)}\zeta d\zeta+\int_{\eta_2}^{\eta_3}\frac{\Delta_1}{R(\zeta)}\zeta d\zeta+\int_{-\eta_3}^{-\eta_2}\frac{\Delta_2}{R(\zeta)}\zeta d\zeta=0,\label{algebraic-equations-2}\\
		\int_{\Sigma_{1,3}}\frac{\log{r(\zeta)}}{R_+(\zeta)}\zeta^2 d\zeta+\int_{\Sigma_{2,4}}\frac{\log\frac{1}{r(\zeta)}}{R_+(\zeta)}\zeta^2 d\zeta+\int_{-\eta_1}^{\eta_1}\frac{\Delta_0}{R(\zeta)}\zeta^2 d\zeta+\int_{\eta_2}^{\eta_3}\frac{\Delta_1}{R(\zeta)}\zeta^2 d\zeta+\int_{-\eta_3}^{-\eta_2}\frac{\Delta_2}{R(\zeta)}\zeta^2 d\zeta=0.
	\end{gather}
	\par	
	Notice that $r(\zeta)$ is an even function, and $R_{+}(\zeta)$ has the opposite sign for $\zeta\in\Sigma_{1,3}$ compared to $\Sigma_{2,4}$. Thus, from the equation (\ref{algebraic-equations-2}), it is deduced that $\Delta_1=\Delta_2$. Moreover, if expand the function $f(\lambda)$ for large $\lambda$, all involved terms are odd functions, which implies that $f(\lambda)$ tends to one as $\lambda$ approaches infinity.
	\par
	Thus the jump matrix $V(\lambda)$ for $T(\lambda)$ is given by
	\begin{equation}
		V(\lambda)=
		\begin{cases}
			\begin{pmatrix}
				e^{x(g_+(\lambda)-g_-(\lambda))}\frac{f_+(\lambda)}{f_-(\lambda)} & -i \\
				0 & e^{-x(g_+(\lambda)-g_-(\lambda))}\frac{f_-(\lambda)}{f_+(\lambda)}
			\end{pmatrix}, & \lambda\in\Sigma_{1,3},
			\\
			\begin{pmatrix}
				e^{x(g_+(\lambda)-g_-(\lambda))}\frac{f_+(\lambda)}{f_-(\lambda)} & 0 \\
				i & e^{-x(g_+(\lambda)-g_-(\lambda))}\frac{f_-(\lambda)}{f_+(\lambda)}
			\end{pmatrix}, & \lambda\in\Sigma_{2,4},
			
			\\
			\begin{pmatrix}
				e^{x\Omega_0+\Delta_0} & 0 \\
				0 & e^{-(x\Omega_0+\Delta_0)}
			\end{pmatrix}, & \lambda\in[-\eta_1,\eta_1],
			\\
			\begin{pmatrix}
				e^{x\Omega_1+\Delta_1} & 0 \\
				0 & e^{-(x\Omega_1+\Delta_1)}
			\end{pmatrix}, & \lambda\in [\eta_2,\eta_3]\cup [-\eta_3,-\eta_2],  	
		\end{cases}
	\end{equation}
	{where $\Sigma_{1,3}:=(\eta_1,\eta_2)\cup(\eta_3,\eta_4)$ and $\Sigma_{2,4}:=(-\eta_2,-\eta_1)\cup(-\eta_4,-\eta_3)$}.
	
	Now, define the analytic continuation $\hat r(\lambda)$ of $r(\lambda)$ off the interval $\Sigma_{1,2,3,4}$ with $\hat r_{\pm}(\lambda)=\pm r(\lambda)$ for $\lambda\in \Sigma_{1,2,3,4}$, and open lenses as follows
	\begin{equation}
		S(\lambda)=
		\begin{cases}
			T(\lambda)\begin{pmatrix}
				1 & 0 \\
				\frac{if^2(\lambda)}{\hat r(\lambda)}e^{2x(g(\lambda)-\lambda)} & 1
			\end{pmatrix}, & {\rm in~the~upper~lens,~above}~\Sigma_{1,3},
			\\
			T(\lambda)\begin{pmatrix}
				1 & 0 \\
				\frac{if^2(\lambda)}{\hat r(\lambda)}e^{2x(g(\lambda)-\lambda)} & 1
			\end{pmatrix}, & {\rm in~the~lower~lens,~below}~\Sigma_{1,3},
			\\
			T(\lambda)\begin{pmatrix}
				1 & \frac{-i}{\hat r(\lambda)f^2(\lambda)}e^{-2x(g(\lambda)-\lambda)} \\
				0 & 1
			\end{pmatrix}, & {\rm in~the~upper~lens,~above}~\Sigma_{2,4},
			\\
			T(\lambda)\begin{pmatrix}
				1 & \frac{-i}{\hat r(\lambda)f^2(\lambda)}e^{-2x(g(\lambda)-\lambda)} \\
				0 & 1
			\end{pmatrix}, & {\rm in~the~lower~lens,~below}~\Sigma_{2,4},
			\\
			T(\lambda), & {\rm outside~the~lenses}.
		\end{cases}
	\end{equation}

	The vector-valued function $S(\lambda)$ satisfies the RH problem
	\begin{equation}
		\begin{aligned}
			S_+(\lambda)&=S_-(\lambda)V_S(\lambda), \\
			S(\lambda)&=\begin{pmatrix}
				1 & 1
			\end{pmatrix} +O\left(\frac{1}{\lambda}\right), \quad \lambda\to\infty,
		\end{aligned}
	\end{equation}
	where the jump matrices $V_S(\lambda)$ are depicted in the Figure \ref{VS}.
	
	\begin{lem}\label{lem23}
		For $\lambda$ near $\Sigma_{1,3}\setminus\{\eta_j\}$ for $j=1,2,3,4$, the inequality  $\re(g(\lambda)-\lambda)<0$ holds. Conversely, for $\lambda$ near $\Sigma_{2,4}\setminus\{-\eta_j\}$, one has $\re(g(\lambda)-\lambda)>0$.
	\end{lem}
    \begin{proof}
      	{  It is noted that Lemma \ref{lem23} is quite similar to the Lemma \ref{Properties-1} and can be proven by the same way. So we omit the proof here for simplicity.}
    \end{proof}
    \begin{figure}[h]
		\centering
		\begin{tikzpicture}[>=latex]
			\draw[lightgray,very thick,dashed] (-8.5,0) to (-7.5,0);
			\draw[lightgray,very thick] (-7.5,0) to (-6.5,0) ;
			\draw[->,very thick,lightgray] (-7.5,0) to (-7,0);
			\filldraw[lightgray] (-6.5,0) node[black,below=1mm]{$-\eta_{4}$} circle (1.5pt);
			\draw[-,very thick] (-6.5,0) to (-4.5,0);
			\draw[->,very thick] (-6.5,0) to (-6.0,0)node[black,right=-1mm]{\small $\begin{pmatrix}
					0 & i \\
					i & 0
				\end{pmatrix}$};
			\draw[->,very thick] (-4.5,0) to (-4.0,0);
			\draw[-,very thick] (-4.5,0) to (-3.8,0)node[black,above=12mm]{ $e^{(t{\Omega}_{1}+{\Delta}_{1})\sigma_3}$};
			\draw[-,very thick] (-4.5,0) to (-3.3,0)node[black,below=8mm]{\small $\begin{pmatrix}
					1 & \textcolor{gray}{\frac{-i}{\hat r(\lambda)f^2(\lambda)}e^{-2x(g(\lambda)-\lambda)}} \\
					0 & 1
				\end{pmatrix}$};
			\draw[->,dashed,very thick] (-4.0,1.2) to (-4.0,0.2);
			\draw[-,very thick] (-4.5,0) to (-3.5,0);
			\filldraw[black] (-4.5,0) node[black,below=1mm]{$-\eta_{3}$} circle (1.5pt);
			\filldraw[black] (-3.5,0) node[black,below=1mm]{$-\eta_2$} circle (1.5pt);
			\draw[-,very thick,lightgray] (-6.5,0) .. controls (-6,1.0) and (-5,1.0).. (-4.5,0);
			\draw[-,very thick,lightgray,rotate around x=180] (-6.5,0) .. controls (-6,1.0) and (-5,1.0).. (-4.5,0);
			\draw[->,very thick,lightgray] (-5.6,0.75) to (-5.4,0.75);
			\draw[<-,very thick,lightgray,rotate around x=180] (-5.6,0.75) to (-5.4,0.75);
			\draw[-,very thick,black] (-3.5,0) to (-0.5,0);
			\draw[->,very thick,black] (-3.5,0) to (-2.5,0) node[black,right=-1mm]{\small $\begin{pmatrix}
					0 & i \\
					i & 0
				\end{pmatrix}$};
			\filldraw[black] (-0.5,0) node[black,below=1mm]{$-\eta_1$} circle (1.5pt);
			\draw[-,very thick,lightgray] (-3.5,0) .. controls (-2.5,1.15) and (-1.5,1.15).. (-0.5,0);
			\draw[->,very thick,lightgray] (-2.1,0.85) to (-2.0,0.85);
			\draw[-,very thick,lightgray,rotate around x=180] (-3.5,0) .. controls (-2.5,1.15) and (-1.5,1.15).. (-0.5,0) ;
			\draw[<-,very thick,lightgray,rotate around x=180] (-2.3,0.85) to (-2.2,0.85) ;
			\draw [very thick] (-0.5,0) to (0.5,0);
			\draw [->,very thick] (-0.5,0) to (0,0)node[black,above=12mm]{ $e^{(t{\Omega}_{0}+{\Delta}_{0})\sigma_3}$};
			\draw[->,dashed,very thick] (0,1.2) to (0,0.2);
			\draw[lightgray,very thick,dashed] (7.5,0) to (8.5,0);
			\draw[lightgray,very thick] (6.5,0) to (7.5,0) ;
			\draw[->,very thick,lightgray] (6.5,0) to (7.2,0);
			\draw[-,very thick](3.5,0) to (4.0,0) node [black,above=12mm]{ $e^{(t{\Omega}_{1}+{\Delta}_{1})\sigma_3}$};
			\draw[->,dashed,very thick] (4.0,1.2) to (4.0,0.2);
			\filldraw[black] (6.5,0) node[black,below=1mm]{$\eta_{4}$} circle (1.5pt);
			\draw[->,very thick] (3.5,0) to (4.0,0);
			\draw[-,very thick] (3.5,0) to (4.5,0);
			\filldraw[black] (4.5,0) node[black,below=1mm]{$\eta_{3}$} circle (1.5pt);
			\filldraw[black] (3.5,0) node[black,below=1mm]{$\eta_{2}$} circle (1.5pt);
			\draw[-,very thick,black] (0.5,0) to (3.5,0);
			\draw[->,very thick,black] (0.5,0) to (1.5,0) node[black,right=-1mm]{\small $\begin{pmatrix}
					0 & -i \\
					-i & 0
				\end{pmatrix}$};
			\filldraw[black] (0.5,0) node[black,below=1mm]{$\eta_1$} circle (1.5pt);
			\draw[-,very thick,black] (4.5,0) to (6.5,0);
			\draw[->,very thick,black] (4.5,0) to (4.9,0)node[black,right=-2.5mm]{\small $\begin{pmatrix}
					0 & -i \\
					-i & 0
				\end{pmatrix}$};
			\draw[-,very thick,lightgray] (4.5,0) .. controls (5,1) and (6,1).. (6.5,0);
			\draw[-,very thick,lightgray,rotate around x=180] (4.5,0) .. controls (5,1) and (6,1).. (6.5,0);
			\draw[->,very thick,lightgray] (5.5,0.75) to (5.6,0.75);
			\draw[<-,very thick,lightgray,rotate around x=180] (5.4,0.75) to (5.5,0.75);
			\draw[-,very thick,lightgray] (0.5,0) .. controls (1.5,1.15) and (2.5,1.15).. (3.5,0);
			\draw[->,very thick,lightgray] (2,0.855) to (2.1,0.855);
			\draw[-,very thick,lightgray,rotate around x=180] (0.5,0) .. controls (1.5,1.15) and (2.5,1.15).. (3.5,0);
			\draw[->,very thick,lightgray,rotate around x=180] (2.1,0.85) to (2.0,0.85)node[black,anchor=south west]at(1,2.1){\small $\begin{pmatrix}
					1 & 0 \\
					\textcolor{gray}{\frac{if^2(\lambda)}{\hat r(\lambda)}e^{2x(g(\lambda)-\lambda)}} & 1
				\end{pmatrix}$} ;
		\end{tikzpicture}
		\caption{{\protect\small
				The jump contours for \( S(\lambda) \) and the associated jump matrices: the gray terms in the matrices vanish exponentially as \( x \to +\infty \), and the gray contours also vanish as \( x\to +\infty \). }}
		\label{VS}
	\end{figure}
	Then for $x\to +\infty$, we arrive at the model problem $S^{\infty}(\lambda)$ as
	\begin{equation}\label{Sinf}
		S^{\infty}_+(\lambda)=S^{\infty}_-(\lambda) \begin{cases}
			\begin{pmatrix}
				0 & -i \\
				-i & 0
			\end{pmatrix}, & \lambda\in\Sigma_{1,3},
			\\	
			\begin{pmatrix}
				0 & i \\
				i & 0
			\end{pmatrix},& \lambda\in\Sigma_{2,4},
			\\
			\begin{pmatrix}
				e^{x\Omega_0+\Delta_0} & 0 \\
				0 & e^{-x\Omega_0-\Delta_0}
			\end{pmatrix}, & \lambda\in[-\eta_1,\eta_1],
			\\
			\begin{pmatrix}
				e^{x\Omega_1+\Delta_1} & 0 \\
				0 & e^{-x\Omega_1-\Delta_1}
			\end{pmatrix}, & \lambda\in[\eta_2,\eta_3]\cup[-\eta_3,-\eta_2],
		\end{cases}
	\end{equation}
	where $S^{\infty}(\lambda)$ satisfies the boundary condition $S^{\infty}(\lambda)\to (1,1)$ as $\lambda\to\infty$ and the symmetry $S^{\infty}(-\lambda)=S^{\infty}(\lambda)
	\begin{pmatrix}
		0&1\\
		1&0	
	\end{pmatrix}$, which serves as a generalization of the model problem discussed in \cite{Teschl2013}. Moreover, the model problem admits at most fourth root singularities at {$\eta_j,~j=1,\cdots,4$}. In particular, the local parametrix near \(\lambda = \pm \eta_j~(j = 1,2, \dots, 4)\) is expressed in terms of the modified Bessel functions under suitable conformal map \cite{Girotti CMP}.
	
	\subsection{The solution of the model RH problem on the $z$-plane}
	
	Now, transform the model RH problem (\ref{Sinf}) from the $\lambda$ plane into $z$ plane for $z=-\lambda^2$ by taking the lower half $\lambda$-plane onto $\C\setminus(-\infty,0)$. As a result, the model RH problem $S^{\infty}(\lambda)$ is converted into the RH problem on the $z$-plane, denoted as $S^{\infty}(z)$ which satisfies the following jump conditions~ {(see also Figure \ref{jumpforSinfty})}
    
	\begin{equation}\label{jump-conditions-S-z}
		S^{\infty}_+(z)=S^{\infty}_-(z)\begin{cases}
			\begin{aligned}
				&\begin{pmatrix}
					0 & e^{x\Omega_0+\Delta_0}\\
					e^{-x\Omega_0-\Delta_0}& 0
				\end{pmatrix}, &&z\in[-\eta_1^2,0],\\			
				&\begin{pmatrix}
					i & 0\\
					0& i
				\end{pmatrix}, &&z\in[-\eta_2^2,-\eta_1^2]\cup[-\eta_4^2,-\eta_3^2],\\
				&\begin{pmatrix}
					0 & e^{x\Omega_1+\Delta_1}\\
					e^{-x\Omega_1-\Delta_1}& 0
				\end{pmatrix}, &&z\in[-\eta_3^2,-\eta_2^2],\\
				&\begin{pmatrix}
					0 & 1\\
					1 & 0
				\end{pmatrix}, &&z\in[-\infty,-\eta_4^2].\\
			\end{aligned}	
		\end{cases}
	\end{equation}
	Furthermore, $S^{\infty}(z)\to\begin{pmatrix}
	    1&1
	\end{pmatrix}$ as $z\to\infty$. The last jump matrix in the jump condition (\ref{jump-conditions-S-z}) is generated by the symmetry $S^{\infty}(-\lambda)=S^{\infty}(\lambda)
	\begin{pmatrix}
		0&1\\
		1&0	
	\end{pmatrix}$, which also changes the diagonal jump matrices in the $\lambda$-plane into off-diagonal ones in the $z$-plane, see \cite{Zhao PhysicaD}.
    \begin{figure}[h]
		\centering
		\begin{tikzpicture}[>=latex]
			\draw[-,very thick] (-8.5,0) to (-6.5,0);
			\draw[-,very thick] (-6.5,0) to (-4.5,0);
			\draw[->,very thick] (-6.5,0) to (-5.5,0)node[black,above=1mm]{\small $\begin{pmatrix}
				0 & 1 \\
				1 & 0
				\end{pmatrix}$};
			\draw[red,-,very thick] (-4.5,0) to (-2.5,0);
               \draw[red,->,very thick](-4.5,0) to (-3.5,0)node[red,above=1mm]{\small $\begin{pmatrix}
				i & 0 \\
				0 & i
				\end{pmatrix}$};

			\filldraw[black] (-4.5,0) node[black,below=1mm]{$-\eta_{4}^2$} circle (1.5pt);
			\filldraw[black] (-2.5,0) node[black,below=1mm]{$-\eta_3^2$} circle (1.5pt);

			\draw[-,very thick,black] (-2.5,0) to (-0.5,0);
			\draw[->,very thick,black] (-2.5,0) to (-1.5,0) node[black,above=1mm]{\scalebox{0.65}{$\begin{pmatrix}
					0 & e^{x\Omega_1+\Delta_1}\\
					e^{-x\Omega_1-\Delta_1}& 0
				\end{pmatrix}$}};
			\filldraw[black] (-0.5,0) node[black,below=1mm]{$-\eta_2^2$} circle (1.5pt);

			\draw [very thick,red] (-0.5,0) to (1.5,0);

            \draw [->,very thick,red] (-0.5,0) to (0.5,0)node[red,above=1mm]{\small $\begin{pmatrix}
				i & 0 \\
				0 & i
				\end{pmatrix}$};

			\draw[-,very thick,lightgray,dashed] (3.5,0) to (5.5,0);

			\filldraw[black] (3.5,0) node[black,below=1mm]{$0$} circle (1.5pt);
			\draw[-,very thick,black] (1.5,0) to (3.5,0);
			\draw[->,very thick,black] (1.5,0) to (2.5,0) node[black,above=1mm]{\scalebox{0.65}{ $\begin{pmatrix}
					0 & e^{x\Omega_0+\Delta_0}\\
					e^{-x\Omega_0-\Delta_0}& 0
				\end{pmatrix}$}};
			\filldraw[black] (1.5,0) node[black,below=1mm]{$-\eta_1^2$} circle (1.5pt);

		\end{tikzpicture}
		\caption{{\protect\small
				{The jump contour for \( S^{\infty}(z) \) and the associated jump matrices.}}}
		\label{jumpforSinfty}
	\end{figure}
	\par
	\begin{figure}
		\centering
		\includegraphics[width=11cm]{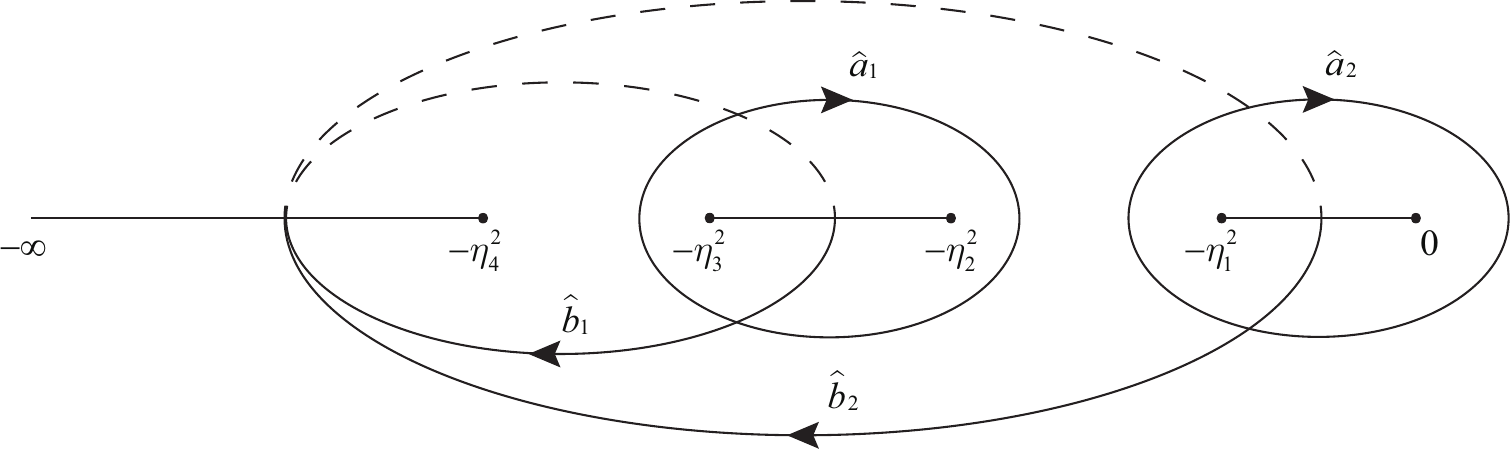}
		\caption{{\protect\small The  Riemann surface $\hat{\mathcal{S}}$ and its basis $\{\hat{a}_j ,\hat{b}_j\},~j=1,2$ of circles.}}
		\label{hat mathcal S}
	\end{figure}
	\par
	Introduce the Riemann surface of genus two (see Figure \ref{hat mathcal S}) as
	$$
	\hat{\mathcal{S}}:=\{(z,y)|y^2=z(z+\eta_1^2)(z+\eta_2^2)(z+\eta_3^2)(z+\eta_4^2)\},
	$$
	and define the cycle basis $\{\hat a_j,\hat b_j\}$ for $j=1,2$. Further, denote $\hat R(z)=\sqrt{z(z+\eta_1^2)(z+\eta_2^2)(z+\eta_3^2)(z+\eta_4^2)}$ and define the normalized holomorphic differential vector $\hat\omega=(\hat\omega_1, \hat\omega_2)^{T}$ such that $\oint_{a_j}\hat{\omega}_k=\delta_{jk}$ for $j=1,2$ where $\delta_{jk}=0$ and $1$ for $j\neq k$ and $j=k$ respectively. The associated period matrix $\hat\tau_{ij}=\oint_{\hat b_j}\hat \omega_i~(i,j=1,2)$ and $\hat{\tau}_j=\hat{\tau} e_j$, where $e_j$ is the $j$-th column of $2\times 2$ unit matrix for $j=1,2$. Thus the Abel map $\hat{J}(z)=\int_{\infty}^z \hat\omega$ has the following properties
	\begin{equation}
		\begin{aligned}
			&\hat{J}_+(z)+\hat{J}_-(z)=0,&&z\in(-\infty,\eta_4^2],\\
			&\hat{J}_+(z)-\hat{J}_-(z)=-e_1-e_2,&&z\in[-\eta_4^2,-\eta_3^2],\\
			&\hat{J}_+(z)+\hat{J}_-(z)=-\hat{\tau}_1,&&z\in[-\eta_3^2,-\eta_2^2],\\
			&\hat{J}_+(z)-\hat{J}_-(z)=-e_2,&&z\in[-\eta_2^2,-\eta_1^2],\\
			&\hat{J}_+(z)+\hat{J}_-(z)=-\hat{\tau}_2,&&z\in[-\eta_1^2,0].\\
		\end{aligned}
	\end{equation}
	Introduce the Riemann-Theta function:
	\begin{equation}\label{rsf}
		\Theta(z;\tau)=\sum_{\vec{n}\in\mathbb{Z}^2}e^{2\pi i\left(\vec{n}^T z+\frac{1}{2}\vec{n}^T\tau \vec{n}\right)},\quad z\in\mathbb{C}^2,
	\end{equation}
	which satisfies the following properties:
	\begin{equation}
		\Theta(z+e_j;\hat{\tau})=\Theta(z-e_j;\hat{\tau})=\Theta(z;\hat{\tau}),\quad
		\Theta(z+\hat{\tau}_j;\hat{\tau})=e^{2\pi i(-z_j-\hat{\tau}_{jj}/2)}\Theta(z;\hat{\tau}),
	\end{equation}
	Furthermore, define
	$$
	\hat\gamma(z)=\left(\frac{(z+\eta_1^2)(z+\eta_3^2)}{(z+\eta_2^2)(z+\eta_4^2)}\right)^{\frac{1}{4}}.
	$$
	Notice that the solution of model RH problem (\ref{jump-conditions-S-z}) on the $\lambda$-plane has at most fourth root singularities at {$\pm \eta_j,~j=1,\cdots,4$}, and this condition transforms into
	the solution of RH problem on the $z$-plane with singularities at $-\eta^2_{1}$ and $-\eta^2_{3}$ less than $1/4$. Furthermore, the zeros of $\hat{\gamma}(z)$ is $-\eta_1^2$ and $-\eta_3^2$ and the solution of the RH problem has the form
	\begin{equation}\label{Omee}
		\frac{\Theta(\pm \hat{J}(z)-\frac{\Omega}{2\pi i}+\hat{d};\hat{\tau})}{\Theta(\pm \hat{J}(z)+\hat{d};\hat{\tau})},\quad \text{with}\quad \Omega=(x\Omega_1+\Delta_1,x\Omega_0+\Delta_0)^T.
	\end{equation}
	It suffices to choose the parameter $\hat{d}$, which makes the zeros of denominator at $-\eta_1^2$ and $-\eta_3^2$, to nail down the singularities. On the other hand, we have
	$$
	\begin{aligned}
		&\hat{J}(\infty)=0,\ \hat{J}(-\eta_4^2)=-\frac{e_1+e_2}{2},\
		\hat{J}(-\eta_3^2)=\frac{\hat{\tau}_1}{2}-\frac{e_1+e_2}{2},\\
		&\hat{J}(-\eta_2^2)=\frac{\hat{\tau}_1}{2}-\frac{e_2}{2},\
		\hat{J}(-\eta_1^2)=\frac{\hat{\tau}_2}{2}-\frac{e_2}{2},\
		\hat{J}(0)=\frac{\hat{\tau}_1}{2}.
	\end{aligned}
	$$
	The Riemann constant $\hat{\mathcal{K}}=\frac{\hat{\tau}_1+\hat{\tau}_2}{2}-\frac{e_1}{2}$, and further $\hat{J}(-\eta_1^2)+\hat{J}(-\eta_3^2)=\hat{\mathcal{K}}$, which implies $\hat{d}=0$.
	Consequently the solution of the model RH problem (\ref{jump-conditions-S-z}) on the $z$-plane is
	\begin{equation}
		\hat{S}^{\infty}(z)=\hat\gamma(z)\frac{\Theta(0;\hat\tau)}{\Theta(\frac{\Omega}{2\pi i};\hat\tau)}\left(\begin{array}{cc}
			\frac{\Theta(\hat{J}(z)-\frac{\Omega}{2\pi i};\hat\tau)}{\Theta(\hat{J}(z);\hat\tau)} & \frac{\Theta(-\hat{J}(z)-\frac{\Omega}{2\pi i};\hat\tau)}{\Theta(-\hat{J}(z);\hat\tau)}
		\end{array}
		\right).
	\end{equation}
	
	In order to get the expansion of the Abel map $\hat{J}(z)$ as $z\to\infty$, we also need to transform the scalar RH problem for $g(\lambda)$ on the $\lambda$-plane into the case on the $z$-plane, that is
	$$
	\begin{aligned}
		&g_+(z)-g_-(z)=2i\sqrt{z}, &&z\in (-\eta_4^2,-\eta_3^2)\cup (-\eta_2^2,-\eta_1^2),\\
		&g_+(z)+g_-(z)=0, &&z\in(-\infty,-\eta_4^2),\\
		&g_+(z)+g_-(z)=\Omega_0, &&z\in(-\eta_1^2,0),\\
		&g_+(z)+g_-(z)=\Omega_1, &&z\in(-\eta_3^2,-\eta_2^2),\\
		&g(z)={\mathcal{O}\left(\frac{1}{\sqrt{z}}\right),} &&z\to\infty.
	\end{aligned}
	$$
	Consequently, the formula of $g(z)$ is
	$$
	g(z)=\hat{R}(z)\left(\int_{-\eta_3^2}^{-\eta_2^2}\frac{\Omega_1}{\hat R_+(\mu)(\mu-z)}d\mu+\int_{-\eta_1^2}^{0}\frac{\Omega_0}{\hat R_+(\mu)(\mu-z)}d\mu
	+\left(\int_{-\eta_4^2}^{-\eta_3^2}+\int_{-\eta_2^2}^{-\eta_1^2}\right)\frac{2i\sqrt{\mu}}{\hat R_+(\mu)(\mu-z)}d\mu\right).
	$$
	In order to keep consistence with the asymptotic condition, it follows that
	$$
	\int_{\hat a_1}{\Omega_1}{\hat\omega_j}+\int_{\hat a_2}{\Omega_0}{\hat\omega_j}+4i\left(\int_{-\eta_4^2}^{-\eta_3^2}+\int_{-\eta_2^2}^{-\eta_1^2}\right)
	\sqrt{z}{\hat\omega_j}=0,\quad j=1,2,
	$$
	and
	$4i\left(\int_{-\eta_4^2}^{-\eta_3^2}+\int_{-\eta_2^2}^{-\eta_1^2}\right)\sqrt{z}{\hat\omega_j}=2\pi  \mathrm{res}_{\infty}\sqrt{z}{\hat\omega_j}$ for $j=1,2$, which indicates that $\mathrm{res}_{\infty}\sqrt{z}{\hat{\omega}}=-\frac{1}{2\pi }\begin{pmatrix}
		\Omega_1 &\Omega_0
	\end{pmatrix}^T=-i\frac{\Omega_x}{2\pi i}$.
	As $z\to\infty$, expand $\hat{S}^{\infty}_1(z)$ as
	$$
	\hat{S}^{\infty}_1(z)=1+\frac{\hat{S}^{\infty}_{11}}{z^{1/2}}+\mathcal{O}\left(\frac{1}{z}\right),
	$$
	with
	\begin{equation}\label{S11}
		\hat{S}^{\infty}_{11}=-i\left[\nabla\log\left(\Theta\left(\frac{\Omega}{2\pi i}\right);\hat\tau\right)-\nabla\log(\Theta(0;\hat\tau))\right]\cdot \frac{ \Omega_x}{2\pi i}.	\end{equation}
	Recall that $z=-\lambda^2$, hence the WKB expansion of solution $\hat S_1^{\infty}(z)$ on the parameter $\lambda$ is
	\begin{equation}\label{WKB of k}
		\hat{S}^{\infty}_1(z)=1-\frac{1}{\lambda}\left[\nabla\log\left(\Theta\left(\frac{\Omega}{2\pi i}\right);\hat\tau\right)-\nabla\log(\Theta(0);\hat\tau)\right]\cdot \frac{ \Omega_x}{2\pi i}+\mathcal{O}\left(\frac{1}{\lambda^2}\right).
	\end{equation}
	Here we choose $i\lambda$, since the map $\lambda\to z=-\lambda^2$ from the upper $\lambda$-plane to $\C\setminus{(-\infty,0)}$.
	\par
	So far, we obtain the expression of $\hat{S}^{\infty}_1(z)$ and its expansion for {$z\to\infty$}. To reconstruct the potential $u(x)$, one needs to concentrate on the error RH problem of $Y(\lambda)$ which in fact contributes the error term $\mathcal{O}(x^{-1})$ in the asymptotic behavior of $u(x)$ for $x\to +\infty$. The techniques for error estimation are illustrated in \cite{Girotti CMP}, although they differ slightly from those in this paper. {Now, to analyze the error RH problem for $Y(\lambda)$, we need to transform the vector-form solution into a matrix form (see Remark \ref{errord}). To ensure logical coherence, we will reconstruct $u(x)$ at the end of subsection \ref{sec3.2} (Remark \ref{rmk3.9}) using equation (\ref{WKB of k}), leveraging the connection between the $z$-plane and the $\lambda$-plane. This allows us to complete the proof of Theorem \ref{thm3.8} from the perspective of the $z$-plane.}

	\begin{rmk}
		In fact, one can also construct the solutions to the model RH problem (\ref{Sinf}) in the $\lambda$-plane, but through the $z$-plane, it is seen that the solution corresponding to the model problem is a Riemann surface $\hat{\mathcal{S}}$ of genus two rather than the Riemann surface $\mathcal{S}$ of genus three in the $\lambda$-plane. For the general Jacobi map on a Riemann surface of genus three, denoted by $J(\lambda)$, they are typically represented as a three-dimensional vector, with the corresponding Riemann-Theta function $\Theta(J(\lambda)-d)$ owning three zeros on the Riemann surface, while the function $\check{\gamma}(\lambda)=\left(\frac{(\lambda^2-\eta_1^2)(\lambda^2-\eta_3^2)}{(\lambda^2-\eta_2^2)
			(\lambda^2-\eta_4^2)}\right)^{\frac{1}{4}}$ has four zeros. Therefore, it is impossible to construct a vector model solution that satisfies the singularity requirements at the branch points $\pm \eta_j,~j=1,2,3,4$.
		\par
		Regarding the matrix solution of the model problem in Ref. \cite{Teschl 2022}, the authors there discussed a similar matrix model RH problem of by constructing the Riemann surface of genus one for the KdV equation, proving that there is no entire matrix solution for the aforementioned matrix model problem.
	\end{rmk}

	\subsection{The solution of the model RH problem on the $\lambda$-plane}\label{sec3.2}
	
	We have derived the solution of model RH problem $\hat S^{\infty}(z)$ on the $z$-plane; nevertheless, we also need to get the solution on $\lambda$-plane. Moreover, notice that the  transformation $z=-\lambda^2$ can be considered as a holomorphic map between the Riemann surfaces $\mathcal{S}$ and $\hat{\mathcal{S}}$ in Figure \ref{Contour-2-Cavitation} and Figure \ref{hat mathcal S}. In general, suppose that $\mathcal{S}$ is a Riemann surface of genus $2g-1$ for $g\in \mathbb{Z}_{+}$, i.e., $\mathcal{S}=\{(\lambda,y)|y^2=\Pi_{j=1}^{2g}(\lambda^2-\eta_j^2)\},$ and Riemann surface $\hat{\mathcal{S}}$ on the $z$-plane, i.e., $\hat{\mathcal{S}}=\{(z,y)|y^2=z\Pi_{j=1}^{2g}(z+\eta_j^2)\}$.
	Define the holomorphic map $\varphi:\mathcal{S}\to\hat{\mathcal{S}}$ with $-\lambda^2\to z$, and according to the Riemann-Hurwitz formula, it follows that the genus of Riemann surface $\hat{\mathcal{S}}$ is $g$. Indeed, we just transform the solution on the $z$-plane into $\lambda$-plane, especially, the normalized holomorphic differentials, basic cycles and period matrix. In the following, we will directly illustrate the model solution of $S^{\infty}(\lambda)$ and then develop the equivalence between the two solutions.
	\par	
	Similarly, define the normalized holomorphic differentials $\omega_j~(j=1,2,3)$ associated with Riemann surface $\mathcal{S}$ by $\omega=(\omega_1, \omega_2, \omega_3)^T=A^{-1}\tilde{\omega}$, where $\tilde{\omega}=(\tilde{\omega}_1, \tilde{\omega}_2, \tilde{\omega}_3)^{T}$ with $\tilde\omega_j=\frac{\zeta^{j-1}d\zeta}{R(\zeta)}~(j=1,2,3)$
	and $A=(a_{ij})_{3\times3}$ with $a_{ij}=\oint_{a_j}\tilde \omega_i$. Define the period matrix $\tau=(\tau_{ij})_{3\times3}$ with $\tau_{ij}=\oint_{b_j}\omega_i$ and recall that the cycles are defined in Fig.\ref{Contour-2-Cavitation}. Indeed, it follows from the parity of $\tilde \omega_j$ and the symmetries of $a$-cycles that $a_{11}=a_{13},a_{21}=-a_{23}$ and $a_{31}=a_{33}$. Moreover, it is straightforward to check that
	\begin{equation}\label{expression of  omega}
		\omega_1=A_{11} \tilde \omega_1+A_{12}\tilde \omega_2+A_{13}\tilde \omega_3,\ \omega_2=A_{21} \tilde \omega_1+A_{23}\tilde \omega_3,\ \omega_3=A_{11} \tilde \omega_1-A_{12}\tilde \omega_2+A_{13}\tilde \omega_3,
	\end{equation}
	where $A_{i,j}$ is the $(i,j)$ element of the matrix $A^{-1}$. It follows that $\omega_1+\omega_3$ and $2\omega_2$ have some nice symmetries on the Riemann surface $\mathcal{S}$ and the period matrix $\tau$ has the following properties:
	\begin{equation}\label{properties of tau}
		\tau_{11}=\tau_{33},\tau_{12}=\tau_{23}.
	\end{equation}
	Now, define the Jacobi map
	\begin{equation}\label{biao}
		\check{J}(\lambda)=\int_{\eta_4}^{\lambda} \check\omega:=\int_{\eta_4}^{\lambda} \begin{pmatrix}
			\omega_1+\omega_3\\
			2\omega_2
		\end{pmatrix},
	\end{equation}
	and the corresponding period matrix is
	\begin{equation}\label{period matrix of hat tau}
		\check\tau=	\begin{pmatrix}
			\tau_{11}+\tau_{31}& \tau_{12}+\tau_{32}\\
			2\tau_{21}& 2\tau_{22}
		\end{pmatrix}.
	\end{equation}
	If one wants to use the above Jacobi map to construct the solution of model RH problem (\ref{Sinf}), it suffices to show that the imaginary part of $\check\tau$ is positive definite. In fact, it follows from the Riemann bilinear identity that $\tau_{ij}=\tau_{ji}$ and the properties of $\tau$ in (\ref{properties of tau}) that $\check\tau$ is a principal minor of $\tau$ up to congruent transformations, which indicates that the imaginary part of  $\check\tau$ is also positive definite. Consequently, the quotient space we choose is not $\C^3\setminus\{\tilde{e}_j,\tau_j\}$ for $j=1,2,3$ but $\C^2\setminus\{e_j,\check\tau_j\}$ for $j=1,2$, where $\tilde{e}_j$ is the $j$-th column of the $3\times 3$ identity matrix and $\check\tau_j$ is the $j$-th column of the matrix $\check{\tau}$. Indeed, we have
	\begin{equation}
		\begin{aligned}
			&\oint_{b_1}\check\omega=\check\tau_1,\ \oint_{b_2}\check\omega=\check\tau_2,\ \oint_{b_3} \check\omega= \begin{pmatrix}
				\tau_{13}+\tau_{33}\\
				2\tau_{23}
			\end{pmatrix}=\check\tau_1,\\
			&\oint_{a_1}\check\omega=e_1,\ \oint_{a_2}\check\omega=2e_2,\ \oint_{b_3} \check\omega=e_1.
		\end{aligned}
	\end{equation}
	Furthermore, it follows that the Jacobi map $\check{J}(z)$ satisfies
	\begin{equation}\label{Jacobi J jumps}
		\begin{aligned}
			&\check{J}_{+}(z)+\check{J}_{-}(z)=0, && z \in \Sigma_3, \\ &\check{J}_{+}(z)+\check{J}_{-}(z)=-e_1, && z \in \Sigma_1, \\ &\check{J}_{+}(z)+\check{J}_{-}(z)=-e_1-2 e_2, && z \in \Sigma_2, \\ &\check{J}_{+}(z)+\check{J}_{-}(z)=-2 e_1-2 e_2, && z \in \Sigma_4, \\ &\check{J}_{+}(z)-\check{J}_{-}(z)=-\check\tau_1, && z \in\left(\eta_2, \eta_3\right), \\
			&\check{J}_{+}(z)-\check{J}_{-}(z)=-\check\tau_2, && z \in\left(-\eta_1, \eta_1\right), \\
			&\check{J}_{+}(z)-\check{J}_{-}(z)=-\check\tau_1, && z \in\left(-\eta_3,-\eta_2\right),
		\end{aligned}
	\end{equation}
	and the Jacobi map on the branch points are half periods, i.e.,
	\begin{equation}\label{Jacobi J half periods}
		\begin{aligned}
			&\check{J}(\eta_4)=0,\ \check{J}(\eta_3)=-\frac{\check\tau_1}{2},\ \check{J}(\eta_2)=-\frac{\check\tau_1}{2}-\frac{e_1}{2},\ \check{J}(\eta_1)=-\frac{\check\tau_2}{2}-\frac{e_1}{2},\\
			&\check{J}(-\eta_4)=0,\
			\check{J}(-\eta_3)=-\frac{\check\tau_1}{2},\
			\check{J}(-\eta_2)=-\frac{\check\tau_1}{2}-\frac{e_1}{2},\
			\check{J}(-\eta_1)=-\frac{\check\tau_2}{2}-\frac{e_1}{2}.
		\end{aligned}
	\end{equation}
	On the other hand, based on the formula (\ref{biao}), for $\lambda\in\C\setminus\R$, one has
	\begin{equation}\label{symmetry of Jacobi}
		\check{J}(-\lambda)=\check{J}(\lambda)+e_1+e_2,
	\end{equation}
	which indicates that $\check{J}(\infty_+)=\frac{e_1+e_2}{2}$. Now, the solution of the model RH problem for $S^{\infty}(\lambda)$ on the $\lambda$-plane can be written down. Initially, suppose that  $\check{\gamma}(\lambda)=\left(\frac{(\lambda^2-\eta_1^2)(\lambda^2-\eta_3^2)}{(\lambda^2-\eta_2^2)
		(\lambda^2-\eta_4^2)}\right)^{\frac{1}{4}}$ and similarly assume that
	$$
	S^{\infty}(\lambda)=\check{c}\check{\gamma}(\lambda)\begin{pmatrix}
		\frac{\Theta(\check{J}(\lambda)-\check{d}+\frac{\Omega}{2\pi i};\check\tau )}{\Theta(\check{J}(\lambda)-\check{d};\check\tau )}&
		\frac{\Theta(-\check{J}(\lambda)-\check{d}+\frac{\Omega}{2\pi i};\check\tau )}{\Theta(-\check{J}(\lambda)-\check{d};\check\tau )}
	\end{pmatrix},
	$$
	where we demote the period matrix $\check{\tau}$ on the $\lambda$-plane comparing to $\hat\tau$ on the $z$-plane. In fact, it will be proven in the Lemma \ref{lem37} that $\check\tau=\hat\tau$. Moreover, the parameters $\check{d}$ and $\check{c}$ should be determined. Note that we require the solution of RH problem on the $\lambda$-plane to have at most fourth rootsingularities at branch points $\pm\eta_j$ for $j=1,\cdots,4$, which implies that the zeros of $\Theta(\check{J}(\lambda)-\check{d};\check\tau )$ and $\Theta(-\check{J}(\lambda)-\check{d};\check\tau )$ both only lies at $\pm\eta_{1}$ and $\pm\eta_{3}$. Thus, it follows from the zeros of the Riemann-Theta function are odd half periods that $\check{d}=\frac{e_1+e_2}{2}$, and combining with $\check{J}(\infty_+)=\frac{e_1+e_2}{2}$ shows that $\check{c}=\frac{\Theta(0;\check\tau)}{\Theta(\frac{\Omega}{2\pi i};\check\tau)}$. Recall that $\Omega=(x\Omega_1+\Delta_1,x\Omega_0+\Delta_0)^T$, thus $S^{\infty}(\lambda)$ exactly satisfies the jump conditions in (\ref{Sinf}).
	Now, we claim that the function $\Theta(\check{J}(\lambda))$ has precise four simple zeros on the Riemann surface $\mathcal{S}$, see \cite{Bertola Riemmansurface}.
	
	\begin{lem}\label{Lemma-4p}
		For arbitrary fixed $d_0\in\C^2$, define the function $\vartheta(\lambda):\mathcal{S}\to\C$ with $\lambda \mapsto\Theta(\check{J}(\lambda)-d_0;\check{\tau})$. We have $\deg(\vartheta)=4$, provided that $\vartheta$ does not vanish identically. Let $\mathcal{D}=(\vartheta)$, then $\check{J}(\mathcal D)=d_0-\check{\mathcal{K}}$, with $\check{\mathcal{K}}_k=\frac{\check\tau_{kk}}{2}-\left(\oint_{a_1}\check{J}(\lambda)(\omega_1+\omega_3)+\oint_{a_2}2\check{J}(\lambda)(\omega_2)+\oint_{a_3}\check{J}(\lambda)(\omega_1+\omega_3)\right)_k,k=1,2$, where the subscript ``$k$'' denotes the $k$-th element of the column vector.
	\end{lem}
	\begin{proof}
		Integrate $d\ln \vartheta$ along the boundary of the Riemann surface $\mathcal{S}$ denoted by $\delta \mathcal{S}$ as follows
		$$
		\begin{aligned}
			\frac{1}{2\pi i}\oint_{\delta\mathcal S}\frac{d\vartheta(\lambda)}{\vartheta(\lambda)}&=\frac{1}{2\pi i} \sum_{j=1}^{3}\left(\int_{\eta_4}^{\eta_4+a_j}+\int_{\eta_4+a_j}^{\eta_4+a_j+b_j}+\int_{\eta_4+a_j+b_j}^{\eta_4+b_j}+\int_{\eta_4+b_j}^{\eta_4}\right)d\ln \vartheta(\lambda)\\
			&=\frac{1}{2\pi i} \sum_{j=1}^{3}\left(\int_{\eta_4}^{\eta_4+a_j}-\int_{\eta_4+b_j}^{\eta_4+a_j+b_j}+\int_{\eta_4+b_j}^{\eta_4}-\int_{\eta_4+a_j+b_j}^{\eta_4+a_j}\right)d\ln \vartheta(\lambda)\\
			&=\int_{\eta_4}^{\eta_4+a_1}(\omega_1+\omega_3)+\int_{\eta_4}^{\eta_4+a_2}2\omega_2+\int_{\eta_4}^{\eta_4+a_3}(\omega_1+\omega_3)=4,
		\end{aligned}
		$$
		where we have used the fact that $d\ln \vartheta(\lambda+b_{1})=d\ln \vartheta(\lambda+b_{3})=-2\pi i d(\check{J}(\lambda))_1+d\ln \vartheta(\lambda)$, $d\ln  \vartheta(\lambda+b_2)=-2\pi i d(\check{J}(\lambda))_2+d\ln \vartheta(\lambda)$ and $d(\check{J}(\lambda))_1=\omega_1+\omega_3,d(\check{J}(\lambda))_2=2\omega_2$. Similar to the above computation, integrate $\check{J}(\lambda)d\ln \vartheta(\lambda)$ along the boundary of Riemann surface $\delta \mathcal{S}$ in the following:
		$$
		\begin{aligned}
			\frac{1}{2\pi i}\oint_{\delta\mathcal S}(\check{J}(\lambda))_kd\ln \vartheta(\lambda)&=\frac{1}{2\pi i} \sum_{j=1}^{3}\left(\int_{\eta_4}^{\eta_4+a_j}+\int_{\eta_4+a_j}^{\eta_4+a_j+b_j}+\int_{\eta_4+a_j+b_j}^{\eta_4+b_j}+\int_{\eta_4+b_j}^{\eta_4}\right)(\check{J}(\lambda))_kd\ln \vartheta(\lambda)\\
			&=\frac{1}{2\pi i} \sum_{j=1}^{3}\left(\int_{\eta_4}^{\eta_4+a_j}-\int_{\eta_4+b_j}^{\eta_4+a_j+b_j}+\int_{\eta_4+b_j}^{\eta_4}-\int_{\eta_4+a_j+b_j}^{\eta_4+a_j}\right)(\check{J}(\lambda))_kd\ln \vartheta(\lambda)\\
			&=\frac{1}{2\pi i} \sum_{j=1}^{3}\int_{\eta_4}^{\eta_4+a_j}\left((\check{J}(\lambda))_kd\ln \vartheta(\lambda)-((\check{J}(\lambda))_k+\hat\tau_{kj})(d\ln \vartheta(\lambda)-2\pi i d(\check{J}(\lambda))_j)\right)\\
			&\quad +\frac{1}{2\pi i} \sum_{j=1}^{3}\int_{\eta_4}^{\eta_4+a_j}\left((\check{J}(\lambda))_kd\ln \vartheta(\lambda)-((\check{J}(\lambda))_k+2\pi i \delta_{jk})d\ln \vartheta(\lambda)\right)\\
			&=\oint_{a_1}(\check{J}(\lambda))_k(\omega_1+\omega_3)+\oint_{a_2}2(\check{J}(\lambda))_k\omega_2+\oint_{a_3}
			(\check{J}(\lambda))_k(\omega_1+\omega_3)-\frac{\hat\tau_{kk}}{2}+(d_0)_k.
		\end{aligned}
		$$
	\end{proof}
	Finally, according to Lemma \ref{Lemma-4p}, it is immediate that both $\Theta(\check{J}(\lambda)-\check{d})$ and $\Theta(- \check{J}(\lambda)-\check{d})$ have four simple zeros, i.e., $\pm\eta_{1}$ and $\pm\eta_{3}$. So the solution to the model RH problem (\ref{Sinf}) is given by the following theorem.
	
	\begin{thm}
		Define $\check{J}(\lambda)=\int_{\eta_4}^{\lambda}\begin{pmatrix}
			\omega_1+\omega_3\\
			2\omega_2
		\end{pmatrix}$, where $\omega_j~(j=1,2,3)$ as defined previously are normalized holomorphic differentials on Riemann surface $\mathcal{S}$, then the corresponding period matrix $\check\tau$ is defined in (\ref{period matrix of hat tau}), and let $\check{d}=\frac{e_1+e_2}{2},\Omega=\begin{pmatrix}
			x\Omega_1+\Delta_1,x\Omega_0+\Delta_0
		\end{pmatrix}^T$, then the vector valued funtion
		\begin{equation}\label{solution of lambda}
			S^{\infty}(\lambda)=\check{\gamma}(\lambda)\frac{\Theta(0;\check\tau)}{\Theta(\frac{\Omega}{2\pi i};\check\tau)}
			\begin{pmatrix}
				\frac{\Theta(\check{J}(\lambda)-\check{d}+\frac{\Omega}{2\pi i};\check\tau)}{\Theta(\check{J}(\lambda)-\check{d};\check\tau)}&
				\frac{\Theta(-\check{J}(\lambda)-\check{d}+\frac{\Omega}{2\pi i};\check\tau)}{\Theta(-\check{J}(\lambda)-\check{d};\check\tau)}
			\end{pmatrix},	
		\end{equation}
		solves the RH problem (\ref{Sinf}).
	\end{thm}
	\begin{rmk}
		On the basis of calculation of $\Omega_j$, $j=0,1,2$ in (\ref{sysofOmega}), it implies that $\frac{\Omega}{2\pi i}$ is a real vector-valued function. Therefore, $S^{\infty}(\lambda)$ has no other singularities except for $\pm\eta_1$ and $\pm\eta_3$, since the zeros of the Riemann-Theta function is positioned at odd half periods.
	\end{rmk}
	\begin{rmk}\label{errord}
		To finish the proof of Theorem \ref{Spatial behavior}, one should consider the error estimation of the potential $u(x)$ for $x\to+\infty$. However, to keep the length of the paper manageable, we omit this step since the detailed discussions are made in \cite{Girotti CMP}. Still, it is necessary to introduce some notations to make our results complete.
		\par
		Consider the $1$-form $dp=\frac{\lambda^4+\alpha\lambda^2+\beta}{R(\lambda)}d\lambda$ and the Abelian integral $p(\lambda)=\int_{\eta_4}^{\lambda}dp$, which satisfied
		\begin{equation}
			\begin{aligned}
				&p_+(\lambda)+p_-(\lambda)=0,  		   && \lambda\in\Sigma_{1,2,3,4},   \\
				&p_+(\lambda)-p_-(\lambda)=\Omega_0,   && \lambda\in[-\eta_1,\eta_1],	\\
				&p_+(\lambda)-p_-(\lambda)=\Omega_1,   && \lambda\in[-\eta_3,-\eta_2]\cup[\eta_2,\eta_3],
			\end{aligned}
		\end{equation}
		and
		\begin{equation}
			p(-\lambda)=-p(\lambda), \quad \lambda\in\mathbb{C}\setminus\mathbb{R}.
		\end{equation}
		Further define
		\begin{equation}
			P^{\infty}(\lambda):=\frac{1}{2}
			\begin{pmatrix}
				(1+\frac{p(\lambda)}{\lambda})S_1^{\infty}+\frac{1}{\lambda}S_{1x}^{\infty}  &  (1-\frac{p(\lambda)}{\lambda})S_2^{\infty}+\frac{1}{\lambda}S_{2x}^{\infty}  \\
				(1-\frac{p(\lambda)}{\lambda})S_1^{\infty}-\frac{1}{\lambda}S_{1x}^{\infty}  &  (1+\frac{p(\lambda)}{\lambda})S_2^{\infty}-\frac{1}{\lambda}S_{2x}^{\infty}
			\end{pmatrix},
		\end{equation}
		which satisfies the following matrix-valued RH problem with same jump matrices as $S^{\infty}(\lambda)$ and
		\begin{equation}
			\begin{aligned}
				&P^{\infty}(\lambda)=\begin{pmatrix}
					1 & 0 \\
					0 & 1
				\end{pmatrix}  + \mathcal{O}\left(\frac{1}{\lambda}\right),\quad \lambda\to\infty,
				\\
				&P^{\infty}(\lambda) {\rm ~is~analytic~for~} \lambda\in\mathbb{C}\setminus[-\eta_4,\eta_4] {\rm ~with~a~singularity~at~} \lambda=0.
			\end{aligned}
		\end{equation}
		One can proof that $\det P^{\infty}(\lambda)\equiv 1$ for $\lambda\in\mathbb{C}$.
		\par
		Now, let the solution of the error vector-valued RH problem defined by
		\begin{equation}
			\mathcal{E}(\lambda)=S(\lambda)(P(\lambda))^{-1},
		\end{equation}
		where the global parametrix $P(\lambda)$ is given by
		\begin{equation}
			P(\lambda)=\begin{cases}
				P^{\infty}(\lambda),   &  \lambda\in\mathbb{C}\setminus\cup_{j=1,2,3,4}B_{\rho}^{\pm\eta_j}, \\
				P^{\diamond}(\lambda), &   \lambda\in B_{\rho}^{\diamond},
			\end{cases}
		\end{equation}
		where $\diamond$ traverses all $8$ branch points, i.e., $\pm\eta_j$, $j=1,2,3,4$ and $B_{\rho}^{\diamond}$ denotes the open disc of radius $\rho$ centered at $\diamond$. We claim that the local parametrix $P^{\pm\eta_j}$ near $\lambda=\pm\eta_j$, $j=1,2,3,4$ can be described by the modified Bessel function which will finally contribute the term $\mathcal{O}(x^{-1})$ in the asymptotic result of $u(x)$ for $x\to +\infty$.
		\par
		Therefore, one can see that
		\begin{equation}\label{YE1}
			Y(\lambda)=\mathcal{E}(\lambda)P(\lambda)e^{-xg(\lambda)\sigma_3}f(\lambda)^{-\sigma_3}
			=\left(\begin{pmatrix}1 & 1\end{pmatrix} + \frac{\mathcal{E}_1(x)}{x\lambda}
			+\mathcal{O}\left(\frac{1}{\lambda^2}\right)\right)P(\lambda)e^{-xg(\lambda)\sigma_3}f(\lambda)^{-\sigma_3}.
		\end{equation}
	\end{rmk}
	Recall that the potential $u(x)$ can be reconstructed from the solution of $Y(\lambda)$ as following
	$$
	u(x)=2 \frac{\mathrm{d}}{\mathrm{d} x}\left[\lim _{\lambda \rightarrow \infty} {\lambda}\left(Y_1(\lambda ; x)-1\right)\right].
	$$
    
	{\begin{thm}\label{thm3.8}
		As $x\to+\infty$, the potential function $u(x)$ is subject to the following asymptotic expression
		\begin{equation}
			u(x)=-\left(2\alpha+{\sum_{j=1}^4\eta_j^2}+2\partial_x^2\log\left(\Theta\left(\frac{\Omega}{2\pi i};\check\tau\right)\right)\right)+\mathcal{O}\left(\frac{1}{x}\right),
		\end{equation}
		where $\alpha$ is associated to a second kind Abel differential and can be calculated in Remark \ref{alpha}.
	\end{thm}
	\begin{proof}
		After a series of deformations of the original RH problem, we have
		$$
		Y_1(\lambda)=\left(S_1^{\infty}(x;\lambda)+\frac{(\mathcal{E}_1(x))_1}{x\lambda}
		+\mathcal{O}\left(\frac{1}{\lambda^2}\right)\right)\frac{e^{-xg(\lambda)}}{f(\lambda)},
		$$
		where $(\mathcal{E}_1(x))_1$ is the first entry of the vector $\mathcal{E}_1(x)$ given by equation (\ref{YE1}). Moreover, reminding the expression of $g(\lambda)$ in (\ref{gfunction}), it follows that
\begin{equation}\label{gfunction infty}
		e^{-xg(\lambda)}=1-\frac{x}{\lambda}\left(\alpha+\frac{\eta_1^2+\eta_2^2+\eta_3^2+\eta_4^2}{2}\right)+\mathcal{O}\left(\frac{1}{\lambda^2}\right).
\end{equation}
		Similarly, from the formula of $f(\lambda)$, one has
\begin{equation}\label{f infty}
	f(\lambda)=1+\frac{f_1}{\lambda}+\mathcal{O}\left(\frac{1}{\lambda^2}\right).
\end{equation}
		In addition, recall the linear equations in (\ref{sysofOmega}), and change $\tilde\omega_j$ into $\omega_j~(j=1,2,3)$, it follows from the Riemann Bilinear relations \cite{Bertola Riemmansurface} that
		$$
		2\pi i\ \sum \mathrm{res}_{\infty_{\pm}}\frac{\omega_1}{z}=\Omega_1,\ 2\pi i\ \sum \mathrm{res}_{\infty_{\pm}}\frac{\omega_2}{z}=\Omega_2,\ 2\pi i\ \sum \mathrm{res}_{\infty_{\pm}}\frac{\omega_3}{z}=\Omega_1.
		$$
		Incorporate with the symmetry in (\ref{symmetry of Jacobi}), which implies that as $\lambda\to\infty$
		$$
		J(\lambda)=\frac{e_1+e_2}{2}-\frac{\Omega_x}{\lambda}+\mathcal{O}\left(\frac{1}{\lambda^2}\right).
		$$
		Moreover, the expanding of $S_1^{\infty}(\lambda)$ as $\lambda\to\infty$ is
		$$
		\begin{aligned}		   {S}^{\infty}_1(\lambda)&=1-\frac{1}{\lambda}\left[\nabla\log\left(\Theta\left(\frac{\Omega}{2\pi i};\check\tau\right)\right)-\nabla\log(\Theta(0;\check\tau))\right]\cdot \frac{ \Omega_x}{2\pi i}+\mathcal{O}\left(\frac{1}{\lambda^2}\right),\\
			&=1-\frac{1}{\lambda}\partial_x\log\left(\Theta\left(\frac{\Omega}{2\pi i};\check\tau\right)\right)+\mathcal{O}\left(\frac{1}{\lambda^2}\right).
		\end{aligned}
		$$
		Thus, it is obtained that
		$$
		Y_1(\lambda)=1-\frac{1}{\lambda}\left[f_1+x\left(\alpha+\frac{1}{2}\sum_{j=1}^4\eta_j^2\right)+\partial_x\log\left(\Theta\left(\frac{\Omega}{2\pi i};\check\tau\right)\right)+\frac{(\mathcal{E}_1(x))_1}{x}\right]+\mathcal{O}\left(\frac{1}{\lambda^2}\right).
		$$
		Consequently, based on the relationship between $u(x)$ and $Y(\lambda)$, we have
		$$
		u(x)=-\left(2\alpha+{\sum_{j=1}^4\eta_j^2}+2\partial_x^2\log\left(\Theta\left(\frac{\Omega}{2\pi i};\check\tau\right)\right)\right)
		+\mathcal{O}\left(\frac{1}{x}\right).
		$$
	\end{proof}
\begin{rmk}\label{rmk3.9}
	We now complete the proof of Theorem 3.8 on the $z$-plane as follows. Similar to the proof in Theorem \ref{thm3.8}, we have 
	\begin{equation}
		Y_1(\lambda)= \left(\hat S_1^{\infty}(x;-\lambda^2)+\frac{(\mathcal{E}_1(x))_1}{x\lambda}+\mathcal{O}\left(\frac{1}{\lambda^2}\right)\right)\frac{e^{-xg(\lambda)}}{f(\lambda)},
	\end{equation}
	where $S_1^{\infty}(x;-\lambda^2)$ is given by (\ref{WKB of k}). Consequently, combining (\ref{WKB of k}), (\ref{gfunction infty}) with (\ref{f infty}), it gives that
    \begin{equation}\label{Y1 infty}
    \begin{aligned}
		Y_1(\lambda)
        &=1-\frac{1}{\lambda}\left(f_1+x\left(\alpha+\frac{\sum_{j=1}^4\eta_j^2}{2}\right)\right.\\
        &\left.+\left(\nabla\log\left(\Theta\left(\frac{\Omega}{2\pi i}\right);\hat\tau\right)-\nabla\log(\Theta(0;\hat\tau))\right)\cdot \frac{ \Omega_x}{2\pi i}
		+\frac{(\mathcal{E}_1(x))_1}{x}\right)+\mathcal{O}\left(\frac{1}{\lambda^2}\right),
        \end{aligned}
	\end{equation}
	which indicates that for $x\to +\infty$ the genus two KdV soliton gas potential behaves
	\begin{equation}\label{u potential}
		u(x)=-\left(2\alpha+{\sum_{j=1}^4\eta_j^2}+2\partial_x^2\log\left(\Theta\left(\frac{\Omega}{2\pi i};\hat\tau\right)\right)\right)+\mathcal{O}\left(\frac{1}{x}\right).
	\end{equation}
\end{rmk}}
	So far, we almost complete the construction of the soliton gas potential in the regime $x\to+\infty$. However, there is a minor flaw, which lies in proving the equivalence of the period matrices $\check\tau$ and $\hat\tau$.
	\begin{lem}\label{lem37}
		Suppose that $\omega_j$, $j=1,2,3$ defined previously are the normalized holomorphic differentials on Riemann surface $\mathcal{S}$, and then $\omega_1+\omega_3,2\omega_2$ are the normalized holomorphic differentials on the transformed Riemann surface $\hat{\mathcal{S}}$. Moreover, the period matrix $\check\tau$ is precisely $\hat\tau$.
	\end{lem}
	\begin{proof}
		Recalling the definition of $\tilde \omega_j$ and changing the variable $\lambda$ into $i\sqrt{z}$, we have the following local expressions
		$$
		\tilde \omega_1=\frac{d\lambda}{R(\lambda)}=\frac{i}{2}\frac{dz}{\hat R(z)},\ \tilde \omega_2=\frac{\lambda d\lambda}{R(\lambda)}=-\frac{1}{2}\frac{\sqrt{z}dz}{\hat R(z)},\
		\tilde \omega_3=\frac{\lambda^2 d\lambda}{R(\lambda)}=-\frac{i}{2}\frac{z dz}{\hat R(z)}.
		$$
		It follows that after the holomorphic transformation $z=-\lambda^2$, the original holomorphic differential $\tilde \omega_{2}$ is not a holomorphic differential in $\hat{\mathcal{S}}$. Fortunately, by the expressions of $\omega_j$ in terms of $\tilde{\omega}_j$ in (\ref{expression of  omega}), $\omega_1+\omega_2$ and $2\omega_{2}$ are irrelated to $\tilde{\omega}_2$. In particular,
		$$
		\begin{aligned}
			&\oint_{\hat a_1}(\omega_1+\omega_3)\rvert_{\hat{\mathcal{S}}}=\oint_{ a_1}(\omega_1+\omega_3)\rvert_{\mathcal{S}}=\oint_{ a_3}(\omega_1+\omega_3)\rvert_{\mathcal{S}}=1,
			\quad
			\oint_{\hat a_2}2\omega_2\rvert_{\hat{\mathcal{S}}}=\oint_{ a_2}\omega_2\rvert_{\mathcal{S}}=1,\\
			&\oint_{\hat a_1}2\omega_2\rvert_{\hat{\mathcal{S}}}=\oint_{ a_1}2\omega_2\rvert_{\mathcal{S}}=\oint_{ a_3}2\omega_2\rvert_{\mathcal{S}}=0,
			\quad
			\oint_{\hat a_2}(\omega_1+\omega_3)\rvert_{\hat{\mathcal{S}}}=\frac{1}{2}\oint_{ a_1}(\omega_1+\omega_3)\rvert_{\mathcal{S}}=0.
		\end{aligned}
		$$
		By the Riemann Bilinear relationship, it indicates that $(\omega_1+\omega_3)\rvert_{\hat{\mathcal{S}}}$ and $2\omega_2\rvert_{\hat{\mathcal{S}}}$ are the normalized holomorphic differentials on $\hat {\mathcal{S}}$ and the corresponding period matrix are equivalence.	
	\end{proof}
	\subsection{Behavior of the genus two KdV soliton gas potential for $x\to-\infty$}
	The behavior of the genus two KdV soliton gas potential is quite simple since the jump matrices of the RH problem $Y(\lambda)$ are all converge to identity matrix exponentially when $x\to -\infty$. To be precise, there exist a fixed constant $c\in\mathbb{R}_+$ such that for $x\to-\infty$, the boundary behavior of the potential $u(x)$ is given by
	\begin{equation}
		u(x)=\mathcal{O}(e^{-c|x|}).
	\end{equation}
	
	\section{Evolution of the genus two KdV soliton gas potential for $t\to+\infty$}\label{long time section}
	
	This section examines the large-time asymptotic behaviors of the genus two KdV soliton gas potential constructed in Section \ref{Soliton gas RH problem} and presents a detailed proof of the Theorem \ref{long time asymp}.
	\par
	If the genus two KdV soliton gas potential $u(x,0)$ evolves in time according to the KdV equation, the reflection coefficient $r(t;\lambda)=r(\lambda)e^{-8\lambda^3t}$. It follows that the RH problem of $Y(\lambda)=Y(\lambda; x,t)$ for the evolution of soliton gas is defined by
	\begin{equation}\label{RHP with t}
		\begin{aligned}
			&Y_{+}( \lambda)=Y_{-}( \lambda) \begin{cases}{\left(\begin{array}{cc}
						1 &  -i r( \lambda) e^{2 \lambda x-8\lambda^3t} \\
						0 & 1
					\end{array}\right)}, & \lambda \in \Sigma_{1,3}, \\
				{\left(\begin{array}{cc}
						1 & 0 \\
						i r( \lambda) e^{-2 \lambda x+8\lambda^3t} & 1
					\end{array}\right)}, & \lambda \in \Sigma_{2,4},\end{cases}\\
		\end{aligned}
	\end{equation}
	with the same boundary condition and symmetry in accordance with the case of the initial soliton gas potential, i.e.,
	$Y(\lambda)=(1\quad 1)+\mathcal{O}\left(\frac{1}{\lambda}\right),$ and $
	Y(-\lambda)=Y(\lambda)\left(\begin{matrix}
		0&1\\
		1&0
	\end{matrix}\right)$, respectively. In order to analyze the asymptotic behavior of $Y(x,t;\lambda)$ in the long time sense, rewrite the phase function $2\lambda x-8\lambda^3 t = 8\lambda t( \xi -\lambda^2)$ with $\xi=\frac{x}{4t}\in\R$.
	\par
	It is immediate that in the case $\xi<\eta_1^2$, the phase functions in the jump matrices are exponentially decreasing as $t\to+\infty$. Consequently, straightforward calculation shows that
	$$
	Y(\lambda)=(1\quad 1)+\mathcal{O}\left(e^{8\eta_1t(\xi-\eta_1^2)}\right),
	$$
	as $t\to+\infty$ with $\xi<\eta_1^2$, which indicates that the potential $u(x,t)$ vanishes rapidly in this region.
	
	\subsection{Modulated one-phase wave region} \label{Modulated-4-1}
	Introduce the critical value $\xi_{crit}^{(1)}$ that is defined by (\ref{xi critical}) below, and consider the constraint
	$$
	\eta_1^2<\xi<\xi_{crit}^{(1)}.
	$$
	\par
	{Refine the contours \(\Sigma_{1}\) and \(\Sigma_2\) as illustrated in Figure \ref{Sigma-alpha1}, by splitting them as follows:  
\begin{equation}  
    \Sigma_{1_{\alpha_1}} = (\eta_1, \alpha_1), \quad \text{and} \quad \Sigma_{2_{\alpha_1}} = (-\alpha_1, -\eta_1),  
\end{equation}  
where \(\alpha_1\) is a function of \(\xi\) to be determined in (\ref{alpha1 formular}) below.} 
\begin{figure}[h!]
    \centering
    \begin{tikzpicture}
    [>=latex]

        \draw[lightgray,very thick,dashed] (-7.5,0) to (-6.5,0);
        
        \filldraw[black] (-6.5,0) node[black,below=1mm]{$-\eta_{4}$} circle (1.5pt);
        \draw[very thick,dashed] (-6.5,0) to (-4.5,0);
        \filldraw[black] (-4.5,0) node[black,below=1mm]{$-\eta_{3}$} circle (1.5pt);
        
        \draw[-,very thick,dashed,lightgray] (-4.5,0) to (-3.5,0);
        \filldraw[black] (-3.5,0) node[black,below=1mm]{$-\eta_2$} circle (1.5pt);
        
        \draw[-,very thick,dashed] (-3.5,0) to (-2,0);
        \draw[-,very thick] (-2,0) to (-0.5,0);
        \draw[->,very thick] (-2,0) to (-1,0) node[black,above=0.5mm]{\small $\Sigma_{2_{\alpha_1}}$};
        \filldraw[black] (-0.5,0) node[black,below=1mm]{$-\eta_1$} circle (1.5pt);
        \filldraw[black] (-2,0) node[black,below=1mm]{$-\alpha_1$} circle (1.5pt);
        
        \draw [lightgray,dashed,very thick] (-0.5,0) to (0.5,0);
        
        \filldraw[black] (0.5,0) node[black,below=1mm]{$\eta_1$} circle (1.5pt);
        \draw[-,very thick,black] (0.5,0) to (2,0);
        \draw[->,very thick,black] (0.5,0) to (1.5,0)node[black,above=0.5mm]{\small $\Sigma_{1_{\alpha_1}}$};
         \filldraw[black] (2,0) node[black,below=1mm]{$\alpha_1$} circle (1.5pt);
        \draw[-,very thick,black,dashed] (2,0) to (3.5,0);
        \filldraw[black] (3.5,0) node[black,below=1mm]{$\eta_{2}$} circle (1.5pt);
        
        \draw[-,dashed,very thick,lightgray] (3.5,0) to (4.5,0);
        \filldraw[black] (4.5,0) node[black,below=1mm]{$\eta_{3}$} circle (1.5pt);
        
        \draw[very thick,dashed] (4.5,0) to (6.5,0);
        \filldraw[black] (6.5,0) node[black,below=1mm]{$\eta_{4}$} circle (1.5pt);
        
        \draw[lightgray,very thick,dashed] (6.5,0) to (7.5,0);
    \end{tikzpicture}
    \caption{{The solid lines represent $\Sigma_{1_{\alpha_1}}$ and $\Sigma_{2_{\alpha_1}}$, where $\eta_1 < \alpha_1 < \eta_2$.}}
    \label{Sigma-alpha1}
\end{figure}

    Similarly, introduce the $g$ function, denoted as $g_{\alpha_1}$, which satisfies the following scalar RH problem:
	\begin{equation}\label{g1 jump}
		\begin{aligned}
			&g_{\alpha_{1,+}}(\lambda)+g_{\alpha_{1,-}}(\lambda)+8\lambda^3-8\xi\lambda=0, &&\lambda\in \Sigma_{1_{\alpha_{1}}}\cup \Sigma_{2_{\alpha_{1}}},\\
			&g_{\alpha_{1,+}}(\lambda)-g_{\alpha_{1,-}}(\lambda)={\Omega_{\alpha_{1}}}, &&\lambda\in [-\eta_1,\eta_1],\\
			&g_{\alpha_{1}}(\lambda)=\mathcal{O}\left(\frac{1}{\lambda}\right), && \lambda\to \infty.
		\end{aligned}
	\end{equation}
	To further deform the RH problem (\ref{RHP with t}), it is required that the $g_{\alpha_1}$ function satisfies the following properties:
	\begin{equation}\label{condition for g1}
		\begin{aligned}
			&g_{\alpha_1}(\lambda)+4\lambda^3-4\xi\lambda=(\lambda\pm\alpha_{1})^{\frac{3}{2}}, &&\lambda\to\pm\alpha_{1},\\
			&\re\left(g_{\alpha_{1}}(\lambda)+4\lambda^3-4\xi\lambda\right)>0, &&\lambda \in (\alpha_{1},\eta_2)\cup (\eta_3,\eta_4),\\
			&\re\left(g_{\alpha_{1}}(\lambda)+4\lambda^3-4\xi\lambda\right)<0, &&\lambda \in (-\eta_2,-\alpha_{1})\cup (-\eta_4,-\eta_3),\\
			&-i(g_{\alpha_1,+}(\lambda)-g_{\alpha_1,-}(\lambda)) \text{ is real-valued and monotonically increasing},&& \lambda\in\Sigma_{1_{\alpha_{1}}},\\
			&-i(g_{\alpha_1,+}(\lambda)-g_{\alpha_1,-}(\lambda)) \text{ is real-valued and monotonically decreasing}, && \lambda\in\Sigma_{2_{\alpha_{1}}}.\\
		\end{aligned}
	\end{equation}
	The \(g_{\alpha_1}\) can be derived from its derivative \( g'_{\alpha_1} \) according to the uniqueness of \( g_{\alpha_1} \). Here the \( g'_{\alpha_1}d\lambda \) can be viewed as the second kind Abel differential on Riemann surface of genus one.
	\par	
	Define
	\begin{equation}\label{g1 derivative}		g'_{\alpha_1}(\lambda)=-12\lambda^2+4\xi+12\frac{Q_{\alpha_1,2}(\lambda)}{R_{\alpha_1}(\lambda)}-4\xi\frac{Q_{\alpha_1,1}(\lambda)}{R_{\alpha_1}(\lambda)},
	\end{equation}
	with
	$$
	R_{\alpha_1}(\lambda)=\sqrt{(\lambda^2-\eta_1^2)(\lambda^2-\alpha_{1}^2)},
	$$
	which is analytic for $\C\setminus\left(\Sigma_{1_{\alpha_1}}\cup\Sigma_{2_{\alpha_1}}\right)$, and takes positive real value for $\lambda>\alpha_{1}$.
	From the definition of $g_{\alpha_1}'(\lambda)$ in (\ref{g1 derivative}), we can derive the expression of $g_{\alpha_1}(\lambda)$:
	\begin{equation}\label{g1 formular}
		g_{\alpha_{1}}(\lambda)=-4\lambda^3+4\xi\lambda+12\int_{\alpha_{1}}^{\lambda}\frac{Q_{\alpha_1,2}(\zeta)}{R_{\alpha_1}(\zeta)}d\zeta-4\xi\int_{\alpha_{1}}^{\lambda}\frac{Q_{\alpha_1,1}(\zeta)}{R_{\alpha_1}(\zeta)}d\zeta.
	\end{equation}
	On the other hand, suppose that
	\begin{equation}\label{Q alpha 1}
		Q_{\alpha_1,1}(\lambda)=\lambda^2+c_{\alpha_1,1},\quad Q_{\alpha_1,2}(\lambda)=\lambda^4-\frac{1}{2}(\alpha_1^2+\eta_1^2)\lambda^2+c_{\alpha_1,2},
	\end{equation}
	where
	\begin{equation}\label{c alpha 1}
		c_{\alpha_1,1}=-\alpha_1^2+\alpha_1^2\frac{E(m_{\alpha_1})}{K(m_{\alpha_1})},\quad c_{\alpha_1,2}=\frac{1}{3} \alpha_1^2\eta_1^2+\frac{1}{6}(\alpha_1^2+\eta_1^2)c_{\alpha_1,1},\quad m_{\alpha_1}=\frac{\eta_1}{\alpha_1}.
	\end{equation}
	Here $K(m_{\alpha_1})$ and $E(m_{\alpha_1})$ are the first and the second kind complete elliptic integral, respectively, i.e., $K(m)=\int_{0}^{\frac{\pi}{2}} \frac{d\vartheta}{\sqrt{1-m^2\sin \vartheta^2}}$ and $E(m)=\int_{0}^{\frac{\pi}{2}} {\sqrt{1-m^2\sin \vartheta^2}}{d\vartheta}$. The $Q_{\alpha_1,1}$ and $Q_{\alpha_1,2}$ are determined by the conditions in (\ref{g1 jump}), and the first property in (\ref{condition for g1}) about the behavior near $\pm\alpha_{1}$ implies that the parameter $\alpha_1$ is determined by
	\begin{equation}\label{alpha1 formular}
		\xi=3\frac{Q_{\alpha_1,2}(\pm \alpha_{1})}{Q_{\alpha_1,1}(\pm \alpha_{1})}=\frac{1}{2}(\alpha_{1}^2+\eta_1^2)+(\alpha_{1}^2-\eta_1^2)\frac{K(m_{\alpha_1})}{E(m_{\alpha_1})},
	\end{equation}
	which states that the $\alpha_1$ is modulated by $\xi$. Indeed, rewrite (\ref{alpha1 formular}) as Whitham velocity
	\begin{equation}\label{Whitham-velocity} \xi=\frac{x}{4t}=\frac{\eta_1^2}{2}\left[1+\frac{1}{m_{\alpha_1}^2}+2\left(\frac{1}{m_{\alpha_1}^2}-1\right)\frac{K(m_{\alpha_1})}{E(m_{\alpha_1})}\right]:=\frac{\eta_1^2}{2}W(m_{\alpha_1}).
	\end{equation}
	It follows that $\partial_{\alpha_1}W(m_{\alpha_1})>0$ for $\eta_1<\alpha_{1}<+\infty$, since the Whitham equation associated with the Whitham velocity (\ref{Whitham-velocity}) is strictly hyperbolic {\cite{Lev88}}. {Alternatively, one obtains the inequality $\partial_{m_{\alpha_{1}}}W(m_{\alpha_{1}})<0$ by direct calculation: 
	$$
    \begin{aligned}
        \partial_{m_{\alpha_1}} W(m_{\alpha_1})
        = & -\frac{2}{m_{\alpha_1}^3} - \frac{4K(m_{\alpha_1})}{m_{\alpha_1}^3 E(m_{\alpha_1})} \\
          & + 2\frac{1 - m_{\alpha_1}^2}{m_{\alpha_1}^2}
            \frac{
                \left( \frac{E(m_{\alpha_1})}{m_{\alpha_1}(1 - m_{\alpha_1}^2)} - \frac{K(m_{\alpha_1})}{m_{\alpha_1}} \right) E(m_{\alpha_1})
                - \frac{E(m_{\alpha_1}) - K(m_{\alpha_1})}{m_{\alpha_1}} K(m_{\alpha_1})
            }{E(m_{\alpha_1})^2} \\
        = & \frac{2K(m_{\alpha_1}) \left[ (1 - m_{\alpha_1}^2) K(m_{\alpha_1}) - 2(2 - m_{\alpha_1}^2) E(m_{\alpha_1}) \right]}{E(m_{\alpha_1})^2 m_{\alpha_1}^3} 
        < 0,
    \end{aligned}
$$
	where the first equality follows from
	\begin{equation}
		K'(m)=\frac{E(m)}{m(1-m^2)}-\frac{K(m)}{m},\quad E'(m)=\frac{E(m)-K(m)}{m},
		\end{equation}
	and the last inequality is given by $K(m)<\frac{E(m)}{\sqrt{1-m^2}}$. Therefore the chain rule gives
	\begin{equation}
		\partial_{\alpha_1}W(m_{\alpha_1})=\frac{\eta_1^2}{2}\partial_{m_{\alpha_{1}}}W(m_{\alpha_{1}})\partial_{\alpha_1}m_{\alpha_1}>0,\quad \eta_1<\alpha_1<+\infty.
		\end{equation}}
     According to the expansions of elliptic functions, it is immediate that
	$$
	\begin{aligned}
		\lim_{\alpha_1\to\eta_1}\xi=\eta_1^2, \ m_{\alpha_1}\to1,\ \text{and}\ \lim_{\alpha_1\to+\infty}\xi=+\infty, \ m_{\alpha_1}\to0.
	\end{aligned}
	$$
	Define
	\begin{equation}\label{xi critical}
		\xi_{crit}^{(1)}:=3\frac{Q_{\alpha_1,2}(\eta_2)}{Q_{\alpha_1,1}(\eta_2)}=\frac{1}{2}(\eta_1^2+\eta_2^2)+(\eta_2^2-\eta_1^2)\frac{K(m_{\eta_{2}})}{E(m_{\eta_{2}})},
	\end{equation}
	where $m_{\eta_{2}}=\frac{\eta_1}{\eta_2}$. Consequently, (\ref{alpha1 formular}) defines $\alpha_{1}$ as a monotone increasing function of $\xi$ in $[\eta_1^2, \xi_{crit}^{(1)}]$ by the implicit function theorem. 
{\begin{rmk}
Whitham modulation theory serves as a powerful tool for analyzing time evolution, particularly in describing the temporal development of rarefaction waves and dispersive shock waves. The theory inherently operates on the spectral plane, which aligns closely with the framework of Riemann-Hilbert problems. Here, we provide an example where Whitham modulation theory is applied to analyze time evolution. For solutions requiring a description via Riemann surfaces, two branch points ($\pm\eta_1$) remain fixed while the other two ($\pm\alpha_1$) evolve dynamically over time — a scenario naturally suited to Whitham modulation theory for characterizing the implicit dynamical dependence between $\alpha_1$ and $\xi$. In Section \ref{modulated 2-genus case}, Whitham modulation theory also plays a pivotal role in characterizing the dynamics of $\pm\alpha_2$ on a genus-two Riemann surface (with branch points $\pm\eta_1$, $\pm\eta_2$ and $\pm\eta_3$ fixed).
\end{rmk}}
\par  
    If the relationship (\ref{alpha1 formular}) is established, we can verify the first condition in (\ref{condition for g1}). Indeed, together with (\ref{c alpha 1}), rewrite the function $g'_{\alpha_{1}}(\lambda)$ as
	$$
	\begin{aligned}
		g'_{\alpha_1}(\lambda)&=-12\lambda^2+4\xi+12\frac{Q_{\alpha_1,2}(\lambda)-Q_{\alpha_1,2}(\alpha_1)}{R_{\alpha_1}(\lambda)}-4\xi\frac{Q_{\alpha_1,1}(\lambda)-Q_{\alpha_1,1}(\alpha_1)}{R_{\alpha_1}(\lambda)}\\
		&=-12\lambda^2+4\xi+12\frac{(\lambda^2-\alpha_1^2)}{R_{\alpha_{1}}(\lambda)}\left[\lambda^2-\left(\frac{\eta_1^2-\alpha_{1}^2}{2}+\frac{\xi}{3}\right)\right].
	\end{aligned}
	$$
	It follows that $g'_{\alpha_1}(\lambda)$ behaves like $(\lambda\pm\alpha_{1})^{\frac{1}{2}}$ as $\lambda\to\pm \alpha_{1}$, and $(\lambda\pm\eta_{1})^{-\frac{1}{2}}$ as $\lambda\to\pm \eta_{1}$, which is coincided with (\ref{condition for g1}).
	Before verifying the other conditions in (\ref{condition for g1}), we will determine ${\Omega}_{\alpha_{1}}$ in the jump condition (\ref{g1 jump}) and introduce the following Riemann surface of genus one related to $R_{\alpha_1}(\lambda)$:
	$$
	\mathcal{S}_{\alpha_1}:=\{(\lambda,\eta)|\eta^2=(\lambda^2-\eta_1^2)(\lambda^2-\alpha_1^2)\},
	$$
	with the $A$, $B$-cycles defined in Figure \ref{1genus}. Define
	$$
	\omega_{\alpha_1}=-\frac{ \alpha_1 d\zeta}{4K(m_{\alpha_1})R_{\alpha_1}(\zeta)},
	$$
	as the normalized holomorphic differential on $\mathcal{S}_{\alpha_1}$ with $\oint_{A}\omega_{\alpha_1}=1$ and $\oint_{B}\omega_{\alpha_1}=\tau_{\alpha_1}$. Reminding the definition of $g_{\alpha_{1}}(\lambda)$ and jump conditions in (\ref{g1 jump}), it implies that
	\begin{equation}\label{tilde Omega equations}
		\begin{aligned}
			&24\int_{\alpha_{1}}^{\eta_1}\frac{Q_{\alpha_1,2}(\zeta)}{R_{\alpha_1}(\zeta)}d\zeta-8\xi\int_{\alpha_{1}}^{
				\eta_1}\frac{Q_{\alpha_1,1}(\zeta)}{R_{\alpha_1}(\zeta)}d\zeta={\Omega}_{\alpha_{1}},\\
			&\int_{-\eta_1}^{\eta_1}\frac{Q_{\alpha_1,2}(\zeta)}{R_{\alpha_1}(\zeta)}d\zeta=\int_{-\eta_1}^{\eta_1}\frac{Q_{\alpha_1,1}(\zeta)}{R_{\alpha_1}(\zeta)}d\zeta=0.\\
		\end{aligned}
	\end{equation}
	In fact, $\frac{Q_{\alpha_1,2}(\lambda)}{R_{\alpha_1}(\lambda)}d\lambda$ and $\frac{Q_{\alpha_1,1}(\lambda)}{R_{\alpha_1}(\lambda)}d\lambda$ are the normalized Abel differentials of the second kind and by using the {Riemann Bilinear relations} \cite{Bertola Riemmansurface}, it follows from (\ref{tilde Omega equations}) that
	\begin{equation}\label{tilde Omega}
		\begin{aligned}
			{\Omega}_{\alpha_{1}}&=-12\int_{B}\frac{Q_{\alpha_1,2}(\zeta)}{R_{\alpha_1}(\zeta)}d\zeta+4\xi\int_{B}\frac{Q_{\alpha_1,1}(\zeta)}{R_{\alpha_1}(\zeta)}d\zeta\\
			&=2\pi i (4\xi \mathrm{res}_{\lambda=\pm \infty}\lambda^{-1} \omega_{\alpha_1}-4 \mathrm{res}_{\lambda=\pm \infty}\lambda^{-3} \omega_{\alpha_1})
			=2\pi i \alpha_1 \frac{\alpha_1^2+\eta_{1}^2-2\xi}{K(m_{\alpha_{1}})}\in i\R.
		\end{aligned}
	\end{equation}
	\begin{figure}[!h]
		\centering
		\begin{overpic}[width=0.7\textwidth]{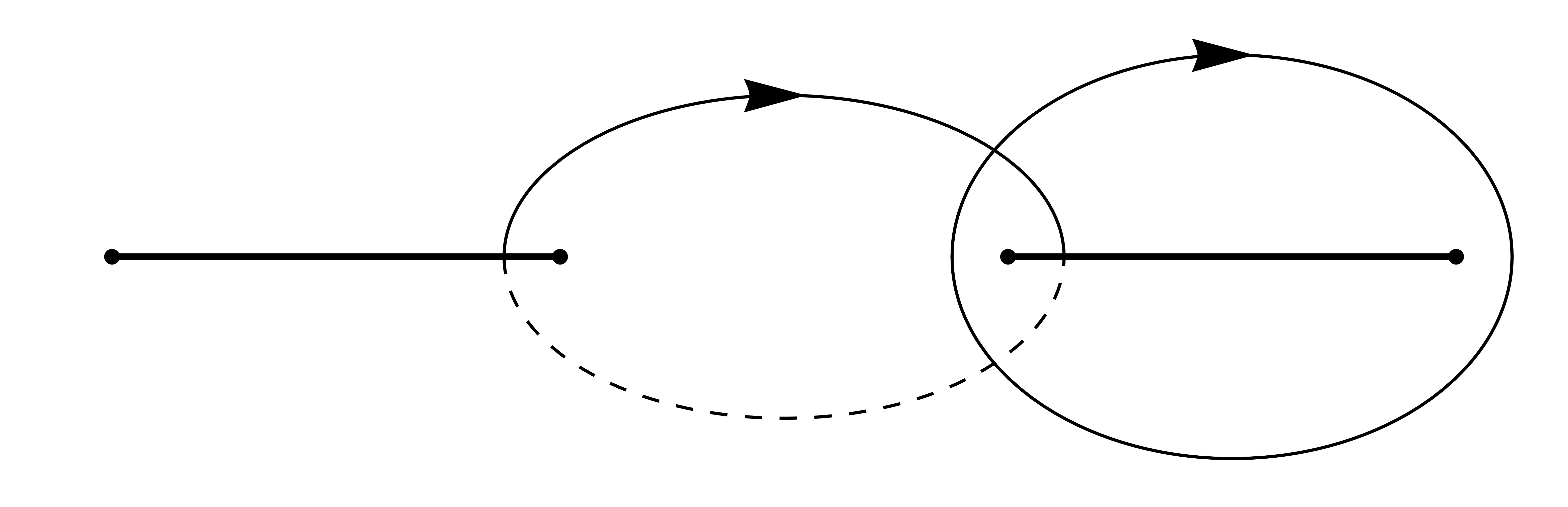}
			\put(43.3,27){ A}
			\put(72.3,29.3){ B}
			\put(32.3,13.3){ $-\eta_1$}
			\put(62.3,13.3){ $\eta_1$}
			\put(3,13.3){ $-\alpha_1$}
			\put(91.3,13.3){ $\alpha_1$}
		\end{overpic}
		\caption{The Riemann surface $\mathcal{S}_{\alpha_1}$ of genus one and its basis of cycles}
		\label{1genus}
	\end{figure}
	\par
	In order to recover the potential function $u(x,t)$, it is necessary to formulate $\partial_x(tg'_{\alpha_{1}}(\lambda))$ and $\partial_x(t{\Omega}_{\alpha_{1}})$. Thus we have the following lemma.
	
	\begin{lem}
		According the representation of $g'_{\alpha_{1}}(\lambda)$ in (\ref{g1 derivative}) and ${\Omega}_{\alpha_{1}}$ in (\ref{tilde Omega equations}), the following two identities hold
		\begin{equation}\label{x derivate of g1}
			\begin{aligned}
				&\partial_x(tg'_{\alpha_{1}}(\lambda))=1-\frac{Q_{\alpha_1,1}(\lambda)}{R_{\alpha_{1}}(\lambda)},\\
				&\partial_x(t{\Omega}_{\alpha_{1}})=-\frac{\pi i\alpha_{1}}{K(m_{\alpha_1})}.
			\end{aligned}
		\end{equation}
	\end{lem}
	\begin{proof}
		The equation (\ref{g1 derivative}) shows that
		$$
		\partial_x(tg'_{\alpha_{1}}(\lambda))=1-\frac{Q_{\alpha_{1},1}
			(\lambda)}{R_{\alpha_{1}}(\lambda)}+\partial_{\alpha_1}
		\left(12t\frac{Q_{\alpha_1,2}(\lambda)}{R_{\alpha_{1}}(\lambda)}
		-x\frac{Q_{\alpha_1,1}(\lambda)}{R_{\alpha_{1}}(\lambda)}\right)\partial_{x}{\alpha_{1}}.
		$$
		It suffices to show that $\partial_{\alpha_1}\left(12t\frac{Q_{\alpha_1,2}(\lambda)}{R_{\alpha_{1}}(\lambda)}d\lambda
		-x\frac{Q_{\alpha_1,1}(\lambda)}{R_{\alpha_{1}}(\lambda)}d\lambda\right)\equiv0$. Indeed, the term $12t\frac{Q_{\alpha_1,2}(\lambda)}{R_{\alpha_{1}}(\lambda)}d\lambda
		-x\frac{Q_{\alpha_1,1}(\lambda)}{R_{\alpha_{1}}(\lambda)}d\lambda$ is a holomorphic differential and its integral along the $A$-cycle is zero. Thus by the Riemann bilinear relation, it is identically zero. On the other hand, from the equation (\ref{tilde Omega equations}), it can be derived that
		$$		\partial_x(t{\Omega}_{\alpha_{1}})=-\partial_x\left(\oint_Bt(g'_{\alpha_{1}}(\lambda)+12\lambda^2-4\xi)\right)=\oint_B\frac{Q_{\alpha_1,1}(\lambda)}{R_{\alpha_{1}}(\lambda)}=-\frac{\pi i\alpha_{1}}{K(m_{\alpha_1})}.
		$$
	\end{proof}
	\par
	By the same way in discussing the asymptotic behavior of $u(x,0)$ as $x\to+\infty$ in Section \ref{potential behavior}, introduce
	$$	T_{\alpha_{1}}(\lambda)=Y(\lambda){e^{tg_{\alpha_{1}}(\lambda)\sigma_3}}f_{\alpha_{1}}(\lambda)^{\sigma_3},
	$$
	where $f_{\alpha_{1}}(\lambda)$ satisfies the following scalar RH problem:
	\begin{equation}\label{RHP for falpha1}
		\begin{aligned}			&f_{\alpha_{1},+}(\lambda)f_{\alpha_{1},-}(\lambda)={r(\lambda)},&&\lambda\in(\eta_{1},\alpha_{1}),\\
			&f_{\alpha_{1},+}(\lambda)f_{\alpha_{1},-}(\lambda)=\frac{1}{r(\lambda)},&&\lambda\in(-\alpha_{1},-\eta_{1}),\\
			&\frac{f_{\alpha_{1},+}(\lambda)}{f_{\alpha_{1},-}(\lambda)}=e^{{\Delta}_{\alpha_{1}}},&&\lambda\in [-\eta_1,\eta_1],\\
			&f_{\alpha_{1}}(\lambda)=1+\mathcal{O}\left(\frac{1}{\lambda}\right),&&\lambda\to\infty.
		\end{aligned}
	\end{equation}
	It is immediate that the function $f_{\alpha_{1}}(\lambda)$ is formulated as
	\begin{equation}
		\begin{aligned}
			f_{\alpha_{1}}(\lambda)=&\exp\left(\frac{R_{\alpha_{1}}(\lambda)}{2\pi i}\left[
			\int_{\eta_1}^{\alpha_1}\frac{\log{r(\zeta)}}{R_{\alpha_{1},+}(\zeta)(\zeta-\lambda)}d\zeta
			+\int_{-\alpha_1}^{-\eta_1}\frac{\log\frac{1}{r(\zeta)}}{R_{\alpha_{1},+}(\zeta)(\zeta-\lambda)}d\zeta
			+\int_{-\eta_1}^{\eta_1}\frac{{\Delta}_{\alpha_{1}}}{R_{\alpha_{1}}(\zeta)(\zeta-\lambda)}d\zeta\right]\right)
		\end{aligned},
	\end{equation}
	in which the normalization condition indicates that
	$ {\Delta}_{\alpha_{1}}=\frac{\alpha_{1}}{K(m_{\alpha_{1}})}\int_{\eta_1}^{\alpha_1}\frac{\log{r(\zeta)}}{R_{\alpha_{1},+}(\zeta)}d\zeta
	$. Thus the row vector $T_{\alpha_{1}}(\lambda)$ obeys the following RH problem:
	\begin{equation}\label{T alpha1 RHP}
		\begin{aligned}
			T_{\alpha_{1},+}(\lambda)=T_{\alpha_{1},-}(\lambda)
			\begin{cases}
				\begin{aligned}
					&\begin{pmatrix}
						e^{t(g_{\alpha_{1},+}(\lambda)-g_{\alpha_{1},-}(\lambda))} \frac{f_{\alpha_{1},+}(\lambda)}{f_{\alpha_{1},-}(\lambda)} & -i \\
						0 & e^{-t(g_{\alpha_{1},+}(\lambda)-g_{\alpha_{1},-}(\lambda))} \frac{f_{\alpha_{1},-}(\lambda)}{f_{\alpha_{1},+}(\lambda)}
					\end{pmatrix}
					, & \lambda \in \Sigma_{1_{\alpha_{1}}},\\
					& \begin{pmatrix}
						e^{t(g_{\alpha_{1},+}(\lambda)-g_{\alpha_{1},-}(\lambda))} \frac{f_{\alpha_{1},+}(\lambda)}{f_{\alpha_{1},-}(\lambda)} & 0 \\
						i  & e^{-t(g_{\alpha_{1},+}(\lambda)-g_{\alpha_{1},-}(\lambda))} \frac{f_{\alpha_{1},-}(\lambda)}{f_{\alpha_{1},+}(\lambda)}
					\end{pmatrix}, & \lambda \in \Sigma_{2_{\alpha_{1}}},   \\
					& \begin{pmatrix}
						e^{t {\Omega}_{\alpha_{1}}+{\Delta}_{\alpha_{1}}}  & 0 \\
						0 & e^{-t {\Omega}_{\alpha_{1}}-{\Delta}_{\alpha_{1}}}
					\end{pmatrix}, & \lambda \in\left[-\eta_1, \eta_1\right], \\
					& \begin{pmatrix}
						1 & \frac{-ir(\lambda)}{ f^2(\lambda)}e^{-2t(g(\lambda)+4\lambda^3-4\xi\lambda)} \\
						0 & 1
					\end{pmatrix}, & \lambda \in[\alpha,\eta_2]\cup\Sigma_{3}, \\
					& \begin{pmatrix}
						1 & 0 \\
						{ir(\lambda)f^2(\lambda)}e^{2t(g(\lambda)+4\lambda^3-4\xi\lambda)} & 1
					\end{pmatrix}, & \lambda \in[-\eta_2,\alpha]\cup\Sigma_{4},\\
					& \begin{pmatrix}
						1 & 0 \\
						0 & 1
					\end{pmatrix}, & \lambda \in[\eta_2,\eta_3]\cup[-\eta_3,-\eta_2],
				\end{aligned}
			\end{cases} \\
		\end{aligned}
	\end{equation}
	and
	$$
	T_{\alpha_{1}}(\lambda)= \begin{pmatrix}
		1&1
	\end{pmatrix}+\mathcal{O}\left(\frac{1}{\lambda}\right),\quad \lambda\to\infty.
	$$
	\par
	Similarly, factorize the RH problem (\ref{T alpha1 RHP}) by introducing
	\begin{equation}
		S_{\alpha_{1}}(\lambda)=
		\begin{cases}
			T_{\alpha_{1}}(\lambda)\begin{pmatrix}
				1 & 0 \\
				\frac{if_{\alpha_{1}}^2(\lambda)}{\hat r(\lambda)}e^{2t(g_{\alpha_{1}}(\lambda)+4\lambda^3-4\xi\lambda)} & 1
			\end{pmatrix}, & {\rm inside~the~contour}~\mathcal{C}_{1,\alpha_{1}},
			\\
			T_{\alpha_{2}}(\lambda)\begin{pmatrix}
				1 & \frac{-i}{\hat r(\lambda)f_{\alpha_{1}}^2(\lambda)}e^{-2t(g_{\alpha_{1}}(\lambda)+4\lambda^3-4\xi\lambda)} \\
				0 & 1
			\end{pmatrix}, & {\rm inside~the~contour}~\mathcal{C}_{2,\alpha_{1}},
			\\
			T_{\alpha_{2}}(\lambda), & {\rm elsewhere},
		\end{cases}
	\end{equation}
	where $\mathcal{C}_{j,\alpha_{1}}$ are around $\Sigma_{j_{\alpha_{1}}}$ for $j=1$ and $j=2$, respectively. Moreover, the contours and jump matrices for $S_{\alpha_{1}}$ are depicted in Figure \ref{S alpha1}.
	\begin{figure}[h]
		\centering
		\begin{tikzpicture}[>=latex]
			\draw[lightgray,very thick,dashed] (-8.5,0) to (-7.5,0);
			\draw[lightgray,very thick] (-7.5,0) to (-5.5,0) node[black,above=10mm]{\small $\begin{pmatrix}
					1 & 0 \\
					\textcolor{gray}{ir(\lambda)f^2(\lambda)e^{2t(g(\lambda)+4\lambda^3-4\xi\lambda)}} & 1
				\end{pmatrix}$};
			\draw[->,dashed,very thick] (-5.8,1.2) to (-6.5,0.2);
			\draw[->,dashed,very thick] (-4.8,1.2) to (-4.0,0.2);
			\draw[->,very thick,lightgray] (-7.5,0) to (-6.5,0);
			\filldraw[lightgray] (-7.5,0) node[black,below=1mm]{$-\eta_{4}$} circle (1.5pt);
			\filldraw[lightgray] (-5.5,0) node[black,below=1mm]{$-\eta_{3}$} circle (1.5pt);
			\draw[->,very thick,lightgray] (-4.5,0) to (-4.0,0);
			\draw[-,very thick,lightgray] (-4.5,0) to (-3.5,0);
			\filldraw[lightgray] (-4.5,0) node[black,below=1mm]{$-\eta_{2}$} circle (1.5pt);
			\filldraw[lightgray] (-3.5,0) node[black,below=1mm]{$-\alpha_1$} circle (1.5pt);
			\draw[-,very thick,black] (-3.5,0) to (-0.5,0);
			\draw[->,very thick,black] (-3.5,0) to (-2.5,0) node[black,right=-1mm]{\small $\begin{pmatrix}
					0 & i \\
					i & 0
				\end{pmatrix}$};
			\filldraw[black] (-0.5,0) node[black,below=1mm]{$-\eta_1$} circle (1.5pt);
			\draw[-,very thick,lightgray] (-3.5,0) .. controls (-2.5,1.15) and (-1.5,1.15).. (-0.5,0);
			\draw[->,very thick,lightgray] (-2.1,0.85) to (-2.0,0.85) node[black,above]{\small $\mathcal{C}_{2,\alpha_{1}}$};
			\draw[-,very thick,lightgray,rotate around x=180] (-3.5,0) .. controls (-2.5,1.15) and (-1.5,1.15).. (-0.5,0) ;
			\draw[<-,very thick,lightgray,rotate around x=180] (-2.3,0.85) to (-2.2,0.85)node[black,below]{\small $\begin{pmatrix}
					1 & \textcolor{gray}{\frac{-i}{\hat r(\lambda)f_{\alpha_{1}}^2(\lambda)}e^{-2t(g_{\alpha_{1}}(\lambda)+4\lambda^3-4\xi\lambda)}} \\
					0 & 1
				\end{pmatrix}$} ;
			\draw [very thick] (-0.5,0) to (0.5,0);
			\draw [->,very thick] (-0.5,0) to (0,0)node[black,above=12mm]{\small $\begin{pmatrix}
					e^{t{\Omega}_{\alpha_{1}}+{\Delta}_{\alpha_{1}}} & 0\\
					0 & e^{-t{\Omega}_{\alpha_{1}}-{\Delta}_{\alpha_{1}}}
				\end{pmatrix}$};
			\draw[->,dashed,very thick] (0,1.2) to (0,0.2);
			\draw[lightgray,very thick,dashed] (7.5,0) to (8.5,0);
			\draw[lightgray,very thick] (5.5,0) to (7.5,0) ;
			\draw[->,very thick,lightgray] (5.5,0) to (6.5,0)node[black,above=10mm]{\small $\begin{pmatrix}
					1 & \textcolor{gray}{\frac{-ir(\lambda)}{ f^2(\lambda)}e^{-2t(g(\lambda)+4\lambda^3-4\xi\lambda)}} \\
					0 & 1
				\end{pmatrix}$};
			\draw[->,dashed,very thick] (5.8,1.2) to (6.5,0.2);
			\draw[->,dashed,very thick] (4.8,1.2) to (4.0,0.2);
			\filldraw[lightgray] (7.5,0) node[black,below=1mm]{$\eta_{4}$} circle (1.5pt);
			\filldraw[lightgray] (5.5,0) node[black,below=1mm]{$\eta_{3}$} circle (1.5pt);
			\draw[->,very thick,lightgray] (3.5,0) to (4.0,0);
			\draw[-,very thick,lightgray] (3.5,0) to (4.5,0);
			\filldraw[lightgray] (4.5,0) node[black,below=1mm]{$\eta_{2}$} circle (1.5pt);
			\filldraw[lightgray] (3.5,0) node[black,below=1mm]{$\alpha_1$} circle (1.5pt);
			\draw[-,very thick,black] (0.5,0) to (3.5,0);
			\draw[->,very thick,black] (0.5,0) to (1.5,0) node[black,right=-1mm]{\small $\begin{pmatrix}
					0 & -i \\
					-i & 0
				\end{pmatrix}$};
			\filldraw[black] (0.5,0) node[black,below=1mm]{$\eta_1$} circle (1.5pt);
			\draw[-,very thick,lightgray] (0.5,0) .. controls (1.5,1.15) and (2.5,1.15).. (3.5,0);
			\draw[->,very thick,lightgray] (2,0.855) to (2.1,0.855)node[black,above]{\small $\mathcal{C}_{1,\alpha_{1}}$} ;
			\draw[-,very thick,lightgray,rotate around x=180] (0.5,0) .. controls (1.5,1.15) and (2.5,1.15).. (3.5,0);
			\draw[<-,very thick,lightgray,rotate around x=180] (2.2,0.85) to (2.1,0.85)node[black,anchor=south west]at(1,2.1){\small $\begin{pmatrix}
					1 & 0 \\
					\textcolor{gray}{\frac{if_{\alpha_{1}}^2(\lambda)}{\hat r(\lambda)}e^{2t(g_{\alpha_{1}}(\lambda)+4\lambda^3-4\xi\lambda)}} & 1
				\end{pmatrix}$} ;
		\end{tikzpicture}
		\caption{{\protect\small
				The contours and the jump matrices for \( S_{\alpha_{1}}(\lambda) \): the gray entries in the matrices vanish exponentially as \( t \to +\infty \), and the gray contours also vanish as \( t \to +\infty \).}}
		\label{S alpha1}
	\end{figure}
	
	In order to transform the RH problem for $S_{\alpha_{1}}(\lambda)$ into a model problem, the other properties of $g_{\alpha_{1}}(\lambda)$ in (\ref{condition for g1}) should be illustrated.
	
	\begin{lem}
		The following inequalities are established:
		\begin{equation}\label{g1 signature}
			\begin{aligned}
				&\re\left(g_{\alpha_{1}}(\lambda)+4\lambda^3-4\xi\lambda\right)>0, &&\lambda\in\Sigma_{3}\cup(\alpha_1,\eta_2)\cup\mathcal{C}_{2,\alpha_{1}}\setminus\{-\eta_1,-\alpha_1\},\\
				&\re\left(g_{\alpha_{1}}(\lambda)+4\lambda^3-4\xi\lambda\right)<0, &&\lambda\in\Sigma_{4}\cup(-\eta_2,-\alpha_1)\cup\mathcal{C}_{1,\alpha_{1}}\setminus\{\eta_1,\alpha_1\}.\\
			\end{aligned}
		\end{equation}
	\end{lem}
	\begin{proof}
		Recall that
		$$		g'_{\alpha_{1}}(\lambda)+12\lambda^2-4\xi=12\frac{(\lambda^2-\alpha_1^2)}{R_{\alpha_{1}}(\lambda)}\left[\lambda^2-\left(\frac{\eta_1^2-\alpha_{1}^2}{2}+\frac{\xi}{3}\right)\right].
		$$
		For $\lambda\in(\eta_1,\alpha_{1})$, one has
		\begin{equation}\label{quadratic1} g'_{\alpha_{1},+}(\lambda)+12\lambda^2-4\xi=12\frac{i\sqrt{\alpha_1^2-\lambda^2}}{\sqrt{\lambda^2-\eta_1^2}}\left[\lambda^2-\left(\frac{\eta_1^2-\alpha_{1}^2}{2}+\frac{\xi}{3}\right)\right],
		\end{equation}
		which indicates that $g_{\alpha_{1},+}$ is purely imaginary and the normalization condition in (\ref{tilde Omega equations}) implies the right hand side of the equation (\ref{quadratic1}) has a nonnegative root in $[0,\eta_1]$. Thus $-i(g'_{\alpha_{1},+}(\lambda)+12\lambda^2-4\xi)>0$ and together with Cauchy-Riemann equation, it follows that
		$$
		\re(g_{\alpha_{1},+}(\lambda)+4\lambda^3-4\xi\lambda)<0,
		$$ for $\lambda\in \mathbb{C}^{+}\cap\mathcal{C}_{1,\alpha_{1}}\setminus\{\eta_1,\alpha_1\}$. Similarly, the signature of $\re\left(g_{\alpha_{1}}(\lambda)+4\lambda^3-4\xi\lambda\right)$ on $\lambda\in\mathcal{C}_{2,\alpha_{1}}\setminus\{-\alpha_1,-\eta_1\}$ can also be proved.
		\par		
		For $\lambda\in(\alpha_{1},\eta_2)\cup\Sigma_{3}$, one has
		\begin{equation}\label{quadratic2} g'_{\alpha_{1}}(\lambda)+12\lambda^2-4\xi=12\frac{\sqrt{\lambda^2-\alpha_1^2}}{\sqrt{\lambda^2-\eta_1^2}}\left[\lambda^2-\left(\frac{\eta_1^2-\alpha_{1}^2}{2}+\frac{\xi}{3}\right)\right].
		\end{equation}
		Since the nonnegative root of the right hand side of the equation (\ref{quadratic1}) lives in $[0,\eta_1]$, it is immediate that $(	g'_{\alpha_{1}}(\lambda)+12\lambda^2-4\xi)>0$. Consequently, by the definition of $g_{\alpha_{1}}$ in (\ref{g1 formular}), it follows $\re\left(g_{\alpha_{1}}(\lambda)+4\lambda^3-4\xi\lambda\right)>0$. In a similar way, the inequality on $(-\eta_2,-\alpha_{1})\cup\Sigma_{4}$ can be obtained.
	\end{proof}
	As a result, as $t\to+\infty$, the jump matrices of the RH problem for $S_{\alpha_{1}}(\lambda)$ restricted on the gray contours in Figure \ref{S alpha1} converge to identity matrix exponentially outside the points $\pm\alpha_{1}$ and $\pm \eta_1$. So we arrive at the model RH problem for $S^{\infty}_{\alpha_{1}}(\lambda)$ as follows:
	\begin{equation}\label{S infty alpha1 RHP}
		{S}_{\alpha_{1},+}^{\infty}(\lambda)={S}_{\alpha_{1},-}^{\infty}(\lambda) \begin{cases}{\begin{pmatrix}
					e^{t {\Omega_{\alpha_{1}}}+{\Delta}_{\alpha_{1}}} & 0 \\
					0 & e^{-t {\Omega_{\alpha_{1}}}-{\Delta}_{\alpha_{1}}}
			\end{pmatrix}}, & \lambda \in[-\eta_1, \eta_1], \\
			{\begin{pmatrix}
					0 & -i \\
					-i & 0
			\end{pmatrix}}, & \lambda \in \Sigma_{1_{\alpha_1}}=(\eta_1,\alpha_{1}), \\
			{\begin{pmatrix}
					0 & i \\
					i & 0
			\end{pmatrix}}, & \lambda \in \Sigma_{ 2_{\alpha_1}}=(-\alpha_{1},-\eta_{1}),\end{cases}
	\end{equation}
	and
	$$
	S^{\infty}_{\alpha_{1}}(\lambda)\to\begin{pmatrix}
		1&1
	\end{pmatrix}+\mathcal{O}\left(\frac{1}{\lambda}\right),\quad \lambda\to\infty,
	$$
	whose solution can be expressed explicitly by
	\begin{equation}\label{S infty alpha1 solution}		S^{\infty}_{\alpha_{1}}(\lambda)=\gamma_{\alpha_{1}}(\lambda)\frac{\vartheta_3(0;2\tau_{\alpha_{1}})}{\vartheta_3(\frac{t{\Omega_{\alpha_{1}}}+{\Delta_{\alpha_{1}}}}{2\pi i};2\tau_{\alpha_{1}})}\begin{pmatrix}
			\frac{\vartheta_3\left(2J_{\alpha_{1}}(\lambda)+\frac{t\Omega_{\alpha_{1}}+{\Delta_{\alpha_{1}}}}{2\pi i}-\frac{1}{2};2\tau_{\alpha_{1}}\right)}{\vartheta_3\left(2J_{\alpha_{1}}(\lambda)-\frac{1}{2};2\tau_{\alpha_{1}}\right)} & \frac{\vartheta_3\left(-2J_{\alpha_{1}}(\lambda)+\frac{t\Omega_{\alpha_{1}}+{\Delta_{\alpha_{1}}}}{2\pi i}-\frac{1}{2};2\tau_{\alpha_{1}}\right)}{\vartheta_3\left(-2J_{\alpha_{1}}(\lambda)-\frac{1}{2};2\tau_{\alpha_{1}}\right)}
		\end{pmatrix},
	\end{equation}
	where $\gamma_{\alpha_{1}}(\lambda)=\left(\frac{\lambda^2-\eta_{1}^2}{\lambda^2-\alpha_{1}^2}\right)^{\frac{1}{4}}$, $J_{\alpha_{1}}(\lambda)=\int_{\alpha_{1}}^\lambda\omega_{\alpha_1}$ and $\tau_{\alpha_{1}}=\oint_{B}\omega_{\alpha_1}$. More precisely, $J_{\alpha_{1},+}-J_{\alpha_{1},-}=-\tau_{\alpha_{1}}$ for $\lambda\in[-\eta_{1},\eta_{1}]$, $J_{\alpha_{1},+}+J_{\alpha_{1},-}=0~\mod(\mathbb{Z})$ for $\lambda\in\Sigma_{ 2_{\alpha_1}}\cup\Sigma_{ 1_{\alpha_1}}$ and $J_{\alpha_{1}}(\infty)=-\frac{1}{4}$. Thus it is immediate that (\ref{S infty alpha1 solution}) solves the RH problem (\ref{S infty alpha1 RHP}).
	\begin{thm}
		For $\xi=\frac{x}{4t}$, in the region $\eta_1^2<\xi<\xi_{crit}^{(1)}$, the large-time asymptotic behavior of the solution to the KdV equation with genus two soliton gas potential is described by
		\begin{equation}\label{1 modulated solution}			u(x,t)=\alpha_{1}^2-\eta_{1}^2-2\alpha_{1}^2\frac{E(m_{\alpha_{1}})}{K(m_{\alpha_{1}})}-2{\partial_x^2}\log\vartheta_3\left(\frac{\alpha_{1}}{2K(m_{\alpha_{1}})}(x-2(\alpha_{1}^2+\eta_{1}^2)t+\phi_{\alpha_{1}});2\tau_{\alpha_{1}}\right)+\mathcal{O}\left(\frac{1}{t}\right),
		\end{equation}
		where
		$$
		\phi_{\alpha_{1}}=\int_{\alpha_{1}}^{\eta_{1}}\frac{\log r(\zeta)}{R_{\alpha_{1},+}(\zeta)}\frac{d\zeta}{\pi i},
		$$
		and the parameter $\alpha_{1}$ is determined by (\ref{alpha1 formular}). Alternatively, the expression (\ref{1 modulated solution}) can also be rewritten as
		\begin{equation}\label{1 modulated solution dn}
			u(x,t)=\alpha_{1}^2-\eta_{1}^2-2\alpha_{1}^2 \mathrm{dn}^2\left(\alpha_{1}(x-2(\alpha_{1}^2+\eta_{1}^2)t+\phi_{\alpha_{1}})
			+K(m_{\alpha_{1}}); m_{\alpha_{1}}\right)+\mathcal{O}\left(\frac{1}{t}\right),
		\end{equation}
		where $\mathrm{dn}(z; m)$ is the Jacobi elliptic function with modulus $m_{\alpha_{1}}=\frac{\eta_{1}}{\alpha_{1}}$.
	\end{thm}
	\begin{proof}
		By the same way as the analysis of the genus two KdV soliton gas potential $u(x,0)$ for $x\to+\infty$, take
		$$	Y_1(\lambda)=\left(S_{\alpha_{1},1}^{\infty}(\lambda)
		+\frac{(\mathcal{E}_{\alpha_1,1}(x,t))_1}{t\lambda}+\mathcal{O}\left(\frac{1}{\lambda^2}\right)\right)
		e^{-tg_{\alpha_{1}}(\lambda)}f_{\alpha_{1}}(\lambda)^{-1},
		$$
		in which the $f_{\alpha_{1}}(\lambda)$ has the asymptotics
		$$		f_{\alpha_{1}}(\lambda)=1+\frac{f_{\alpha_{1}}^{(1)}(\alpha_{1},\eta_{1})}{\lambda}+\mathcal{O}\left(\frac{1}{\lambda^2}\right),
		$$
		with
		$$
		f_{\alpha_{1}}^{(1)}(\alpha_{1},\eta_{1})=\int_{\alpha_{1}}^{\eta_{1}}\frac{\zeta^2\log r(\zeta)}{R_{\alpha_{1}}(\zeta)}\frac{d\zeta}{\pi i}-{\Delta_{\alpha_{1}}}\int_{-\eta_{1}}^{\eta_{1}}\frac{\zeta^2}{R_{\alpha_{1}}(\zeta)}\frac{d\zeta}{2\pi i},
		$$
		and the term $\frac{(\mathcal{E}_{\alpha_1,1}(x,t))_1}{t\lambda}$ is the first entry of the error vector corresponding to the modulated one-phase wave region. It is noted that the local paramertix near $\pm\alpha_1$ and $\pm\eta_1$ can be described by Airy function and modified Bessel function respectively and both of them contribute the error term $\mathcal{O}(t^{-1})$ in the asymptotic behavior of potential $u(x,t)$.
		\par		
		The derivative of term $e^{-tg_{\alpha_{1}}(\lambda)}$ on $x$ is
		$$ \partial_xe^{-tg_{\alpha_{1}}(\lambda)}=-\frac{1}{\lambda}\left[\frac{\alpha_{1}^2+\eta_{1}^2}{2}+\alpha_{1}^2\left(\frac{E(m_{\alpha_{1}})}{K(m_{\alpha_{1}})}-1\right)\right]+\mathcal{O}\left(\frac{1}{\lambda^2}\right).
		$$
		\par
		So the function $S_{\alpha_{1},1}^{\infty}(\lambda)$ behaves		
		$$		S_{\alpha_{1},1}^{\infty}(\lambda)=1+\frac{1}{\lambda}\left[\nabla\log\left(\vartheta_3\left(\frac{t\Omega_{\alpha_{1}}+{\Delta_{\alpha_{1}}}}{2\pi i};2\tau_{\alpha_{1}}\right)\right)-\nabla\log(\vartheta_3(0;2\tau_{\alpha_{1}}))\right] \frac{ \alpha_{1}}{2K(m_{\alpha_{1}})}
		+\mathcal{O}\left(\frac{1}{\lambda^2}\right),
		$$
		where $\nabla$ stands for the derivative of $\vartheta_3$. In particular, one has
		$$		\partial_xS_{\alpha_{1},1}^{\infty}(\lambda)=\frac{1}{\lambda}\nabla^2\log\left(\vartheta_3\left(\frac{t\Omega_{\alpha_{1}}+{\Delta_{\alpha_{1}}}}{2\pi i};2\tau_{\alpha_{1}}\right)\right)\frac{ \alpha_{1}}{2K(m_{\alpha_{1}})}\left(\frac{\partial_x(t\Omega_{\alpha_{1}}+{\Delta_{\alpha_{1}}})}{2\pi i}\right)+\mathcal{O}\left(\frac{1}{\lambda^2}\right).
		$$
		Together with (\ref{x derivate of g1}) and $\partial_x{\Delta_{\alpha_{1}}}=\frac{\partial_{\alpha_{1}}{\Delta_{\alpha_{1}}}\partial_{\xi}\alpha_{1}}{4t}$, it follows that
		$$		\partial_xS_{\alpha_{1},1}^{\infty}(\lambda)=-\frac{1}{\lambda}\left[\partial^2_x\log\left(\vartheta_3\left(\frac{t\Omega_{\alpha_{1}}+{\Delta_{\alpha_{1}}}}{2\pi i};2\tau_{\alpha_{1}}\right)\right)+\mathcal{O}\left(\frac{1}{t}\right)\right]
		+\mathcal{O}\left(\frac{1}{\lambda^2}\right).
		$$
		Combining all the above expansions with the fact that $\partial_x\alpha_{1}\sim
		\frac{1}{t}$ for $t\to+\infty$, it is sufficient to show the large-time asymptotics of $u(x,t)$ in (\ref{1 modulated solution}) for $\eta_{1}^2<\xi<\xi_{crit}^{(1)}$.
	\end{proof}

	\subsection{Unmodulated one-phase wave region}\label{Unmodulated 1-genus case}
	
	The equality (\ref{alpha1 formular}) shows that $\alpha_{1}(\xi_{crit}^{(1)})=\eta_2$.
	Thus when $\xi_{crit}^{(1)}<\xi<\xi_{crit}^{(2)}$, where the parameter $\xi_{crit}^{(2)}$ is determined by (\ref{xi crit}), the jump matrices on $\Sigma_{3,4}$ are still exponentially decreasing to identity matrix for $t \to +\infty$. The large-time behavior of $u(x,t)$ in this region is expressed by an unmodulated one-phase Jacobi elliptic wave. For convenience, we just need to modify the subscripts in Subsection \ref{Modulated-4-1}, such as \( g_{\eta_{2}} \), $R_{\eta_{2}}(\lambda)$,  \( \mathcal{S}_{\eta_{2}} \), $\Omega_{\eta_{2}}$ and $\Delta_{\eta_{2}}$, which are defined by replacing \( \alpha_{1} \) with \( \eta_{2} \). In particular, the equation (\ref{g1 formular}) becomes
	$$	g_{\eta_{2}}(\lambda)=-4\lambda^3+4\xi\lambda+12\int_{\eta_{2}}^{\lambda}\frac{Q_{\eta_{2},2}
		(\zeta)}{R_{\eta_{2}}(\zeta)}d\zeta-4\xi\int_{\eta_{2}}^{\lambda}\frac{Q_{\eta_{2},1}(\zeta)}
	{R_{\eta_{2}}(\zeta)}d\zeta,
	$$
	{with $R_{\eta_2}(\lambda):=\sqrt{(\lambda^2-\eta_1^2)(\lambda^2-\eta_2^2)}$,} which satisfies the following RH problem:
	$$
	\begin{aligned}
		&g_{\eta_{2,+}}(\lambda)+g_{\eta_{2,-}}(\lambda)+8\lambda^3-8\xi\lambda=0, &&\lambda\in \Sigma_{1}\cup \Sigma_{2},\\
		&g_{\eta_{2,+}}(\lambda)-g_{\eta_{2,-}}(\lambda)={\Omega_{\eta_{2}}}, &&\lambda\in [-\eta_1,\eta_1],\\
		&g_{\eta_{2}}(\lambda)=\mathcal{O}\left(\frac{1}{\lambda}\right), && \lambda\to \infty.
	\end{aligned}
	$$
	Moreover, the function $g_{\eta_{2}}(\lambda)$ obeys the similar conditions
	in (\ref{condition for g1}) as follows
	$$
	\begin{aligned}
		&g_{\eta_2}(\lambda)+4\lambda^3-4\xi\lambda=(\lambda\pm\eta_2)^{\frac{3}{2}}, &&\lambda\to\pm \eta_2,\\
		&\re\left(g_{\eta_2}(\lambda)+4\lambda^3-4\xi\lambda\right)>0, &&\lambda \in  \Sigma_3,\\
		&\re\left(g_{\eta_2}(\lambda)+4\lambda^3-4\xi\lambda\right)<0, &&\lambda \in  \Sigma_4,\\
		&-i(g_{\eta_2,+}(\lambda)-g_{\eta_2,-}(\lambda)) \text{ is real-valued and monotonically increasing},&& \lambda\in\Sigma_{1},\\
		&-i(g_{\eta_2,+}(\lambda)-g_{\eta_2,-}(\lambda)) \text{ is real-valued and monotonically decreasing}, && \lambda\in\Sigma_{2}.\\
	\end{aligned}
	$$
	\par
	Consequently, the RH problem for $Y(\lambda)$ can be transformed into a model problem for $S^{\infty}_{\eta_{2}}(\lambda)$ whose solution is
	$$ S^{\infty}_{\eta_2}(\lambda)=\gamma_{\eta_2}(\lambda)\frac{\vartheta_3(0;2\tau_{\eta_2})}{\vartheta_3(\frac{t{\Omega_{\eta_2}}+{\Delta_{\eta_2}}}{2\pi i};2\tau_{\eta_2})}\begin{pmatrix}
		\frac{\vartheta_3\left(2J_{\eta_2}(\lambda)+\frac{t\Omega_{\eta_2}+{\Delta_{\eta_2}}}{2\pi i}-\frac{1}{2};2\tau_{\eta_2}\right)}{\vartheta_3\left(2J_{\eta_2}(\lambda)-\frac{1}{2};2\tau_{\eta_2}\right)} & \frac{\vartheta_3\left(-2J_{\eta_2}(\lambda)+\frac{t\Omega_{\eta_2}+{\Delta_{\eta_2}}}{2\pi i}-\frac{1}{2};2\tau_{\eta_2}\right)}{\vartheta_3
			\left(-2J_{\eta_2}(\lambda)-\frac{1}{2};2\tau_{\eta_2}\right)}
	\end{pmatrix}.
	$$
    
        {In particular, the model problem satisfies the jump conditions as illustrated in Figure \ref{jumpfor-Sinfty-eta2}.}
        \begin{figure}[h!]
    \centering
    \begin{tikzpicture}
    [>=latex]

        \draw[-,,dashed,lightgray,very thick] (-6.5,0) to (-4.5,0);
        
        \draw[-,very thick] (-4.5,0) to (-1.5,0);
        \filldraw[black] (-4.5,0) node[black,below=1mm]{$-\eta_2$} circle (1.5pt);
        \draw[->,very thick] (-4.5,0) to (-3.0,0) node[black,above=0.5mm]{\small $\begin{pmatrix}
					0 & i \\
					i & 0
			\end{pmatrix}$};
        \filldraw[black] (-1.5,0) node[black,below=1mm]{$-\eta_1$} circle (1.5pt);
        
        \draw [-,very thick,red] (-1.5,0) to (1.5,0);
        \draw [->,very thick,red] (-1.5,0) to (0,0)node[red,above=2mm]{ $e^{(t {\Omega_{\eta_{2}}}+{\Delta}_{\eta_{2}})\sigma_3}$};
        \filldraw[black] (1.5,0) node[black,below=1mm]{$\eta_1$} circle (1.5pt);
        \draw[-,very thick,black] (1.5,0) to (4.5,0);
        \draw[->,very thick,black] (1.5,0) to (3,0)node[black,above=0.5mm]{\small $\begin{pmatrix}
					0 & -i \\
					-i & 0
			\end{pmatrix}$};

        \filldraw[black] (4.5,0) node[black,below=1mm]{$\eta_{2}$} circle (1.5pt);
        
       \draw[-,,dashed,lightgray,very thick] (4.5,0) to (6.5,0);
    \end{tikzpicture}
    \caption{{The jump contour for \( S^{\infty}_{\eta_2}(\lambda) \) and the associated jump matrices.}}
    \label{jumpfor-Sinfty-eta2}
\end{figure}
	\par
	Thus for $\xi_{crit}^{(1)}<\xi<\xi_{crit}^{(2)}$, the following theorem holds.
	\begin{thm}
		For $\xi=\frac{x}{4t}$, in the region $\xi_{crit}^{(1)}<\xi<\xi_{crit}^{(2)}$, the large-time asymptotic behavior of the solution to the KdV equation with genus two soliton gas potential is described by
		\begin{equation}\label{unmodulated 1 genus solution}
			u(x,t)=\eta_{2}^2-\eta_{1}^2-2\eta_{2}^2 \dn^2\left(\eta_{2}(x-2(\eta_{2}^2+\eta_{1}^2)t+\phi_{\eta_{2}})+K(m_{\eta_{2}}); m_{\eta_{2}}\right)+\mathcal{O}\left(\frac{1}{t}\right),
		\end{equation}
		where the modulus of the Jacobi elliptic function is $m_{\eta_{2}}=\frac{\eta_{1}}{\eta_{2}}$ and
		$$
		\phi_{\eta_{2}}=\int_{\eta_{2}}^{\eta_{1}}\frac{\log r(\zeta)}{R_{\eta_{2},+}(\zeta)}\frac{d\zeta}{\pi i}.
		$$
	\end{thm}
	\begin{rmk}
		The error estimation is omitted and the local parametrix near $\pm\eta_j$ ($j=1,2$) can be described by modified Bessel function which contribute the $\mathcal{O}(t^{-1})$ term in the expression (\ref{unmodulated 1 genus solution}).	
	\end{rmk}
	
	\subsection{Modulated two-phase wave region}\label{modulated 2-genus case}
	
	This section considers the case that the parameter $\xi$ obeys
	$$
	\xi_{crit}^{(2)}<\xi<\xi_{crit}^{(3)},
	$$
	where the parameters $\xi_{crit}^{(2)}$ and $\xi_{crit}^{(3)}$ are determined by {(\ref{xi2-final})} and (\ref{xi crit})  below. Define
	\begin{equation}
		\Sigma_{3_{\alpha_2}}=(\eta_3,\alpha_2),\ \text{and}\ \Sigma_{4_{\alpha_2}}=(-\alpha_2,-\eta_3),
	\end{equation}
	{see Figure \ref{Sigma-alpha2}}, where $\alpha_{2}$ is defined in (\ref{alpha2 formular}) below.
    \begin{figure}[h!]
    \centering
    \begin{tikzpicture}
    [>=latex]

        \draw[lightgray,very thick,dashed] (-7.5,0) to (-6.5,0);
        
        \filldraw[black] (-6.5,0) node[black,below=1mm]{$-\eta_{4}$} circle (1.5pt);
        \draw[very thick,dashed] (-6.5,0) to (-5.5,0);
        \draw[-,very thick] (-5.5,0) to (-4.5,0);
        \filldraw[black] (-4.5,0) node[black,below=1mm]{$-\eta_{3}$} circle (1.5pt);
        \filldraw[black] (-5.5,0) node[black,below=1mm]{$-\alpha_2$} circle (1.5pt);
         \draw[->,very thick] (-5.5,0) to (-4.8,0) node[black,above=0.5mm]{\small $\Sigma_{4_{\alpha_2}}$};
        
        \draw[-,very thick,dashed,lightgray] (-4.5,0) to (-3.5,0);
        \filldraw[black] (-3.5,0) node[black,below=1mm]{$-\eta_2$} circle (1.5pt);

        \draw[-,dashed,very thick] (-3.5,0) to (-2,0)node[black,above=0.5mm]{\small $\Sigma_{2}$};
        \draw[-,dashed,very thick] (-2,0) to (-0.5,0);
        
        \filldraw[black] (-0.5,0) node[black,below=1mm]{$-\eta_1$} circle (1.5pt);

        \draw [lightgray,dashed,very thick] (-0.5,0) to (0.5,0);
        
        \filldraw[black] (0.5,0) node[black,below=1mm]{$\eta_1$} circle (1.5pt);
        \draw[-,very thick,black,dashed] (0.5,0) to (2,0)node[black,above=0.5mm]{\small $\Sigma_{1}$};

        \draw[-,very thick,black,dashed] (2,0) to (3.5,0);
        \filldraw[black] (3.5,0) node[black,below=1mm]{$\eta_{2}$} circle (1.5pt);
        
        \draw[-,dashed,very thick,lightgray] (3.5,0) to (4.5,0);
        \filldraw[black] (4.5,0) node[black,below=1mm]{$\eta_{3}$} circle (1.5pt);
        
        \draw[very thick] (4.5,0) to (5.5,0);
        
        \draw[very thick,dashed] (5.5,0) to (6.5,0);
        \filldraw[black] (6.5,0) node[black,below=1mm]{$\eta_{4}$} circle (1.5pt);
        \filldraw[black] (5.5,0) node[black,below=1mm]{$\alpha_2$} circle (1.5pt);
        \draw[->,very thick,black] (4.5,0) to (5.2,0)node[black,above=0.5mm]{\small $\Sigma_{3_{\alpha_2}}$};
        \draw[lightgray,very thick,dashed] (6.5,0) to (7.5,0);
    \end{tikzpicture}
    \caption{{The solid lines represent $\Sigma_{3_{\alpha_2}}$ and $\Sigma_{4_{\alpha_2}}$, where $\eta_3 < \alpha_2 < \eta_4$.}}
    \label{Sigma-alpha2}
\end{figure}
    
    In addition, we introduce $g_{\alpha_{2}}(\lambda)$ which is subject to the following scalar RH problem:
	\begin{equation}\label{g2 jumps}
		\begin{aligned}
			&g_{\alpha_{2},+}(\lambda)+g_{\alpha_{2},-}(\lambda)+8 \lambda^3-8 \xi \lambda=0, && \lambda \in \Sigma_{1} \cup \Sigma_2 \cup \Sigma_{3_{\alpha_{2}}} \cup \Sigma_{4_{\alpha_{2}}}, \\
			&g_{\alpha_{2},+}(\lambda)-g_{\alpha_{2},-}(\lambda)={\Omega}_{\alpha_{2},0}, && \lambda \in\left[-\eta_1, \eta_1\right], \\
			&g_{\alpha_{2},+}(\lambda)-g_{\alpha_{2},-}(\lambda)={\Omega}_{\alpha_{2},1}, && \lambda \in\left[\eta_2, \eta_3\right]\cup\left[-\eta_3,-\eta_2\right], \\
			&g_{\alpha_{2}}(\lambda)=\mathcal{O}\left(\frac{1}{\lambda}\right), && \lambda \rightarrow \infty.
		\end{aligned}
	\end{equation}
	It remains to show that the function $g_{\alpha_{2}}(\lambda)$ satisfies the following properties:
	\begin{equation}\label{condition for g2}
		\begin{aligned}
			&g_{\alpha_{2}}(\lambda)+4\lambda^3-4\xi\lambda=(\lambda\pm\alpha_{2})^{\frac{3}{2}}, &&\lambda\to\pm\alpha_{2},\\
			&Re\left(g_{\alpha_{2}}(\lambda)+4\lambda^3-4\xi\lambda\right)>0, &&\lambda \in (\alpha_{2},\eta_4),\\
			&Re\left(g_{\alpha_{2}}(\lambda)+4\lambda^3-4\xi\lambda\right)<0, &&\lambda \in (-\eta_4,-\alpha_{2}),\\
			&-i(g_{\alpha_{2},+}(\lambda)-g_{\alpha_{2},-}(\lambda)) \text{ is real-valued and monotonically increasing},&& \lambda\in\Sigma_{3_{\alpha_{2}}}\cup\Sigma_{1},\\
			&-i(g_{\alpha_{2},+}(\lambda)-g_{\alpha_{2},-}(\lambda)) \text{ is real-valued and monotonically decreasing}, && \lambda\in\Sigma_{4_{\alpha_{2}}}\cup\Sigma_{2}.\\
		\end{aligned}
	\end{equation}
	Further, introduce
	\begin{equation}\label{g2 derivative}		g'_{\alpha_{2}}(\lambda)=-12\lambda^2+4\xi+12\frac{Q_{\alpha_{2},2}(\lambda)}{R_{\alpha_{2}}(\lambda)}-4\xi\frac{Q_{\alpha_{2},1}(\lambda)}{R_{\alpha_{2}}(\lambda)},
	\end{equation}
	with
	$$	R_{\alpha_{2}}(\lambda)=\sqrt{(\lambda^2-\eta_1^2)(\lambda^2-\eta_2^2)(\lambda^2-\eta_3^2)(\lambda^2-\alpha_{2}^2)},
	$$
	where $R_{\alpha_{2}}$ is analytic for $\C\setminus\Sigma_{\{1,2, 3_{\alpha_2},4_{\alpha_2}\}}$ with positive real value for $\lambda>\alpha_{2}$. More precisely, we have
	\begin{equation}		Q_{\alpha_{2},1}=\lambda^4+b_{\alpha_{2},1} \lambda^2+c_{\alpha_{2},1},\
		Q_{\alpha_{2},2}=\lambda^6-\frac{1}{2}(\eta_{1}^2+\eta_{2}^2+\eta_{3}^2+\alpha_{2}^2)\lambda^4+b_{\alpha_{2},2} \lambda^2+c_{\alpha_{2},2},\
	\end{equation}
	where the constants $b_{\alpha_{2},1},b_{\alpha_{2},2}$ and $c_{\alpha_{2},1},c_{\alpha_{2},2}$ can be determined by
	\begin{equation}\label{Q alpha2 condition}
		\int_{-\eta_{1}}^{\eta_{1}}\frac{Q_{\alpha_{2},j}(\zeta)}{R_{\alpha_{2}}(\zeta)}d\zeta=0,\quad \int_{\eta_{2}}^{\eta_{3}}\frac{Q_{\alpha_{2},j}(\zeta)}{R_{\alpha_{2}}(\zeta)}d\zeta=0,\quad j=1,2.
	\end{equation}
    {\begin{lem} \label{Lemma-Property-2}
		The following identities are established
		\begin{equation}\label{x derivate of g2}			\partial_x(tg'_{\alpha_{2}}(\lambda))=1-\frac{Q_{\alpha_2,1}(\lambda)}{R_{\alpha_{2}}(\lambda)},\
			\partial_{x}\left(t\Omega_{\alpha_{2},1}\right)=\oint_{b_1}\frac{Q_{\alpha_{2},1}(\zeta)}{R_{\alpha_{2}}(\zeta)}d\zeta,\
			\partial_{x}\left(t\Omega_{\alpha_{2},0}\right)=\oint_{b_2}\frac{Q_{\alpha_{2},1}(\zeta)}{R_{\alpha_{2}}(\zeta)}d\zeta.
		\end{equation}
	\end{lem}
	\begin{proof}
		From the definition of the function $g'_{\alpha_{2}}(\lambda)$ in (\ref{g2 derivative}), it is seen that
		$$	\partial_x(tg'_{\alpha_{2}}(\lambda))=1-\frac{Q_{\alpha_{2},1}
			(\lambda)}{R_{\alpha_{2}}(\lambda)}+\partial_{\alpha_{2}}\left(12t\frac{Q_{\alpha_{2},2}
			(\lambda)}{R_{\alpha_{2}}(\lambda)}-x\frac{Q_{\alpha_{2},1}(\lambda)}{R_{\alpha_{2}}(\lambda)}\right)
		\partial_x{\alpha_{2}},
		$$
		and since $\partial_{\alpha_{2}}\left(12t\frac{Q_{\alpha_{2},2}(\lambda)}{R_{\alpha_{2}}(\lambda)}
		d\lambda-x\frac{Q_{\alpha_{2},1}(\lambda)}{R_{\alpha_{2}}(\lambda)}d\lambda\right)$ doesn't have singularities at $\pm\alpha_{2}$ and $\infty$, the term $12t\frac{Q_{\alpha_{2},2}(\lambda)}{R_{\alpha_{2}}(\lambda)}
		d\lambda-x\frac{Q_{\alpha_{2},1}(\lambda)}{R_{\alpha_{2}}(\lambda)}d\lambda$ is a holomorphic differential 1-form. Simultaneously, it is normalized to zero on the $a_{\alpha_2,j}~(j=1,2,3)$-cycles, thus by Riemann bilinear relation one has
		$$
		\partial_{\alpha_{2}}\left(12t\frac{Q_{\alpha_{2},2}(\lambda)}{R_{\alpha_{2}}(\lambda)}
		-x\frac{Q_{\alpha_{2},1}(\lambda)}{R_{\alpha_{2}}(\lambda)}\right)=0.
		$$
		On the other hand, it is obvious that
		$$
		\Omega_{\alpha_{2},1}=-\oint_{b_1}(g'(\zeta)+12\zeta^2-4\xi)d\zeta,\ \Omega_{\alpha_{2},0}=-\oint_{b_2}(g'(\zeta)+12\zeta^2-4\xi)d\zeta.
		$$
		So it follows that
		$$
		\begin{aligned}
			&\partial_{x}\left(t\Omega_{\alpha_{2},1}\right)=-\partial_x\left(t\oint_{b_1}(g'(\zeta)+12\zeta^2-4\xi)d\zeta\right)=\oint_{b_1}\frac{Q_{\alpha_{2},1}(\zeta)}{R_{\alpha_{2}}(\zeta)}d\zeta,\\
			&\partial_{x}\left(t\Omega_{\alpha_{2},0}\right)=-\partial_x\left(t\oint_{b_2}(g'(\zeta)+12\zeta^2-4\xi)d\zeta\right)=\oint_{b_2}\frac{Q_{\alpha_{2},1}(\zeta)}{R_{\alpha_{2}}(\zeta)}d\zeta.\\
		\end{aligned}
		$$
	\end{proof}}
	In particular, the function $g_{\alpha_{2}}(\lambda)+4\lambda^3-4\xi\lambda$ is endowed with the behavior $\mathcal{O}((\lambda\pm\alpha_{2})^{\frac{3}{2}})$ as $\lambda\to\pm \alpha_{2}$ if and only if $\xi$ and $\alpha_{2}$ satisfy the following relationship:
	\begin{equation}\label{alpha2 formular}
		\xi=3\frac{Q_{\alpha_2,2}(\pm \alpha_{2})}{Q_{\alpha_2,1}(\pm \alpha_{2})}.
	\end{equation}
	Also, $g_{\alpha_{2}}(\lambda)+4\lambda^3-4\xi\lambda=\mathcal{O}((\lambda\pm\alpha_{2})^{\frac{3}{2}})$ as $\lambda\to\pm \alpha_{2}$ implies that $g'_{\alpha_{2}}(\alpha_{2})+12\alpha_{2}^2-4\xi=0$, thus the equation (\ref{alpha2 formular}) is established. Conversely, if the equation (\ref{alpha2 formular}) holds, the equation (\ref{g2 derivative}) can be rewritten as
	\begin{equation}\label{eq:dg}
	\begin{aligned}
		g'_{\alpha_{2}}(\lambda)+12\lambda^2-4\xi&=12\frac{Q_{\alpha_{2},2}(\lambda)-Q_{\alpha_{2},2}(\alpha_{2})}{R_{\alpha_{2}}(\lambda)}-4\xi\frac{Q_{\alpha_{2},1}(\lambda)-Q_{\alpha_{2},1}(\alpha_{2})}{R_{\alpha_{2}}(\lambda)}\\
		&=12\frac{(\lambda^2-\alpha_{2}^2)(\lambda^2-\lambda_{1}^2)(\lambda^2-\lambda_{2}^2)}{R_{\alpha_{2}}(\lambda)},
	\end{aligned}
	\end{equation}
	where $\lambda_{1}\in(0,\eta_{1})$ and $\lambda_{2}\in(\eta_{2},\eta_{3})$ due to the normalization condition (\ref{Q alpha2 condition}). Consequently, it follows $g'_{\alpha_{2}}(\lambda)+12\lambda^2-4\xi=\mathcal{O}(\sqrt{\lambda-\alpha_{2}})$ as $\lambda\to \alpha_{2}$. In addition, the expression of $g_{\alpha_{2}}(\lambda)$ is
	\begin{equation}\label{g2}		g_{\alpha_{2}}(\lambda)=-4\lambda^3+4\xi\lambda+12\int_{\alpha_{2}}^{\lambda}\frac{Q_{\alpha_{2},2}(\zeta)}{R_{\alpha_{2}}(\zeta)}d\zeta-4\xi\int_{\alpha_{2}}^{\lambda}\frac{Q_{\alpha_{2},1}(\zeta)}{R_{\alpha_{2}}(\zeta)}d\zeta.
	\end{equation}
	It is immediate that $g_{\alpha_{2}}(\lambda)$ satisfies the first condition in (\ref{condition for g2}).
    
{
As in the case of the modulated one-phase wave region, the Whitham equation~\eqref{alpha2 formular} is strictly hyperbolic; see~\cite{Lev88}. By following the approach developed in \cite{GravaTian2002}, we conclude that $\alpha_2$ is a monotonically increasing function of $\xi$.
More precisely, we introduce the following Abelian differentials of the second kind, referred to as the quasi-momentum and quasi-energy differentials:
\[
dp=\frac{Q_{\alpha_2,1}(\lambda)}{R_{\alpha_2}(\lambda)}\,d\lambda,
\qquad
dq=\frac{Q_{\alpha_2,2}(\lambda)}{R_{\alpha_2}(\lambda)}\,d\lambda,
\]
and define
\[
d\varphi=\bigl(g'_{\alpha_2}(\lambda)+12\lambda^2-4\xi\bigr)\,d\lambda
=12\,dq-4\xi\,dp.
\]
\begin{lem}
    Suppose the parameter $\alpha_2$ satisfies the equation \eqref{alpha2 formular}, then for $\alpha_2\in(\eta_3,\eta_4)$, $\alpha_2$ is a monotonically increasing function of $\xi$.
\end{lem}
\begin{proof}
By equation~\eqref{eq:dg}, we obtain
\begin{equation*}
\partial_{\xi}(d\varphi)
=-\left(
\frac{2\lambda_1\,\partial_{\xi}\lambda_1}{\lambda^2-\lambda_1^2}
+\frac{2\lambda_2\,\partial_{\xi}\lambda_2}{\lambda^2-\lambda_2^2}
+\frac{\alpha_2\,\partial_{\xi}\alpha_2}{\lambda^2-\alpha_2^2}
\right)d\varphi.
\end{equation*}
On the other hand, by Lemma~\ref{Lemma-Property-2}, we have
\(
\partial_{\xi}(d\varphi)=-4\,dp
\).
Comparing the two expressions yields the identity
\begin{equation}\label{eq:alpha-xi}
4\frac{dp}{d\varphi}
=
\frac{2\lambda_1\,\partial_{\xi}\lambda_1}{\lambda^2-\lambda_1^2}
+\frac{2\lambda_2\,\partial_{\xi}\lambda_2}{\lambda^2-\lambda_2^2}
+\frac{\alpha_2\,\partial_{\xi}\alpha_2}{\lambda^2-\alpha_2^2}
=
\frac{Q_{\alpha_2,1}(\lambda)}{3(\lambda^2-\alpha_2^2)(\lambda^2-\lambda_1^2)(\lambda^2-\lambda_2^2)}.
\end{equation}
Since $Q_{\alpha_2,1}(\lambda)$ is even, has zeros in $(0,\eta_1)$ and $(\eta_2,\eta_3)$ for $\lambda>0$, and satisfies $Q_{\alpha_2,1}(\lambda)\sim \lambda^4$ as $\lambda\to\infty$, we conclude that $Q_{\alpha_2,1}(\alpha_2)>0$.
Taking the residue at $\lambda=\alpha_2$ in~\eqref{eq:alpha-xi}, it follows that
\begin{equation}\label{genuine-nonlinear}
\partial_{\xi}\alpha_2
=
\frac{Q_{\alpha_2,1}(\alpha_2)}{6\alpha_2(\alpha_2^2-\lambda_1^2)(\alpha_2^2-\lambda_2^2)}>0.
\end{equation}
\end{proof}
Since $dp$ is the normalized Abelian differential of the second kind, we have $Q_{\alpha_2,1}(\alpha_2)>0$. As a result, \( \alpha \) is a monotone increasing function of \( \xi \) with \( \alpha_{2} \in (\eta_{3},\eta_{4}) \). Thus, we define  
}


\par 

\begin{equation}\label{xi crit}
    \xi_{crit}^{(2)}=3\frac{Q_{\alpha_2,2}(\eta_{3})}{Q_{\alpha_2,1}(\eta_{3})},\quad
    \xi_{crit}^{(3)}=3\frac{Q_{\alpha_2,2}(\eta_{4})}{Q_{\alpha_2,1}(\eta_{4})}.
\end{equation}
More precisely, note that (\ref{alpha2 formular}) also depends on the Riemann surface \( R_{\alpha_2}(\lambda) \), which is given by
\[
    \xi=3\frac{Q_{\alpha_2,2}(\pm\alpha_2;\eta_1,\eta_2,\eta_3,\alpha_2)}{Q_{\alpha_2,1}(\pm{\alpha_2};\eta_1,\eta_2,\eta_3,\alpha_2)}
\]
and the critical values \( \xi_{crit}^{(2)} \) and \( \xi_{crit}^{(3)} \) can be rewritten as
\begin{equation}\label{xi2-epsilon}
    \xi_{crit}^{(2)}=\lim_{\epsilon\to0}3\frac{Q_{\alpha_2,2}(\pm(\eta_3+\epsilon);\eta_1,\eta_2,\eta_3,\eta_3+\epsilon)}{Q_{\alpha_2,1}(\pm(\eta_3+\epsilon);\eta_1,\eta_2,\eta_3,\eta_3+\epsilon)},
\end{equation}
\[
    \xi_{crit}^{(3)}=3\frac{Q_{\alpha_2,2}(\pm\eta_4;\eta_1,\eta_2,\eta_3,\eta_4)}{Q_{\alpha_2,1}(\pm{\eta_4};\eta_1,\eta_2,\eta_3,\eta_4)}.
\]
By the genuine nonlinearity condition in (\ref{genuine-nonlinear}), for any \( \epsilon>0 \), it follows that 
\[
    3\frac{Q_{\alpha_2,2}(\pm(\eta_3+\epsilon);\eta_1,\eta_2,\eta_3,\eta_3+\epsilon)}{Q_{\alpha_2,1}(\pm(\eta_3+\epsilon);\eta_1,\eta_2,\eta_3,\eta_3+\epsilon)}<\xi_{crit}^{(3)}.
\]
If the limit exists as \( \epsilon\to0 \), then the inequality \( \xi_{crit}^{(3)}<\xi_{crit}^{(2)} \) is established. In fact, consider the Riemann surface 
\[
    \mathcal{S}_{\alpha_2}^{(\epsilon)}:=\{(\lambda,\eta) \mid \eta^2=(\lambda^2-\eta_1^2)(\lambda^2-\eta_2^2)(\lambda^2-\eta_3^2)(\lambda^2-(\eta_3+\epsilon)^2)\}.
\]
As \( \epsilon \to 0 \), the algebraic curve associated with \( \mathcal{S}_{\alpha_2}^{(\epsilon)} \) develops two nodes, causing the genus of the limiting Riemann surface to degenerate to \( 1 \). By the expansion given in \cite{Tamara, Fay}, the limit of (\ref{xi2-epsilon}) exists, yielding
{
\begin{equation}\label{xi2}
   \begin{aligned}
\xi_{crit}^{(2)} &= 3\lim_{\epsilon\to0}\frac{Q_{\alpha_2,2}(\pm(\eta_3+\epsilon);\eta_1,\eta_2,\eta_3,\eta_3+\epsilon)}{Q_{\alpha_2,1}(\pm(\eta_3+\epsilon);\eta_1,\eta_2,\eta_3,\eta_3+\epsilon)}\\
&= 3\frac{\int_{-R_{\eta_2}(\eta_3)}^{R_{\eta_2}(\eta_3)}\frac{Q_{\eta_2,2}(\zeta)}{R_{\eta_2}(\zeta)}d\zeta}{\int_{-R_{\eta_2}(\eta_3)}^{R_{\eta_2}(\eta_3)}\frac{Q_{\eta_2,1}(\zeta)}{R_{\eta_2}(\zeta)}d\zeta},
   \end{aligned}
\end{equation}
where \( R_{\eta_2}(\lambda) = \sqrt{(\lambda^2-\eta_1^2)(\lambda^2-\eta_2^2)} \) and
\[
Q_{\eta_2,1}(\lambda) = \lambda^2 + c_{\eta_2,1}, \quad Q_{\eta_2,2}(\lambda) = \lambda^4 - \frac{1}{2}(\eta_2^2+\eta_1^2)\lambda^2 + c_{\eta_2,2},
\]
with 
\[
c_{\eta_2,1} = -\eta_2^2 + \eta_2^2\frac{E(m_{\eta_2})}{K(m_{\eta_2})}, \quad
c_{\eta_2,2} = \frac{1}{3} \eta_2^2\eta_1^2 + \frac{1}{6}(\eta_2^2+\eta_1^2)c_{\eta_2,1}, \quad
m_{\eta_2} = \frac{\eta_1}{\eta_2}.
\]}
{
Moreover, using the formula in \cite{Tamara}, we obtain that
\[
\int_{-R_{\eta_2}(\eta_3)}^{R_{\eta_2}(\eta_3)}\frac{Q_{\eta_2,1}(\zeta)}{R_{\eta_2}(\zeta)}d\zeta = -4\left(\oint_A\frac{d\zeta}{R_{\eta_2}(\zeta)}\right)^{-1} \oint_A \frac{d\zeta}{R_{\eta_2}(\zeta)} \frac{R_2(\eta_3)}{\zeta^2-\eta_3^2},
\]
and
\[
\int_{-R_{\eta_2}(\eta_3)}^{R_{\eta_2}(\eta_3)}\frac{Q_{\eta_2,2}(\zeta)}{R_{\eta_2}(\zeta)}d\zeta = -\frac{4}{3} \left(-R_{\eta_2}(\eta_3) + \frac{\eta_1^2+\eta_2^2}{2} \left(\oint_A\frac{d\zeta}{R_{\eta_2}(\zeta)}\right)^{-1} \oint_A \frac{d\zeta}{R_{\eta_2}(\zeta)} \frac{R_2(\eta_3)}{\zeta^2-\eta_3^2} \right),
\]
where \( A \) is the contour shown in Figure \ref{1genus}. Consequently, the expression for \( \xi^{(2)}_{crit} \) simplifies to
\begin{equation}\label{xi2-final}
\xi^{(2)}_{crit} = \frac{\eta_1^2+\eta_2^2}{2} - \frac{\int_{0}^{\eta_1}\frac{d\zeta}{R_{\eta_2}(\zeta)}}{\int^{\eta_1}_{0}\frac{d\zeta}{R_{\eta_2}(\zeta)(\zeta^2-\eta_3^2)}} = \frac{\eta_1^2+\eta_2^2}{2} + \eta_3^2{\frac{K(m_{\eta_2})}{\Pi(\kappa,m_{\eta_2})}},
\end{equation}
where \( \Pi(\kappa,m_{\eta_2}) := \int_{0}^{\frac{\pi}{2}} \frac{d\theta}{(1-\kappa\sin^2\theta)\sqrt{1-m_{\eta_2}^2\sin^2\theta}} \)
is the complete elliptic integral of the third kind, with \( \kappa = \frac{\eta_1^2}{\eta_3^2} \) and \( m_{\eta_2} = \frac{\eta_1}{\eta_2} \).}

Notably, it is found that
\[
\xi_{crit}^{(2)} - \xi_{crit}^{(1)} = \eta_3^2{\frac{K(m_{\eta_2})}{\Pi(\kappa,m_{\eta_2})}} - (\eta_2^2-\eta_1^2)\frac{K(m_{\eta_2})}{E(m_{\eta_2})}.
\]
Since \( {\Pi(\kappa,m)\to \frac{E(m)}{1-m^2}} \), {as $\kappa\to m^2$,}  { and let
$$
\Delta_{\xi}(\alpha)=K(m_{\eta_2})\left(\frac{\alpha^2}{\Pi\left(\frac{\eta_1^2}{\alpha^2};m_{\eta_2}\right)}-\frac{\eta_2^2-\eta_1^2}{E(m_{\eta_2})}\right).
$$
Notice that $\Delta_{\xi}(\eta_2)=0$, which means that $\xi_{crit}^{(2)} = \xi_{crit}^{(1)}$ as $\eta_3=\eta_2$, and by \( \partial_{\kappa} \Pi(\kappa,m_{\eta_2}) > 0 \) for \( 0 < \kappa < 1 \), we conclude that $\partial_{\alpha}\Delta_{\xi}(\alpha)>0$ and \( \xi^{(2)}_{crit} > \xi^{(1)}_{crit} \) for $\eta_3>\eta_2$. }

{}
	\par	
	On the other hand, based on the expression of the function $g_{\alpha_{2}}(\lambda)$ in (\ref{g2}), the inequalities like (\ref{g1 signature}) are obtained in the following lemma.
	
	\begin{lem}\label{Properties-1}
		For $\xi_{crit}^{(2)}<\xi<\xi_{crit}^{(3)}$, the inequalities below hold
		\begin{equation}\label{g2 signature}
			\begin{aligned}
				&\re\left(g_{\alpha_{2}}(\lambda)+4\lambda^3-4\xi\lambda\right)>0, &&\lambda\in(\alpha_2,\eta_4)\cup\mathcal{C}_{2,\alpha_{2}}\setminus\{-\eta_1,-\eta_2,-\eta_3,-\alpha_2\},\\
				&\re\left(g_{\alpha_{2}}(\lambda)+4\lambda^3-4\xi\lambda\right)<0, &&\lambda\in(-\eta_4,-\alpha_2)\cup\mathcal{C}_{1,\alpha_{2}}\setminus\{\eta_1,\eta_2,\eta_3,\alpha_2\},\\
			\end{aligned}
		\end{equation}
		where the contours $\mathcal{C}_{j,\alpha_{2}}$ for $j=1,2$ are depicted in Figure \ref{S alpha2}.
	\end{lem}
	\begin{proof}
		Similarly, recall that
		$$
		g'_{\alpha_{2}}(\lambda)+12\lambda^2-4\xi=12\frac{(\lambda^2-\alpha_{2}^2)(\lambda^2-\lambda_{1}^2)(\lambda^2-\lambda_{2}^2)}{R_{\alpha_{2}}(\lambda)},
		$$
		with $\lambda_{1}\in[0,\eta_{1}]$ and $\lambda_{2}\in[\eta_{2},\eta_{3}]$. As a result, it can be 
        derived that
		$$
		g'_{\alpha_{2},+}(\lambda)+12\lambda^2-4\xi=
		\begin{cases}
			\begin{aligned}
				&12\frac{\sqrt{\lambda^2-\alpha_{2}^2}(\lambda^2-\lambda_{1}^2)(\lambda^2-\lambda_{2}^2)}
				{\sqrt{(\lambda^2-\eta^2_{1})(\lambda^2-\eta^2_{2})(\lambda^2-\eta^2_{3})}}, &&\lambda\in(\alpha_{2},\eta_{4}),\\
				&12\frac{i\sqrt{\alpha_{2}^2-\lambda^2}(\lambda^2-\lambda_{1}^2)(\lambda^2-\lambda_{2}^2)}
				{\sqrt{(\lambda^2-\eta^2_{1})(\lambda^2-\eta^2_{2})(\lambda^2-\eta^2_{3})}},
				&&\lambda\in(\eta_{3},\alpha_{2}),\\
				&12\frac{i\sqrt{\alpha_{2}^2-\lambda^2}(\lambda^2-\lambda_{1}^2)(\lambda_{2}^2-\lambda^2)}
				{\sqrt{(\lambda^2-\eta^2_{1})(\eta^2_{2}-\lambda^2)(\eta^2_{3}-\lambda^2)}},
				&&\lambda\in(\eta_{1},\eta_{2}),
			\end{aligned}
		\end{cases}
		$$
		which implies that $\re\left(g_{\alpha_{2}}(\lambda)+4\lambda^3-4\xi\lambda\right)>0$ for $\lambda\in(\alpha_{2},\eta_{4})$ and $\Im\left(g'_{\alpha_{2},+}(\lambda)+12\lambda^2-4\xi\right)>0$ for $\lambda\in\Sigma_{1}\cup\Sigma_{3_{\alpha_{2}}}$. Together with the Cauchy-Riemann equation, it follows that $\re\left(g_{\alpha_{2}}(\lambda)+4\lambda^3-4\xi\lambda\right)<0$ for $\lambda\in\mathcal{C}_{1,\alpha_{2}}\setminus\{\eta_1,\eta_2,\eta_3,\alpha_2\}$. The other cases in (\ref{g2 signature}) can also be proved by the similar way.
	\end{proof}
	\par
	To deform the RH problem for $Y(x,t;\lambda)$ into the model problem, introduce
	$$
	T_{\alpha_{2}}(\lambda)=Y(\lambda){e^{tg_{\alpha_{2}}(\lambda)\sigma_3}}f_{\alpha_{2}}(\lambda)^{\sigma_3},
	$$
	where $f_{\alpha_{2}}(\lambda)$ is subject to the following scalar RH problem:
	\begin{equation}\label{RHP for falpha2}
		\begin{aligned}
			&f_{\alpha_{2},+}(\lambda)f_{\alpha_{2},-}(\lambda)={r(\lambda)},&&\lambda\in\Sigma_{\{1,3_{\alpha_{2}}\}},\\
			&f_{\alpha_{2},+}(\lambda)f_{\alpha_{2},-}(\lambda)=\frac{1}{r(\lambda)},&&\lambda\in\Sigma_{\{2,4_{\alpha_{2}}\}},\\
			&\frac{f_{\alpha_{2},+}(\lambda)}{f_{\alpha_{2},-}(\lambda)}=e^{\Delta_0,\alpha_{2}},&&\lambda\in [-\eta_1,\eta_1],\\
			&\frac{f_{\alpha_{2},+}(\lambda)}{f_{\alpha_{2},-}(\lambda)}=e^{\Delta_1,\alpha_{2}},&&\lambda\in [\eta_2,\eta_3]\cup[-\eta_3,-\eta_2],\\
			&f(\lambda)=1+\mathcal{O}\left(\frac{1}{\lambda}\right),&&\lambda\to\infty.
		\end{aligned}
	\end{equation}
	In a similar manner, the function \( f_{\alpha_{2}}(\lambda) \) can be derived akin to (\ref{f formular}). The normalization condition then indicates the values of \(\Delta_{j,\alpha_{2}},~j=0,1\), which correspond to (\ref{f infty}). All the calculations here are omitted for simplicity.
	\begin{figure}[h]
		\centering
		\begin{tikzpicture}[>=latex]
			\draw[lightgray,very thick,dashed] (-8.5,0) to (-7.5,0);
			\draw[lightgray,very thick] (-7.5,0) to (-6.5,0) ;
			\draw[->,dashed,very thick] (-7.0,-1.0) to (-7.0,-0.1);
			\draw[->,very thick,lightgray] (-7.5,0) to (-7,0)node[black,anchor=south east]at(-6,-2.1){\small $\begin{pmatrix}
					1 & 0 \\
					\textcolor{gray}{(V_T)_{21}} & 1
				\end{pmatrix}$};
			\filldraw[lightgray] (-7.5,0) node[black,below=1mm]{$-\eta_{4}$} circle (1.5pt);
			\filldraw[lightgray] (-6.5,0) node[black,below=1mm]{$-\alpha_2$} circle (1.5pt);
			\draw[-,very thick] (-6.5,0) to (-4.5,0);
			\draw[->,very thick] (-6.5,0) to (-6.0,0)node[black,right=-1mm]{\small $\begin{pmatrix}
					0 & i \\
					i & 0
				\end{pmatrix}$};
			\draw[->,very thick] (-4.5,0) to (-4.0,0);
			\draw[-,very thick] (-4.5,0) to (-3.8,0)node[black,above=12mm]{ $e^{(t{\Omega}_{1,\alpha_{2}}+{\Delta}_{1,\alpha_{2}})\sigma_3}$};
			\draw[-,very thick] (-4.5,0) to (-3.3,0)node[black,below=8mm]{\small $\begin{pmatrix}
					1 & \textcolor{gray}{\frac{-i}{\hat r(\lambda)f_{\alpha_{1}}^2(\lambda)}e^{-2t(g_{\alpha_{1}}(\lambda)+4\lambda^3-4\xi\lambda)}} \\
					0 & 1
				\end{pmatrix}$};
			\draw[->,dashed,very thick] (-4.0,1.2) to (-4.0,0.2);
			\draw[-,very thick] (-4.5,0) to (-3.5,0);
			\filldraw[black] (-4.5,0) node[black,below=1mm]{$-\eta_{3}$} circle (1.5pt);
			\filldraw[black] (-3.5,0) node[black,below=1mm]{$-\eta_2$} circle (1.5pt);
			\draw[-,very thick,lightgray] (-6.5,0) .. controls (-6,1.0) and (-5,1.0).. (-4.5,0);
			\draw[-,very thick,lightgray,rotate around x=180] (-6.5,0) .. controls (-6,1.0) and (-5,1.0).. (-4.5,0);
			\draw[->,very thick,lightgray] (-5.6,0.75) to (-5.4,0.75) node [above,black] {\small $\mathcal{C}_{2,\alpha_{2}}$};
			\draw[<-,very thick,lightgray,rotate around x=180] (-5.6,0.75) to (-5.4,0.75);
			\draw[-,very thick,black] (-3.5,0) to (-0.5,0);
			\draw[->,very thick,black] (-3.5,0) to (-2.5,0) node[black,right=-1mm]{\small $\begin{pmatrix}
					0 & i \\
					i & 0
				\end{pmatrix}$};
			\filldraw[black] (-0.5,0) node[black,below=1mm]{$-\eta_1$} circle (1.5pt);
			\draw[-,very thick,lightgray] (-3.5,0) .. controls (-2.5,1.15) and (-1.5,1.15).. (-0.5,0);
			\draw[->,very thick,lightgray] (-2.1,0.85) to (-2.0,0.85) node[black,above]{\small $\mathcal{C}_{2,\alpha_{2}}$};
			\draw[-,very thick,lightgray,rotate around x=180] (-3.5,0) .. controls (-2.5,1.15) and (-1.5,1.15).. (-0.5,0) ;
			\draw[<-,very thick,lightgray,rotate around x=180] (-2.3,0.85) to (-2.2,0.85) ;
			\draw [very thick] (-0.5,0) to (0.5,0);
			\draw [->,very thick] (-0.5,0) to (0,0)node[black,above=12mm]{ $e^{(t{\Omega}_{0,\alpha_{2}}+{\Delta}_{0,\alpha_{2}})\sigma_3}$};
			\draw[->,dashed,very thick] (0,1.2) to (0,0.2);
			\draw[lightgray,very thick,dashed] (7.5,0) to (8.5,0);
			\draw[lightgray,very thick] (6.5,0) to (7.5,0) ;
			\draw[->,very thick,lightgray] (6.5,0) to (7.2,0)node[black,below=8mm]{\small $\begin{pmatrix}
					1 & \textcolor{gray}{(V_T)_{12}} \\
					0 & 1
				\end{pmatrix}$};
			\draw[->,dashed,very thick] (7.0,-0.8) to (7.0,-0.1);
			\draw[-,very thick](3.5,0) to (4.0,0) node [black,above=12mm]{ $e^{(t{\Omega}_{1,\alpha_{2}}+{\Delta}_{1,\alpha_{2}})\sigma_3}$};
			\draw[->,dashed,very thick] (4.0,1.2) to (4.0,0.2);
			\filldraw[lightgray] (7.5,0) node[black,below=1mm]{$\eta_{4}$} circle (1.5pt);
			\filldraw[black] (6.5,0) node[black,below=1mm]{$\alpha_{2}$} circle (1.5pt);
			\draw[->,very thick] (3.5,0) to (4.0,0);
			\draw[-,very thick] (3.5,0) to (4.5,0);
			\filldraw[black] (4.5,0) node[black,below=1mm]{$\eta_{3}$} circle (1.5pt);
			\filldraw[black] (3.5,0) node[black,below=1mm]{$\eta_{2}$} circle (1.5pt);
			\draw[-,very thick,black] (0.5,0) to (3.5,0);
			\draw[->,very thick,black] (0.5,0) to (1.5,0) node[black,right=-1mm]{\small $\begin{pmatrix}
					0 & -i \\
					-i & 0
				\end{pmatrix}$};
			\filldraw[black] (0.5,0) node[black,below=1mm]{$\eta_1$} circle (1.5pt);
			\draw[-,very thick,black] (4.5,0) to (6.5,0);
			\draw[->,very thick,black] (4.5,0) to (4.9,0)node[black,right=-2.5mm]{\small $\begin{pmatrix}
					0 & -i \\
					-i & 0
				\end{pmatrix}$};
			\draw[-,very thick,lightgray] (4.5,0) .. controls (5,1) and (6,1).. (6.5,0);
			\draw[-,very thick,lightgray,rotate around x=180] (4.5,0) .. controls (5,1) and (6,1).. (6.5,0);
			\draw[->,very thick,lightgray] (5.5,0.75) to (5.6,0.75) node [above,black] {\small $\mathcal{C}_{1,\alpha_{2}}$};
			\draw[<-,very thick,lightgray,rotate around x=180] (5.4,0.75) to (5.5,0.75);
			\draw[-,very thick,lightgray] (0.5,0) .. controls (1.5,1.15) and (2.5,1.15).. (3.5,0);
			\draw[->,very thick,lightgray] (2,0.855) to (2.1,0.855)node[black,above]{\small $\mathcal{C}_{1,\alpha_{2}}$} ;
			\draw[-,very thick,lightgray,rotate around x=180] (0.5,0) .. controls (1.5,1.15) and (2.5,1.15).. (3.5,0);
			\draw[->,very thick,lightgray,rotate around x=180] (2.1,0.85) to (2.0,0.85)node[black,anchor=south west]at(1,2.1){\small $\begin{pmatrix}
					1 & 0 \\
					\textcolor{gray}{\frac{if_{\alpha_{1}}^2(\lambda)}{\hat r(\lambda)}e^{2t(g_{\alpha_{1}}(\lambda)+4\lambda^3-4\xi\lambda)}} & 1
				\end{pmatrix}$} ;
		\end{tikzpicture}
		\caption{{\protect\small
				The jump contours for \( S_{\alpha_{2}}(\lambda) \) and the associated jump matrices: the gray terms in the matrices vanish exponentially as \( t \to +\infty \), and the gray contours also vanish as \( t \to +\infty \). Here $(V_T)_{12}=-\frac{ir(\lambda)}{ f^2(\lambda)}e^{-2t(g(\lambda)+4\lambda^3-4\xi\lambda)}$ and $(V_T)_{21}=ir(\lambda)f^2(\lambda)e^{2t(g(\lambda)+4\lambda^3-4\xi\lambda)}$. }}
		\label{S alpha2}
	\end{figure}
	
	In the same way, open lenses of the RH problem for $T_{\alpha_{2}}$ by the transformation
	\begin{equation}
		S_{\alpha_{2}}(\lambda)=
		\begin{cases}
			T_{\alpha_{2}}(\lambda)\begin{pmatrix}
				1 & 0 \\
				\frac{if_{\alpha_{1}}^2(\lambda)}{\hat r(\lambda)}e^{2t(g_{\alpha_{1}}(\lambda)+4\lambda^3-4\xi\lambda)} & 1
			\end{pmatrix}, & {\rm inside~the~contour}~\mathcal{C}_{1,\alpha_{2}},
			\\
			T_{\alpha_{2}}(\lambda)\begin{pmatrix}
				1 & \frac{-i}{\hat r(\lambda)f_{\alpha_{1}}^2(\lambda)}e^{-2t(g_{\alpha_{1}}(\lambda)+4\lambda^3-4\xi\lambda)} \\
				0 & 1
			\end{pmatrix}, & {\rm inside~the~contour}~\mathcal{C}_{2,\alpha_{2}},
			\\
			T_{\alpha_{2}}(\lambda), & {\rm elsewhere}.
		\end{cases}
	\end{equation}
	The jump matrices and contours are illustrated in Figure \ref{S alpha2}. As $t\to+\infty$, the gray contours in Figure \ref{S alpha2} vanish exponentially according to Lemma \ref{Properties-1}. Once again
	for $t\to+\infty$, we arrive at the model problem for $S_{\alpha_{2}}^{\infty}(\lambda)$ below
	\begin{equation}\label{S infty alpha2 RHP}
		{S}_{\alpha_{2},+}^{\infty}(\lambda)={S}_{\alpha_{2},-}^{\infty}(\lambda) \begin{cases}{\begin{pmatrix}
					e^{t {\Omega_{0,\alpha_{2}}}+{\Delta}_{0,\alpha_{2}}} & 0 \\
					0 & e^{-t {\Omega_{0,\alpha_{2}}}-{\Delta}_{0,\alpha_{2}}}
			\end{pmatrix}}, & \lambda \in[-\eta_1, \eta_1], \\
			\begin{pmatrix}
				e^{t {\Omega_{1,\alpha_{2}}}+{\Delta}_{1,\alpha_{2}}} & 0 \\
				0 & e^{-t {\Omega_{1,\alpha_{2}}}-{\Delta}_{1,\alpha_{2}}}
			\end{pmatrix}, & \lambda \in[\eta_2, \eta_3]\cup[-\eta_3, -\eta_2], \\
			{\begin{pmatrix}
					0 & -i \\
					-i & 0
			\end{pmatrix}}, & \lambda \in \Sigma_{1}\cup\Sigma_{3_{\alpha_1}}, \\
			{\begin{pmatrix}
					0 & i \\
					i & 0
			\end{pmatrix}}, & \lambda \in \Sigma_{2}\cup\Sigma_{ 4_{\alpha_1}},\end{cases}
	\end{equation}
	and
	$$
	{S}_{\alpha_{2}}^{\infty}(\lambda)=\begin{pmatrix}
		1&1
	\end{pmatrix}+\mathcal{O}\left(\frac{1}{\lambda}\right),\quad \lambda\to\infty.
	$$
	\par
	Furthermore, as the solution of the model problem for $S^{\infty}(\lambda)$ in (\ref{solution of lambda}), the solution of $S_{\alpha_{2}}^{\infty}(\lambda)$ can be derived directly by
	\begin{equation}	S_{\alpha_{2}}^{\infty}(\lambda)=\gamma_{\alpha_{2}}(\lambda)\frac{\Theta(0;\hat{\tau}_{\alpha_{2}})}{\Theta\left(\frac{\Omega_{\alpha_{2}}}{2\pi i};\hat{\tau}_{\alpha_{2}}\right)}
		\begin{pmatrix}
			\frac{\Theta\left(J_{\alpha_{2}}(\lambda)-d_{\alpha_{2}}+\frac{\Omega_{\alpha_{2}}}{2\pi i};\hat{\tau}_{\alpha_{2}}\right)}{\Theta\left(J_{\alpha_{2}}(\lambda)-d_{\alpha_{2}};\hat{\tau}_{\alpha_{2}}\right)}&
			\frac{\Theta\left(-J_{\alpha_{2}}(\lambda)-d_{\alpha_{2}}+\frac{\Omega_{\alpha_{2}}}{2\pi i};\hat{\tau}_{\alpha_{2}}\right)}{\Theta\left(-J_{\alpha_{2}}(\lambda)-d_{\alpha_{2}};\hat{\tau}_{\alpha_{2}}\right)}
		\end{pmatrix},
	\end{equation}
	where
	$\gamma_{\alpha_{2}}(\lambda)=\left(\frac{(\lambda^2-\eta_1^2)(\lambda^2-\eta_3^2)}{(\lambda^2-\eta_2^2)(\lambda^2-\alpha_{2}^2)}\right)^{\frac{1}{4}}
	$ and $\Omega_{\alpha_{2}}=\begin{pmatrix}
		{t\Omega_{\alpha_{2},1}+{\Delta_{\alpha_{2},1}}}&{t\Omega_{\alpha_{2},0}+{\Delta_{\alpha_{2},0}}}
	\end{pmatrix}^T$. Before computing the expression of large-time asymototics of $u(x,t)$ in this region, introduce the corresponding Riemann surface of genus three with ${a}_{\alpha_{2},j} ,{b}_{\alpha_{2},j}$-cycles for $j=1,2,3$, which is depicted in Figure \ref{mathcal S alpha} and
	$$ \mathcal{S}_{\alpha_{2}}:=\{(\lambda,\eta)|\eta^2=(\lambda^2-\eta_1^2)(\lambda^2-\eta_2^2)
	(\lambda^2-\eta_3^2)(\lambda^2-\alpha_{2}^2)\}.
	$$
	\begin{figure}
		\centering
		\includegraphics[height=5cm]{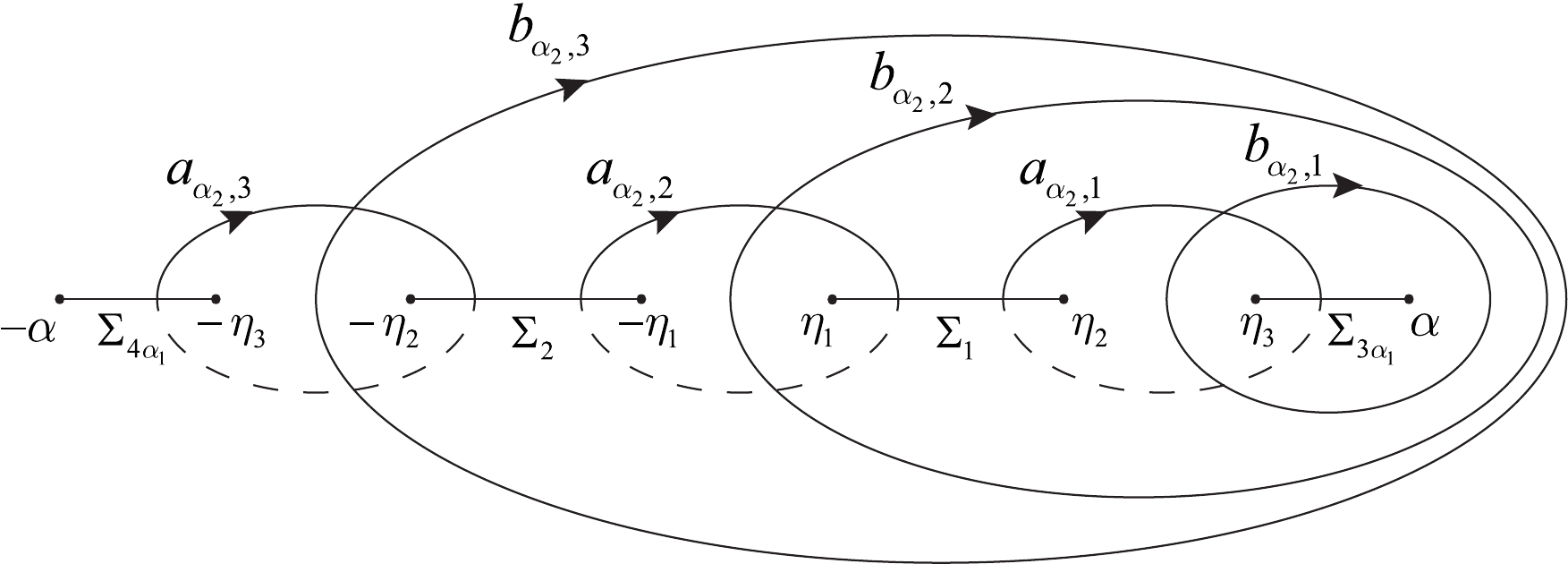}
		\caption{{\protect\small The Riemann surface $\mathcal{S}_{\alpha_{2}}$ of genus three and its basis $\{{a}_{\alpha_{2},j} ,{b}_{\alpha_{2},j}\}~(j=1,2,3)$ of circles.}}
		\label{mathcal S alpha}
	\end{figure}
	\par
	The normalized holomorphic differentials associated with the Riemann surface $\mathcal{S}_{\alpha_{2}}$ are denoted as $\omega_{\alpha_{2},j}~(j=1,2,3)$. Suppose the period matrix of $\omega_{\alpha_{2},j}$ is $\tau_{\alpha_{2}}:=(\tau_{\alpha_{2},ij})_{3\times3}$ with the symmetry $\tau_{\alpha_{2},11}=\tau_{\alpha_{2},33},\tau_{\alpha_{2},12}=\tau_{\alpha_{2},23}$ like the case in (\ref{properties of tau}). Similarly, define the Jacobi map $J_{\alpha_{2}}(\lambda)$ as
	\begin{equation}\label{Jacobi alpha2}
		J_{\alpha_{2}}(\lambda)=\int_{\alpha_{2}}^{\lambda} \hat\omega_{\alpha_{2}}:=\int_{\alpha_{2}}^{\lambda} \begin{pmatrix}
			\omega_{\alpha_{2},1}+\omega_{\alpha_{2},3}\\
			2\omega_{\alpha_{2},2}
		\end{pmatrix},
	\end{equation}
	and the corresponding period matrix is
	\begin{equation}\label{period matrix of hat tau alpha2}
		\hat\tau_{\alpha_{2}}=	\begin{pmatrix}
			\tau_{\alpha_{2},11}+\tau_{\alpha_{2},31}& \tau_{\alpha_{2},12}+\tau_{\alpha_{2},32}\\
			2\tau_{\alpha_{2},21}& 2\tau_{\alpha_{2},22}
		\end{pmatrix}.
	\end{equation}
	In fact, the Jacobi map $J_{\alpha_{2}}(\lambda)$ satisfies the similar properties in (\ref{Jacobi J jumps}) and (\ref{Jacobi J half periods}) just by replacing $\eta_{4}$ with $\alpha_{2}$. In addition, by the symmetry of $J_{\alpha_{2}}(\lambda)$ like that in (\ref{symmetry of Jacobi}), it is derived that $d_{\alpha_{2}}=d=\frac{e_2+e_1}{2}$. According to the second and third jump conditions in (\ref{g2 jumps}), it follows
	\begin{equation}\label{Omega alpha2}
		\begin{aligned}
			& {\Omega}_{\alpha_{2},1}=24 \int_{\alpha_{2}}^{\eta_3} \frac{Q_{\alpha_{2},2}(\zeta)}{R_{\alpha_{2}}(\zeta)} d \zeta-8 \xi \int_{\alpha_{2}}^{\eta_3} \frac{Q_{\alpha_{2},1}(\zeta)}{R_{\alpha_{2}}(\zeta)} d \zeta, \\
			& {\Omega}_{\alpha_{2},0}={\Omega}_{\alpha_{2},1}+24 \int_{\eta_{2}}^{\eta_1} \frac{Q_{\alpha_{2},2}(\zeta)}{R_{\alpha_{2}}(\zeta)} d \zeta-8 \xi \int_{\eta_{2}}^{\eta_1} \frac{Q_{\alpha_{2},1}(\zeta)}{R_{\alpha_{2}}(\zeta)} d \zeta.
		\end{aligned}
	\end{equation}
	\par
	Moreover, since the reconstruction formula involves the derivative with respect to $x$, the following lemma is necessary.

	Thus for $\xi_{crit}^{(2)}<\xi<\xi_{crit}^{(3)}$, the following theorem holds.
	\begin{thm}
		
		For $\xi=\frac{x}{4t}$, in the region $\xi_{crit}^{(2)}<\xi<\xi_{crit}^{(3)}$, the large-time asymptotic behavior of the solution to the KdV equation with genus two soliton gas potential is described by
		\begin{equation}\label{ modulated 2 genus solution}		u(x,t)=-\left(2b_{\alpha_{2},1}+{\sum_{j=1}^3\eta_j^2+\alpha_{2}^2}+2\partial_x^2\log\left(\Theta\left(\frac{\Omega_{\alpha_{2}}}{2\pi i};\hat \tau_{\alpha_{2}}\right)\right)\right)+\mathcal{O}\left(\frac{1}{t}\right),
		\end{equation}
		where the parameter $\alpha_{2}$ is determined by (\ref{alpha2 formular}).
	\end{thm}
	\begin{proof}
		Recall that
		$$		Y_1(\lambda)=\left(S_{\alpha_{2},1}^{\infty}(\lambda)
		+\frac{(\mathcal{E}_{\alpha_2,1}(x,t))_1}{t\lambda}+\mathcal{O}\left(\frac{1}{\lambda^2}\right)\right)
		e^{-tg_{\alpha_{2}}(\lambda)}f_{\alpha_{2}}(\lambda)^{-1},
		$$
		in which the function $f_{\alpha_{2}}(\lambda)$ behaves
		$$	f_{\alpha_{2}}(\lambda)=1+\frac{f_{\alpha_{2}}^{(1)}(\alpha_{2})}{\lambda}+\mathcal{O}\left(\frac{1}{\lambda^2}\right),
		$$
		with
		$$	f_{\alpha_{2}}^{(1)}(\alpha_{2})=\left(\int_{\eta_1}^{\eta_{2}}+\int_{\eta_3}^{\alpha_{2}}\right)\frac{\zeta^4\log{r(\zeta)}}{R_{\alpha_{2}}(\zeta)} \frac{d\zeta}{\pi i}+\Delta_{\alpha_{2},0}\int_{-\eta_1}^{\eta_1}\frac{\zeta^4}{R_{\alpha_{2}}(\zeta)} \frac{d\zeta}{2\pi i}+\Delta_{\alpha_{2},1}\int_{\eta_2}^{\eta_3}\frac{\zeta^4}{R(\zeta)} \frac{d\zeta}{\pi i},
		$$
		and the term $\frac{(\mathcal{E}_{\alpha_2,1}(x,t))_1}{t\lambda}$ is first entry of the error vector subject to the modulated two-phase wave region. Similar to the modulated one-phase case, one can conclude that the local parametrix near $\pm\alpha_2$ and $\pm\eta_j$ ($j=1,2,3$) can be described by Airy function and modified Bessel function respectively and both of them contribute the error term $\mathcal{O}(t^{-1})$ in the asymptotic behavior of potential $u(x,t)$ for $t\to+\infty$.
		\par
		The derivative of the term $e^{-tg_{\alpha_{2}}(\lambda)}$ has the asymptotics
		$$	\partial_xe^{-tg_{\alpha_{1}}(\lambda)}=-\frac{1}{\lambda}\left[\frac{\eta_{1}^2+\eta_{2}^2+\eta_{3}^2+\alpha_{2}^2}{2}+b_{\alpha_{2},1}\right]+\mathcal{O}\left(\frac{1}{\lambda^2}\right).
		$$
		In addition, from Lemma \ref{Lemma-Property-2}, it follows from
		{Riemann Bilinear relations} \cite{Bertola Riemmansurface} that
		$$
		\sum \res_{\infty_{\pm}}\frac{\omega_{\alpha_{2},1}}{\lambda}=\frac{\partial_x\left(t\Omega_{\alpha_{2},1}\right)}{2\pi i},\  \sum \res_{\infty_{\pm}}\frac{\omega_{\alpha_{2},2}}{\lambda}=\frac{\partial_x(t\Omega_{\alpha_{2},0})}{2\pi i},\  \sum \res_{\infty_{\pm}}\frac{\omega_{\alpha_{2},3}}{\lambda}=\frac{\partial_x(t\Omega_{\alpha_{2},1})}{2\pi i}.
		$$
		Recalling $J_{\alpha_{2}}(\infty)=(e_1+e_2)/2$, the Jacobi map $J_{\alpha_{2}}(\lambda)$ has the asymptotics
		$$
		J_{\alpha_{2}}(\lambda)=\frac{e_1+e_2}{2}-\frac{\begin{pmatrix}
				\partial_x(t\Omega_{\alpha_{2},1})&\partial_x(t\Omega_{\alpha_{2},0})
			\end{pmatrix}^T}{\lambda}+\mathcal{O}\left(\frac{1}{\lambda^2}\right),\quad \lambda\to\infty.
		$$
		Since $\partial_x{\Delta_{\alpha_{2}}}=\mathcal{O}\left(\frac{1}{t}\right)$ as $t\to+\infty$, the Jacobi map $J_{\alpha_{2}}(\lambda)$ can be rewritten as
		$$ J_{\alpha_{2}}(\lambda)=\frac{e_1+e_2}{2}-\frac{\partial_x(\Omega_{\alpha_{2}})}{2\pi i \lambda}+\mathcal{O}\left(\frac{1}{\lambda^2}\right)+\mathcal{O}\left(\frac{1}{t}\right).
		$$
		Thus the expansion of the function $S_{\alpha_{2},1}^{\infty}(\lambda)$ as $\lambda\to\infty$ is expressed by
		$$
		\begin{aligned}		   {S}^{\infty}_{\alpha_{2},1}(\lambda)&=1-\frac{1}{\lambda}\left[\nabla\log\left(\Theta\left(\frac{\Omega_{\alpha_{2}}}{2\pi i};\hat\tau_{\alpha_{2}}\right)\right)-\nabla\log(\Theta(0;\hat\tau_{\alpha_{2}}))\right]\cdot \frac{ \partial_x(\Omega_{\alpha_{2}})}{2\pi i}+\mathcal{O}\left(\frac{1}{t}\right)+\mathcal{O}\left(\frac{1}{\lambda^2}\right),\\
			&=1-\frac{1}{\lambda}\partial_x\log\left(\Theta\left(\frac{\Omega_{\alpha_{2}}}{2\pi i};\hat\tau_{\alpha_{2}}\right)\right)+\mathcal{O}\left(\frac{1}{t}\right)+\mathcal{O}\left(\frac{1}{\lambda^2}\right).
		\end{aligned}
		$$
		\par
		Therefore, the asymptotics $\partial_x f_{\alpha_{2}}^{(1)}=\mathcal{O}\left(\frac{1}{t}\right)$ for $t\to +\infty$ and all the formulae above result in the large-time asymptotic behavior of $u(x,t)$ as
		$$
		u(x,t)=-\left(2b_{\alpha_{2},1}+{\sum_{j=1}^3\eta_j^2+\alpha_{2}^2}+2\partial_x^2\log\left(\Theta\left(\frac{\Omega_{\alpha_{2}}}{2\pi i};\hat \tau_{\alpha_{2}}\right)\right)\right)+\mathcal{O}\left(\frac{1}{t}\right),\quad t\to +\infty.
		$$
	\end{proof}
	
	\subsection{Unmodulated two-phase wave region}
	
	For \(\xi > \xi_{crit}^{(3)}\), the large-time behavior of \(u(x,t)\) is described by an unmodulated two-phase Riemann-Theta function. Similar to the case in Section \ref{unmodulated 1 genus solution}, we only need to modify relevant notations, such as \( g_{\alpha_{2}} \), \( R_{\alpha_{2}}(\lambda) \), \( \mathcal{S}_{\alpha_{2}} \), \( \Omega_{\alpha_{2}} \), and \( \Delta_{\alpha_{2}} \) in Section \ref{modulated 2-genus case} by replacing \(\alpha_{2}\) with \(\eta_{4}\), { for example $R_{\eta_4}:=\sqrt{(\lambda^2-\eta_4^2)(\lambda^2-\eta_2^2)(\lambda^2-\eta_3^2)(\lambda^2-\eta_4^2)}$}. Similarly, one can verify that \( g_{\eta_{4}} \) and \( f_{\eta_{4}} \) can still deform the RH problem for \( Y(x,t;\lambda) \) into a model problem \( S_{\eta_4}^{\infty}(\lambda) \). Indeed, the model problem \( S_{\eta_4}^{\infty}(\lambda) \) is also similar to \( S^{\infty}(\lambda) \) in (\ref{Sinf}) with the same jump contours, but the diagonal matrices are replaced by \( e^{(t{\Omega}_{j,\eta_{4}} + {\Delta}_{j,\eta_{4}})\sigma_3} \) for $j=0,1$. We omit all the details here for brevity. Thus for \(\xi > \xi_{crit}^{(3)}\), the following theorem holds.
	\begin{thm}
		For $\xi=\frac{x}{4t}$, in the region $\xi_{crit}^{(3)}<\xi$, the large-time asymptotic behavior of the solution to the KdV equation with genus two soliton gas potential is described by
		\begin{equation}\label{unmodulated 2 genus solution}		u(x,t)=-\left(2b_{\eta_{4},1}+{\sum_{j=1}^4\eta_j^2}+2\partial_x^2\log\left(\Theta\left(\frac{\Omega_{\eta_{4}}}{2\pi i};\hat \tau_{\eta_{4}}\right)\right)\right)+\mathcal{O}\left(\frac{1}{t}\right),
		\end{equation}
		where $\hat \tau_{\eta_{4}}=\hat{\tau}$ in (\ref{period matrix of hat tau}) and $b_{\eta_{4}}, \Omega_{\eta_{4}}=\begin{pmatrix}
			{t\Omega_{\eta_{4},1}+{\Delta_{\eta_{4},1}}}&{t\Omega_{\eta_{4},0}+{\Delta_{\eta_{4},0}}}
		\end{pmatrix}^T$ are defined by (\ref{Q alpha2 condition}) and (\ref{Omega alpha2}), respectively.
	\end{thm}
	\begin{rmk}
		The error estimation is quite similar with the discussion in Section \ref{potential behavior} and the parametrix near $\pm\eta_j$ ($j=1,2,3,4$) can be described by the modified Bessel functions \cite{Grava CPAM}, which contribute the $\mathcal{O}(t^{-1})$ term in (\ref{unmodulated 2 genus solution}).
	\end{rmk}
	
	\section{The genus $\mathcal{N}$ KdV soliton gas}\label{Ngenus sector}

	\begin{figure}
		\centering
		\begin{tikzpicture}[>=latex]
			\draw[->,black,very thick] (-5,0) to (5,0) node[black,below=1mm]  {\small $x$};
			\draw[->,dashed,black,very thick] (-3,0) to (-3,5) node[black,right=1mm]  {\small $t$};
			\draw[-,black,very thick] (-3,0) to (5.0,3.2) node[black,above=0.5mm]  {\small $\xi=\xi_{crit}^{(2\mathcal{N}-2)}$};
			\draw[-,black,very thick] (-3,0) to (4.8,1.5) node[black,right=0.8mm]  {\small $\xi=\xi_{crit}^{(2\mathcal{N}-1)}$};
			\draw[-,black,very thick] (-3,0) to (-2,4.8) node[black,above=0.5mm]  {\small $\xi=\eta_{1}^2$};
			\draw[-,black,very thick] (-3,0) to (0.5,4.8) node[black,above=0.5mm]  {\small $\xi=\xi_{crit}^{(1)}$};
			\draw[-,black,very thick] (-3,0) to (3.1,4.5) node[black,above=0.1mm]  {\small $\xi=\xi_{crit}^{(2)}$} ;
			\node at (3.3,2.0) {\small Modulated};
			\node at (3.3,1.6) {\small $\mathcal{N}$ phase };
			\node at (3.5,0.9) {\small Unmodulated};
			\node at (3.5,0.5) {\small $\mathcal{N}$ phase };
			\filldraw[black] (3.0,3.6)  circle (1.2pt);
			\filldraw[black] (3.3,3.4)  circle (1.2pt);
			\filldraw[black] (3.6,3.2)  circle (1.2pt);
			\node at (1,4) {\small Unmodulated };
			\node at (1,3.5) {\small 1 phase };
			\node at (-1.2,4) {\small Modulated };
			\node at (-1.2,3.5) {\small 1 phase };
			\node at (-3.1,3.5) {\small {Quiescent} };
			\node at (-3.1,3) {\small {Region} };
			\node at (-3,-0.2) {\small 0 };
		\end{tikzpicture}
		\caption{{\protect\small
				The long-time asymptotic regions of the genus $\mathcal{N}$ KdV soliton gas potential in the $x$-$t$ half plane.}}
		\label{Ngenus pic}
	\end{figure}

	In general, when constructing the RH problem for the KdV equation, one can consider the discrete spectral points gathering in $2\mathcal{N}$ symmetric bands, where the integer $\mathcal{N} >2$. These bands are defined as {$\Sigma_+:=\cup_{j=1}^{\mathcal{N}}(\eta_{2j-1},\eta_{2j})$ and $\Sigma_-:=\cup_{j=1}^{\mathcal{N}}(-\eta_{2j},-\eta_{2j-1})$}. Consequently, the RH problem for the genus- $\mathcal{N}$ KdV soliton gas is given by:
	
	\begin{equation}\label{N genus RHP}
		X_+^{(\mathcal{N})}(\lambda)=X_-^{(\mathcal{N})}(\lambda)
		\begin{cases}
			\begin{aligned}
				&\begin{pmatrix}
					1 & -2ir_2(\lambda)e^{-2i\lambda x-8i\lambda^3t}\\
					0 & 1
				\end{pmatrix}, && \lambda\in i{\Sigma_{+}},\\
				&\begin{pmatrix}
					1 & 0\\
					2ir_2(\lambda)e^{2i\lambda x+8i\lambda^3t} & 1
				\end{pmatrix}, && \lambda\in i{\Sigma_{-}},			
			\end{aligned}
		\end{cases}
	\end{equation}
	
	\begin{equation}
		X^{(\mathcal{N})}(\lambda)\to \begin{pmatrix}
			1 & 1
		\end{pmatrix}, \quad \lambda\to\infty,
	\end{equation}
	
	\begin{equation}
		X^{(\mathcal{N})}(-\lambda)=X^{(\mathcal{N})}(\lambda)\begin{pmatrix}
			0 & 1\\
			1 & 0
		\end{pmatrix}.
	\end{equation}
	
	Similarly, the genus $\mathcal{N}$ KdV soliton gas potential can be constructed by
	
	\begin{equation}\label{Ngenus-potential-KdV}
		u(x,t)=2 \frac{\mathrm{d}}{\mathrm{d} x}\left(\lim_{\lambda \rightarrow \infty} \frac{\lambda}{i}\left(X_1^{(\mathcal{N})}(\lambda)-1\right)\right),
	\end{equation}
	where $X_1^{(\mathcal{N})}(\lambda)$ is the first component of 	$X^{(\mathcal{N})}(\lambda)$. {Especially, for a constant $r_2(\lambda)$ with $r_1r_2\neq0$, Nabelek \cite{Nabelek 2020 PhysD} showed that all algebro-geometric finite gap solutions to the KdV equation can be realized as limits of $\mathcal{N}$-soliton solutions, and explores the computation of such solutions using the primitive solution framework.}
	\par
	Following the similar procedure outlined in Section \ref{potential behavior}, the large $x$ behavior of the genus $\mathcal{N}$ KdV soliton gas potential is described by:
	
	\begin{equation}\label{Ngenus potential}
		u_{\mathcal{N}}(x)=
		\begin{cases}
			\begin{aligned}
				&-\left(2\alpha_{\mathcal{N}}+\sum_{j=1}^{2\mathcal{N}}\eta_j^2+2\partial_x^2\log
				\left(\tilde{\Theta}\left(\frac{\Omega_{\mathcal{N}}}{2\pi i};\tau_{\mathcal{N}}\right)\right)\right)+\mathcal{O}\left(\frac{1}{x}\right), && x\to+\infty,\\
				&\mathcal{O}(e^{-c|x|}), && x\to-\infty,
			\end{aligned}
		\end{cases}
	\end{equation}
	where $\alpha_{\mathcal{N}}$ is a parameter and $\tilde{\Theta}(\bullet;\tau_{\mathcal{N}})$ is the $\mathcal{N}$-phase Riemann-Theta function with
	period matrix \(\tau_{\mathcal{N}}\) and $\mathcal{N}$-dimensional column vector $\Omega_{\mathcal{N}}$.
	\par
\begin{figure}
		\centering
		\includegraphics[width=14cm]{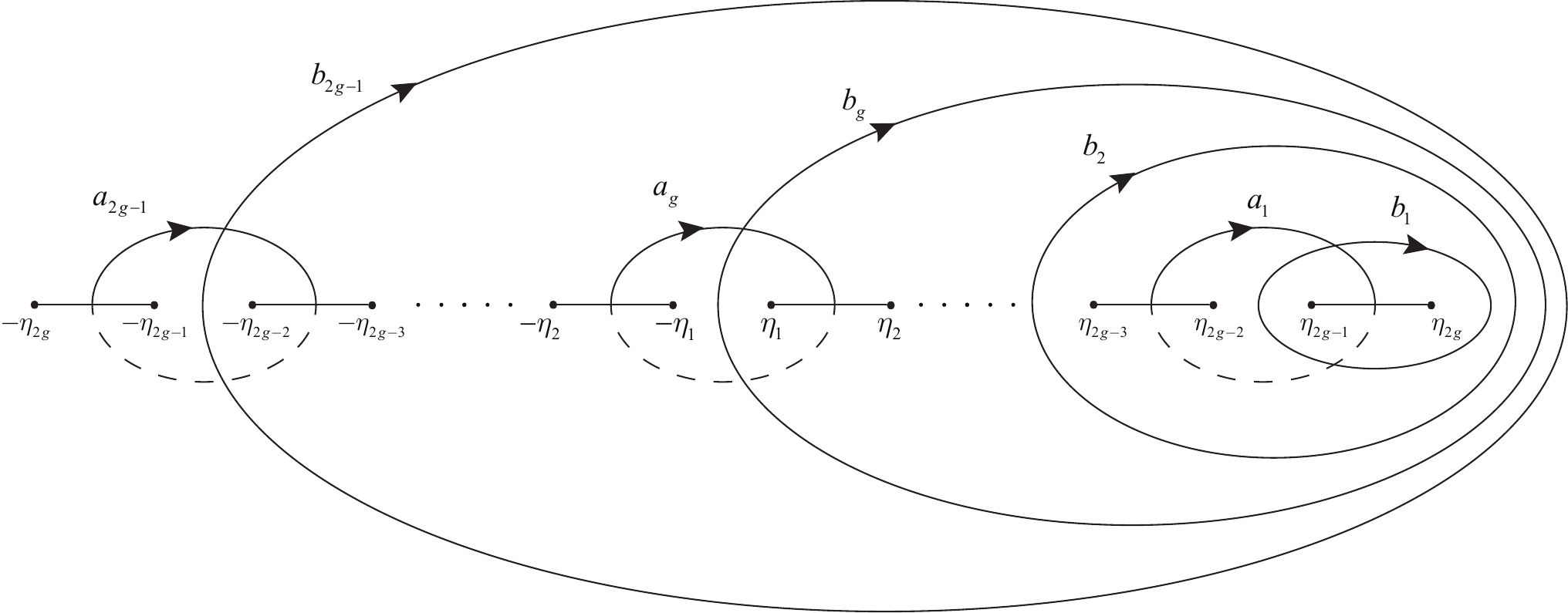}
		\caption{{\protect\small The Riemann surface $\mathcal{S}_{2g-1}$ of genus $2g-1$ and its basis $\{{a}_{j} ,{b}_{j}\}~(j=1,2,\cdots,2g-1)$ of circles.}}
		\label{circleN}
	\end{figure}	
    
{
\begin{conj}\label{conj}
As the genus $\mathcal{N}$ soliton gas potential \eqref{Ngenus-potential-KdV} evolves according to the KdV equation \eqref{KdV} for $t\to\infty$, we conjecture that the long-time asymptotic regions can be classified into $2\mathcal{N}+1$ regions in the $x$-$t$ plane. From left to right, these regions are the quiescent region, modulated one-phase wave, unmodulated one-phase wave, $\cdots$, modulated $\mathcal{N}$-phase wave, and unmodulated $\mathcal{N}$-phase wave, respectively (see Fig. \ref{Ngenus pic} for details).
\end{conj}
In particular, before discussing the general inequalities $\xi_{crit}^{(j-1)}<\xi_{crit}^{(j)},~j=1,\cdots,2\mathcal{N}-1$, we introduce the following notations. For $1\le g\le\mathcal{N}$, the  genus $2g-1$ Riemann surface $\mathcal{S}_{g}$ is defined by
\begin{equation}\label{def:Sg}
   \mathcal{S}_{g}(\eta_1,\cdots,\eta_{2g}):=\{(\lambda,\eta)|\eta^2=(\lambda^2-\eta_1^2)(\lambda^2-\eta_2^2)\cdots(\lambda^2-\eta_{2g}^2)\} 
\end{equation}}
and 
\begin{equation}\label{def:Rg}
    R_{g}(\lambda;\eta_1,\cdots,\eta_{2g}):=\sqrt{(\lambda^2-\eta_1^2)(\lambda^2-\eta_2^2)\cdots(\lambda^2-\eta_{2g}^2)}.
\end{equation}
Define $\{a_j,b_j\},~j=1,\cdots,2g-1$ as the basis of the homology $H_1(\mathcal{S}_{g})$ depicted in Figure \ref{circleN}.
Introduce the polynomials $Q_2^{(g)}(\lambda;\eta_1,\cdots,\eta_{2g}),Q_1^{(g)}(\lambda;\eta_1,\cdots,\eta_{2g})$ as 
\begin{equation}\label{def:Q1Q2}
    \begin{aligned}
        &Q_{1}^{(g)}(\lambda;\eta_1,\cdots,\eta_{2g})=\lambda^{2g}+c^{(g)}_{1}\lambda^{2g-2}+\cdots+c^{(g)}_{g},\\
&Q_{2}^{(g)}(\lambda;\eta_1,\cdots,\eta_{2g})=\lambda^{2g+2}+\beta^{(g)}\lambda^{2g}+\gamma^{(g)}_{1}\lambda^{2g-2}+\cdots+\gamma^{(g)}_{g},
    \end{aligned}
\end{equation}
where $\beta^{(g)}:=-\frac{\eta_1^2+\cdots+\eta_{2g}^2}{2}$. 
{In what follows, we write $R_{g}(\lambda)$, $Q_1^{(g)}(\lambda)$, and $Q_2^{(g)}(\lambda)$ without explicitly indicating their dependence on the parameters $\eta_1, \ldots, \eta_{2g}$ when it is clear from the context.  
We then define the quasi-momentum and quasi-energy differentials $dp^{(g)}$ and $dq^{(g)}$ as  
\begin{equation}\label{def:dpdq}
   dp^{(g)} := \frac{Q_1^{(g)}(\zeta)\, d\zeta}{R_{g}(\zeta)}, 
\qquad
dq^{(g)} := \frac{Q_2^{(g)}(\zeta)\, d\lambda}{R_{g}(\zeta)}, 
\end{equation}
Thus, the parameters $c_{k}^{(g)}$ and $\gamma_{k}^{(g)}~(k=1,\cdots,g)$ are defined by the normalized conditions
\begin{equation}\label{dpdq:normalization}
 \oint_{a_j}dp^{(g)}=\oint_{a_j}dq^{(g)}=0,~j=1,\cdots,2g-1.   
\end{equation}
Similarly, introduce the parameter $\alpha_g$ within the $g$-th band
\[
\Sigma_g := (\eta_{2g-1},\,\eta_{2g}), \qquad g = 1, \dots, \mathcal{N},
\]
which varies with $\dfrac{x}{t}$. Moreover, define
\[
S(\alpha_g) := 3\,\frac{Q_2^{(g)}(\alpha_g; \eta_1, \dots, \alpha_g)}
{Q_1^{(g)}(\alpha_g; \eta_1, \dots, \alpha_g)}, \qquad 1 \le g \le \mathcal{N}.
\]
Note that $S(\alpha_g)$ is well defined on $\Sigma_g = (\eta_{2g-1}, \eta_{2g})$ for $g = 1, \dots, \mathcal{N}$ (see Appendix~\ref{appendix}).  
Fix $g$. By the genuine nonlinearity established in~\cite{Lev88} or from the analysis in Appendix~\ref{appendix}, it follows that the function $S(\alpha_g)$ is monotone with respect to $\alpha_g$ on the interval $\eta_{2g-1} < \alpha_g < \eta_{2g}$. Indeed, similar to~(\ref{alpha2 formular}), one can establish the relationship between 
$\xi := \dfrac{x}{4t}$ and $S(\alpha_g)$ as
\begin{equation}\label{eq:xi-alpha}
    \xi = \frac{x}{4t}
   = 3\,\frac{Q_2^{(g)}(\alpha_g; \eta_1, \dots, \alpha_g)}%
           {Q_1^{(g)}(\alpha_g; \eta_1, \dots, \alpha_g)}
   = S(\alpha_g).
\end{equation}
By the genuine nonlinearity and the inverse function theorem, 
$\alpha_g$ can be regarded as a modulated parameter depending on $\dfrac{x}{t}$, see Figure \ref{Ngenus}.}
\begin{figure}[H]
		\centering
		\begin{tikzpicture}
			\node[anchor=south west, inner sep=0] (image) at (0,0) {\includegraphics[width=11cm]{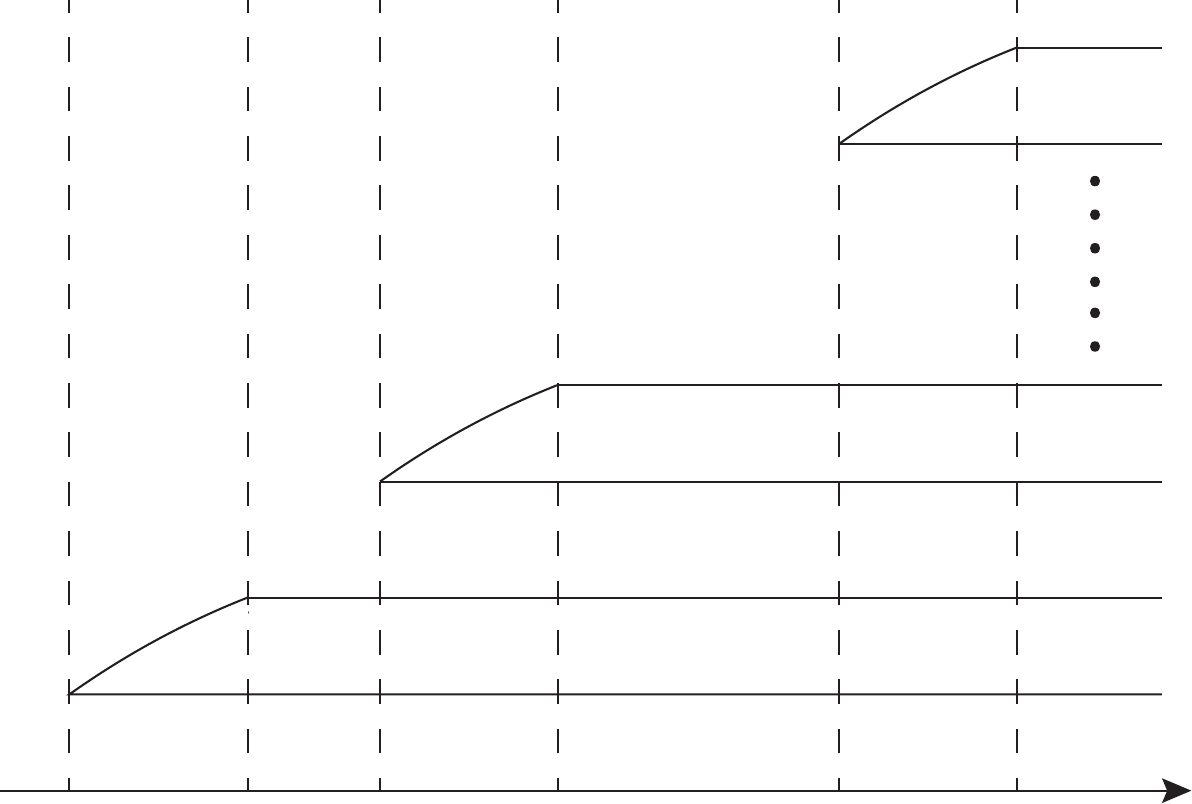}};
			\newcommand{\imgwidth}{11}
			\newcommand{\imgheight}{7}
            \node[black, font=\small, align=center, text width=4cm] at (1.01*\imgwidth, -0.03*\imgheight) {$\frac{x}{4t}$};
            \node[black, font=\small, align=center, text width=4cm] at (1.01*\imgwidth, 0.99*\imgheight) {$\eta_{2\mathcal{N}}$};
			\node[black, font=\small, align=center, text width=4cm] at (1.03*\imgwidth, 0.87*\imgheight) {$\eta_{2\mathcal{N}-1}$};
			\node[black, font=\small, align=center, text width=4cm] at (1.005*\imgwidth, 0.54*\imgheight) {$\eta_{4}$};
			\node[black, font=\small, align=center, text width=4cm] at (1.005*\imgwidth, 0.42*\imgheight) {$\eta_{3}$};
            \node[black, font=\small, align=center, text width=4cm] at (1.005*\imgwidth, 0.27*\imgheight) {$\eta_{2}$};
            \node[black, font=\small, align=center, text width=4cm] at (1.005*\imgwidth, 0.14*\imgheight) {$\eta_{1}$};
            \node[black, font=\small, align=center, text width=4cm] at (0.05*\imgwidth, -0.03*\imgheight) {$\eta_{1}^2$};
            \node[black, font=\small, align=center, text width=4cm] at (0.13*\imgwidth, 0.25*\imgheight) {$\alpha_1$};
            \node[black, font=\small, align=center, text width=4cm] at (0.2*\imgwidth, -0.03*\imgheight) {$\xi_{crit}^{(1)}$};
            \node[black, font=\small, align=center, text width=4cm] at (0.32*\imgwidth, -0.03*\imgheight) {$\xi_{crit}^{(2)}$};
             \node[black, font=\small, align=center, text width=4cm] at (0.38*\imgwidth, 0.53*\imgheight) {$\alpha_2$};
            \node[black, font=\small, align=center, text width=4cm] at (0.47*\imgwidth, -0.03*\imgheight) {$\xi_{crit}^{(3)}$};
            \node[black, font=\small, align=center, text width=4cm] at (0.7*\imgwidth, -0.03*\imgheight) {$\xi_{crit}^{(2\mathcal{N}-2)}$};
            \node[black, font=\small, align=center, text width=4cm] at (0.78*\imgwidth, 0.98*\imgheight) {$\alpha_{\mathcal{N}}$};
            \node[black, font=\small, align=center, text width=4cm] at (0.88*\imgwidth, -0.03*\imgheight) {$\xi_{crit}^{(2\mathcal{N}-1)}$};
		\end{tikzpicture}
		\caption{\small The figure schematically illustrates the evolution of the Riemann invariants $\alpha_g$, $1 \le g \le \mathcal{N}$, for a genus-$\mathcal{N}$ KdV soliton gas. In particular, $\alpha_g\left(\tfrac{x}{t}\right)$ is modulated within the interval $(\xi_{\mathrm{crit}}^{(2g-2)}, \xi_{\mathrm{crit}}^{(2g-1)})$, varying from $\eta_{2g-1}$ to $\eta_{2g}$. Moreover, the figure also depicts the transition of the wave (or “gas”) from the zero state to the genus-$\mathcal{N}$ configuration.}
		\label{Ngenus}
	\end{figure}
{Now, for $1\le g\le \mathcal{N}$, as $\alpha_g$ approaches the critical values $\eta_{2g-1}$ and $\eta_{2g}$, 
we define the corresponding critical quantities 
$\xi_{\mathrm{crit}}^{(2g-2)}$ and $\xi_{\mathrm{crit}}^{(2g-1)}$ by
\[
\xi_{\mathrm{crit}}^{(2g-2)}
 :=3\lim_{\epsilon \to 0^+} S(\eta_{2g-1}+\epsilon)=3\lim_{\epsilon \to 0^+}
   \frac{Q_2^{(g)}(\eta_{2g-1} + \epsilon; \eta_1, \dots, \eta_{2j-1}, \eta_{2g-1} + \epsilon)}
        {Q_1^{(g)}(\eta_{2g-1} + \epsilon; \eta_1, \dots, \eta_{2j-1}, \eta_{2g-1} + \epsilon)},
\]
and similarly
\[
\xi_{\mathrm{crit}}^{(2g-1)}
 :=3\lim_{\epsilon \to 0^+} S(\eta_{2g}-\epsilon)= 3\lim_{\epsilon \to 0^+}
   \frac{Q_2^{(g)}(\eta_{2g} - \epsilon; \eta_1, \dots, \eta_{2g-1}, \eta_{2g} - \epsilon)}
        {Q_1^{(g)}(\eta_{2g} - \epsilon; \eta_1, \dots, \eta_{2g-1}, \eta_{2g} - \epsilon)}.
\]
The existence of the limit $S(\eta_{2g-1}+\epsilon)$, as $\epsilon\to0^+$, see Appendix \ref{appendix}. Regarding to the the limit $S(\eta_{2g}-\epsilon)$, the associated algebraic curve $\mathcal{S}(\eta_1,\cdots,\eta_{2g-1},\eta_{2g-1}+\epsilon)$ develops two nodes, and then the genus of the associated Riemann surface decreases by 2. By the formula in \cite{Fay}, for $2\le j\le \mathcal{N}$ we have
$$
\xi_{crit}^{(2g-2)}=3\frac{\int_{-R_{g-1}(\eta_{g+1})}^{R_{g-1}(\eta_{g+1})}dq^{(g-1)}}{\int_{-R_{g-1}(\eta_{g+1})}^{R_{g-1}(\eta_{g+1})}dp^{(g-1)}}.
$$
Furthermore, by the genuine nonlinearity, it follows that $\xi_{crit}^{(2g-1)}>\xi_{crit}^{(2g-2)}$. However, it is a challenge to show that $\xi_{crit}^{(2g-2)}>\xi_{crit}^{(2g-3)}$ for $2\le g\le\mathcal{N}$. Especially, for $g=2$, the analysis has been established previously, and we conjecture that for the general $\mathcal{N}$ case, the inequalities still hold. 
}

{
Consequently, according to the relationship in~(\ref{eq:xi-alpha}), as $t \to \infty$, 
the critical values 
$\xi_{\mathrm{crit}}^{(j)}$, $0 \le j \le 2\mathcal{N}-1$ 
(with $\xi_{\mathrm{crit}}^{(0)} = \eta_1^2$), 
divide the $(x,t)$-plane into $2\mathcal{N}+1$ distinct regions; 
see Figure~\ref{Ngenus pic}.
}



\appendix
\section{Appendix}\label{appendix}
{Recall the definition of the genus-$2g-1$ Riemann surface $\mathcal{S}_{g}(\eta_1,\ldots,\eta_{2g})$ in~(\ref{def:Sg}) 
and the corresponding function $R_g(\lambda;\eta_1,\ldots,\eta_{2g})$ in~(\ref{def:Rg}), 
whose branch cuts (or bands) lie in $\pm(\eta_{2k-1},\eta_{2k})$ for $k=1,\ldots,g$. 
In addition, recall the quasi-momentum $dp^{(g)}$ and quasi-energy $dq^{(g)}$ defined in~(\ref{def:dpdq}), 
together with the corresponding polynomials $Q^{(g)}_1$ and $Q^{(g)}_2$ in~(\ref{def:Q1Q2}). 
Note that $dp^{(g)}$ and $dq^{(g)}$ are normalized Abelian differentials on the Riemann surface $\mathcal{S}_{g}(\eta_1,\ldots,\eta_{2g})$. 
}

{
As the solution evolves in time, a parameter $\alpha_g$ is modulated by the ratio $\frac{x}{t}$ and takes values in the interval $(\eta_{2g-1},\eta_{2g})$. Recalling the identity \eqref{eq:xi-alpha}, it follows that
\begin{equation}\label{eq:xi-general}
    \xi
    = 3\,\frac{Q_2^{(g)}(\lambda;\eta_1,\ldots,\alpha_g)}
    {Q_1^{(g)}(\lambda;\eta_1,\ldots,\alpha_g)},
\end{equation}
where the parameters $\eta_j$, $j=1,\ldots,2g-1$, are fixed and $\eta_{2g}$ is replaced by $\alpha_g$. Proceeding in the same manner as in the genus--$2$ case, we obtain the following result.
}
{\begin{lem}
Suppose that $0<\eta_1<\cdots<\eta_{2g}$, and let $dp^{(g)}$ and $dq^{(g)}$ be defined as in \eqref{def:dpdq}. For $\alpha_g\in(\eta_{2g-1},\eta_{2g})$ related to $\xi=\frac{x}{4t}$ by \eqref{eq:xi-general}, the parameter $\alpha_g$ is a monotonically increasing function of $\xi$.
\end{lem}
}
\begin{proof}
{
Introduce the linear combination of $dp^{(g)}$ and $dq^{(g)}$ defined by
\[
    d\varphi := 12\,dq^{(g)} - 4\xi\,dp^{(g)},
\]
where
\[
    dp^{(g)} := \frac{Q_1^{(g)}(\lambda;\eta_1,\ldots,\eta_{2g-1},\alpha_g)\,d\lambda}
    {R_g(\lambda;\eta_1,\ldots,\eta_{2g-1},\alpha_g)}, 
    \qquad
    dq^{(g)} := \frac{Q_2^{(g)}(\lambda;\eta_1,\ldots,\eta_{2g-1},\alpha_g)\,d\lambda}
    {R_g(\lambda;\eta_1,\ldots,\eta_{2g-1},\alpha_g)}.
\]
Since $dp^{(g)}$ and $dq^{(g)}$ are normalized Abelian differentials of the second kind and $\alpha$ is a soft edge which indicates that $\alpha$ is zero of $d\varphi$, and by the definition of $Q_1^{(g)}$ and $Q_2^{(g)}$ in \eqref{def:Q1Q2}, we obtain
\begin{align*}
    d\varphi
    &= dq^{(g)} - 4\xi\,dp^{(g)} 
    = 12\,
    \frac{\prod_{j=1}^{g}(\lambda^2-\lambda_j^2)(\lambda^2-\alpha_g^2)}
    {R(\lambda;\eta_1^2,\ldots,\eta_{2g-1}^2,\alpha_g^2)}\,d\lambda 
    = 12\,
    \frac{\prod_{j=1}^{g}(\lambda^2-\lambda_j^2)\sqrt{\lambda^2-\alpha_g^2}}
    {\sqrt{\prod_{k=1}^{2g-1}(\lambda^2-\eta_k^2)}}\,d\lambda,
\end{align*}
where $\lambda_j\in(\eta_{2j-2},\eta_{2j-1})$ for $j=1,\ldots,g$, and we set $\eta_0=0$.
}
{
Taking the derivative of $d\varphi$ with respect to $\xi$, we have
\[
    \partial_{\xi}d\varphi
    = -4\,dp^{(g)} 
    + \partial_{\alpha_g}\!\left(dq^{(g)}-4\xi\,dp^{(g)}\right)\partial_{\xi}\alpha_g.
\]
A direct computation shows that
\[
    \partial_{\alpha_g}\!\left(dq^{(g)}-4\xi\,dp^{(g)}\right)
    = -12\,\frac{\alpha_g\prod_{j=1}^{g}(\lambda^2-\lambda_j^2)}
    {R(\lambda;\eta_1^2,\ldots,\eta_{2g-1}^2,\alpha_g^2)}\,d\lambda,
\]
which has no singularities at $\lambda=\pm\alpha_g$ or at $\lambda=\pm\infty$. By the Riemann bilinear relations, it follows that
\(
    \partial_{\alpha_g}\!\left(dq^{(g)}-4\xi\,dp^{(g)}\right)\equiv 0.
\)
}

{
On the other hand, differentiating the explicit expression of $d\varphi$ with respect to $\xi$ yields
\[
    \partial_{\xi}d\varphi
    = -\left(
    \sum_{j=1}^{g}\frac{2\lambda_j\,\partial_{\xi}\lambda_j}{\lambda^2-\lambda_j^2}
    + \frac{\alpha_g\,\partial_{\xi}\alpha_g}{\lambda^2-\alpha_g^2}
    \right)d\varphi.
\]
Comparing the two expressions for $\partial_{\xi}d\varphi$, we obtain
\[
    4\,\frac{dp^{(g)}}{d\varphi}
    = \frac{1}{3}\,
    \frac{Q_1(\lambda)}
    {\prod_{j=1}^{g}(\lambda^2-\lambda_j^2)(\lambda^2-\alpha_g^2)}
    = \sum_{j=1}^{g}\frac{2\lambda_j\,\partial_{\xi}\lambda_j}{\lambda^2-\lambda_j^2}
    + \frac{\alpha_g\,\partial_{\xi}\alpha_g}{\lambda^2-\alpha_g^2}.
\]
Since $Q_1(\lambda)\sim \lambda^{2g}$ as $\lambda\to\infty$ and all zeros of $Q_1(\lambda)$ lie in the intervals $(\eta_{2j-2},\eta_{2j-1})$, $j=1,\ldots,2g$, we have $Q_1(\alpha_g)>0$. By taking the residue at $\lambda=\alpha_g$, we conclude that
\[
    \partial_{\xi}\alpha_g
    = \frac{Q_1(\alpha_g)}
    {3\alpha_g\prod_{j=1}^{g}(\alpha_g^2-\lambda_j^2)} > 0,
\]
which proves that $\alpha_g$ is a monotonically increasing function of $\xi$.}
\end{proof}

\noindent{\bf Conflict of interest declaration.} We declare we have no competing interests.

\noindent{\bf Financial Support.}
Support is acknowledged from the National Natural Science Foundation of China, Grant No. 12371247 and No. 12431008 and Beijing Natural Science Foundation Grant No. 1262012. {The last author acknowledge the support of the scholarship provided by the China Scholarship Council (CSC) under Grant No. 202406040149 and the GNFM-INDAM group and the research project Mathematical Methods in NonLinear Physics (MMNLP), Gruppo 4-Fisica Teorica of INFN.} 

\section*{Acknowledgments}
The authors extend their heartfelt appreciation to Fudong Wang and Peng Zhao for their invaluable contributions
to this project through stimulating conversations and enlightening discussions. The authors thank Peng Yan for the
inspiration regarding the topic of this paper. The authors also wish to express sincere gratitude to Professor Tamara
Grava for her insightful feedback and thoughtful discussions. Finally, the insightful suggestions from the reviewers
are gratefully acknowledged.

	\bibliographystyle{amsplain}

\end{document}